\newif\ifdraft \draftfalse

\documentclass[runningheads]{llncs}
\usepackage{url}
\usepackage{xspace}
\usepackage{amsmath}
\usepackage{amssymb}
\usepackage{proof}
\usepackage{monadicLens}
\usepackage{graphicx}
\usepackage{br}

\usepackage[sectionbib]{natbib}

\makeatletter
\@ifundefined{lhs2tex.lhs2tex.sty.read}%
  {\@namedef{lhs2tex.lhs2tex.sty.read}{}%
   \newcommand\SkipToFmtEnd{}%
   \newcommand\EndFmtInput{}%
   \long\def\SkipToFmtEnd#1\EndFmtInput{}%
  }\SkipToFmtEnd

\newcommand\ReadOnlyOnce[1]{\@ifundefined{#1}{\@namedef{#1}{}}\SkipToFmtEnd}
\usepackage{amstext}
\usepackage{amssymb}
\usepackage{stmaryrd}
\DeclareFontFamily{OT1}{cmtex}{}
\DeclareFontShape{OT1}{cmtex}{m}{n}
  {<5><6><7><8>cmtex8
   <9>cmtex9
   <10><10.95><12><14.4><17.28><20.74><24.88>cmtex10}{}
\DeclareFontShape{OT1}{cmtex}{m}{it}
  {<-> ssub * cmtt/m/it}{}

\DeclareFontShape{OT1}{cmtt}{bx}{n}
  {<5><6><7><8>cmtt8
   <9>cmbtt9
   <10><10.95><12><14.4><17.28><20.74><24.88>cmbtt10}{}
\DeclareFontShape{OT1}{cmtex}{bx}{n}
  {<-> ssub * cmtt/bx/n}{}

\newcommand{\Conid}[1]{\mathit{#1}}
\newcommand{\Varid}[1]{\mathit{#1}}
\newcommand{\anonymous}{\kern0.06em \vbox{\hrule\@width.5em}}

\newcommand{\bind}{\mathbin{>\!\!\!>\mkern-6.7mu=}}

\usepackage{polytable}

\@ifundefined{mathindent}%
  {\newdimen\mathindent\mathindent\leftmargini}%
  {}%

\def\resethooks{%
  \global\let\SaveRestoreHook\empty
  \global\let\ColumnHook\empty}
\newcommand*{\savecolumns}[1][default]%
  {\g@addto@macro\SaveRestoreHook{\savecolumns[#1]}}
\newcommand*{\restorecolumns}[1][default]%
  {\g@addto@macro\SaveRestoreHook{\restorecolumns[#1]}}
\newcommand*{\aligncolumn}[2]%
  {\g@addto@macro\ColumnHook{\column{#1}{#2}}}

\resethooks

\newcommand{\onelinecommentchars}{\quad-{}- }
\newcommand{\commentbeginchars}{\enskip\{-}
\newcommand{\commentendchars}{-\}\enskip}

\newcommand{\visiblecomments}{%
  \let\onelinecomment=\onelinecommentchars
  \let\commentbegin=\commentbeginchars
  \let\commentend=\commentendchars}

\newcommand{\invisiblecomments}{%
  \let\onelinecomment=\empty
  \let\commentbegin=\empty
  \let\commentend=\empty}

\visiblecomments

\newlength{\blanklineskip}
\newcommand{\hsindent}[1]{\quad}
\let\hspre\empty
\let\hspost\empty

\EndFmtInput
\makeatother
\ReadOnlyOnce{polycode.fmt}%
\makeatletter

\newcommand{\hsnewpar}[1]%
  {{\parskip=0pt\parindent=0pt\par\vskip #1\noindent}}

\newcommand{\hscodestyle}{}

\newcommand{\sethscode}[1]%
  {\expandafter\let\expandafter\hscode\csname #1\endcsname
   \expandafter\let\expandafter\endhscode\csname end#1\endcsname}

  {\par\noindent
   \advance\leftskip\mathindent
   \hscodestyle
   \let\\=\@normalcr
   \let\hspre\(\let\hspost\)%
   \pboxed}%
  {\endpboxed\)%
   \par\noindent
   \ignorespacesafterend}

  {\hsnewpar\abovedisplayskip
   \advance\leftskip\mathindent
   \hscodestyle
   \let\hspre\(\let\hspost\)%
   \pboxed}%
  {\endpboxed%
   \hsnewpar\belowdisplayskip
   \ignorespacesafterend}

  {\hsnewpar\abovedisplayskip
   \advance\leftskip\mathindent
   \hscodestyle
   \let\\=\@normalcr
   \(\pboxed}%
  {\endpboxed\)%
   \hsnewpar\belowdisplayskip
   \ignorespacesafterend}

\newcommand{\plainhs}{\sethscode{plainhscode}}

\plainhs

  {\hsnewpar\abovedisplayskip
   \advance\leftskip\mathindent
   \hscodestyle
   \let\\=\@normalcr
   \(\parray}%
  {\endparray\)%
   \hsnewpar\belowdisplayskip
   \ignorespacesafterend}

  {\parray}{\endparray}

  {\(\parray}{\endparray\)}

\def\codeframewidth{\arrayrulewidth}
\RequirePackage{calc}

  {\parskip=\abovedisplayskip\par\noindent
   \hscodestyle
   \arrayrulewidth=\codeframewidth
   \tabular{@{}|p{\linewidth-2\arraycolsep-2\arrayrulewidth-2pt}|@{}}%
   \hline\framedhslinecorrect\\{-1.5ex}%
   \let\endoflinesave=\\
   \let\\=\@normalcr
   \(\pboxed}%
  {\endpboxed\)%
   \framedhslinecorrect\endoflinesave{.5ex}\hline
   \endtabular
   \parskip=\belowdisplayskip\par\noindent
   \ignorespacesafterend}

\newcommand{\framedhslinecorrect}[2]%
  {#1[#2]}

  {\(\def\column##1##2{}%
   \let\>\undefined\let\<\undefined\let\\\undefined
   \newcommand\>[1][]{}\newcommand\<[1][]{}\newcommand\\[1][]{}%
   \def\fromto##1##2##3{##3}%
   }{\) }%

  {\let\orighscode=\hscode
   \let\origendhscode=\endhscode
   \def\endhscode{\def\hscode{\endgroup\def\@currenvir{hscode}\\}\begingroup}
   \orighscode\def\hscode{\endgroup\def\@currenvir{hscode}}}%
  {\origendhscode
   \global\let\hscode=\orighscode
   \global\let\endhscode=\origendhscode}%

\makeatother
\EndFmtInput
\ReadOnlyOnce{forall.fmt}%
\makeatletter

\let\HaskellResetHook\empty
\newcommand*{\AtHaskellReset}[1]{%
  \g@addto@macro\HaskellResetHook{#1}}
\newcommand*{\HaskellReset}{\HaskellResetHook}

\newcommand\hsforall{\global\let\hsdot=\hsperiodonce}
\newcommand*\hsperiodonce[2]{#2\global\let\hsdot=\hscompose}
\newcommand*\hscompose[2]{#1}

\AtHaskellReset{\global\let\hsdot=\hscompose}

\HaskellReset

\makeatother
\EndFmtInput

\def\get#1{\Varid{get}_{\mskip-2mu#1}}
\def\set#1{\Varid{set}_{\mskip-2mu#1}}

\def\commentbegin{\quad$[\![\enskip$}
\def\commentend{$\enskip]\!]$}

\title{Reflections on monadic lenses}
\def\oxford{\inst{1}}
\def\edinburgh{\inst{2}}

\author{Faris~Abou-Saleh\oxford \and 
  James~Cheney\edinburgh \and 
  Jeremy~Gibbons\oxford \and 
  James~McKinna\edinburgh \and 
  Perdita~Stevens\edinburgh}
\authorrunning{F.\,Abou-Saleh,
  J.\,Cheney,
  J.\,Gibbons,
  J.\,McKinna,
  P.\,Stevens}
\institute{University of Oxford, \email{firstname.lastname@cs.ox.ac.uk}
\and University of Edinburgh, \email{firstname.lastname@ed.ac.uk}
}

\begin{document}
\maketitle

\begin{abstract}
  Bidirectional transformations (bx) have primarily been modeled as
  pure functions, and do not account for the possibility of the
  side-effects  that are available in most programming languages. Recently
  several formulations of bx that use monads to account for effects
  have been proposed, both among practitioners and in academic
  research. The combination of bx with effects turns out to be
  surprisingly subtle, leading to problems with some of these
  proposals and increasing the complexity of others. This paper
  reviews the proposals for monadic lenses to date,
  and offers some improved definitions, paying particular attention to
  the obstacles to naively adding monadic effects to existing
  definitions of pure bx such as lenses and symmetric lenses, and the
  subtleties of equivalence of symmetric bidirectional transformations
  in the presence of effects.
\end{abstract}

\section{Introduction}

Programming with multiple concrete representations of the same conceptual
information is a commonplace, and challenging, problem. It is commonplace
because data is everywhere, and not all of it is relevant or appropriate
for every task: for example, one may want to work with only a subset of
one's full email account on a mobile phone or other low-bandwidth device.
It is challenging because the most direct approach to mapping data across
sources \ensuremath{\Conid{A}} and \ensuremath{\Conid{B}} is to write separate functions,
one mapping to \ensuremath{\Conid{B}} and one to \ensuremath{\Conid{A}},
following some (not always explicit) specification of what it
means for an \ensuremath{\Conid{A}} value and a \ensuremath{\Conid{B}} value to be \emph{consistent}. Keeping
these transformations coherent with each other, and with the specification,
is a considerable maintenance burden, yet it remains the main approach
found in practice.

Over the past decade, a number of promising proposals to ease
programming such \emph{bidirectional transformations} have emerged,
including \emph{lenses}~\citep{lens-toplas}, bx based on consistency
relations~\citep{stevens09:sosym}, \emph{symmetric
  lenses}~\citep{symlens}, and a number of variants and
extensions (e.g.~\citep{pacheco14pepm,johnson14bx}).  
Most of these proposals consist of an interface
with pure functions and some equational laws that characterise
good behaviour; the interaction of bidirectionality with
other effects has received comparatively little attention.  


Some programmers and researchers have already proposed ways to combine
lenses and monadic effects~\citep{reddit,pacheco14pepm}.  Recently, we
have proposed symmetric notions of bidirectional computation based on
\emph{entangled state monads}~\citep{cheney14bx2,abousaleh15mpc} and
\emph{coalgebras}~\citep{abousaleh15bx}.  As a result, there are now
several alternative proposals for bidirectional transformations with
effects.  While this diversity is natural and healthy, reflecting an
active research area, the different proposals tend to employ somewhat
different terminology, and the relationships among them are not well
understood.  Small differences in definitions can have
disproportionate impact.  

In this paper we summarise and compare the existing proposals, offer
some new alternatives, and attempt to provide general and useful
definitions of ``monadic lenses'' and ``symmetric monadic lenses''.
Perhaps surprisingly, it appears challenging even to define the
composition of lenses in the presence of effects, especially in the
symmetric case.  We first review the definition of pure asymmetric
lenses and two
prior
proposals for extending them with monadic effects. These
definitions have some limitations, and we propose a new definition of
monadic lens that overcomes them.  

Next we consider the symmetric case. The effectful bx and coalgebraic
bx in our previous work are symmetric, but their definitions rely on
relatively heavyweight machinery (monad transformers and morphisms,
coalgebra).  It seems natural to ask whether just adding monadic
effects to symmetric lenses in the style of Hofmann et al.~\citep{symlens} would also
work.  We show that, as for asymmetric lenses, adding monadic effects 
to symmetric lenses is challenging, and give examples illustrating
the problems with the most obvious generalisation.
We then briefly discuss our recent work on symmetric forms of bx with
monadic effects~\citep{cheney14bx2,abousaleh15mpc,abousaleh15bx}.
Defining composition for these approaches also turns out to be tricky,
and our definition of monadic lenses arose out of exploring this
space.  The essence of composition of symmetric monadic bx, we now
believe, can be presented most easily in terms of monadic lenses, by
considering \emph{spans}, an approach also advocated (in the pure
case) by \citet{johnson14bx}.

Symmetric pure bx need to be equipped with a notion of equivalence, to
abstract away inessential differences of representation of their
``state'' or ``complement'' spaces.  As noted by \citet{symlens} and \citet{johnson14bx},
isomorphism of state spaces is unsatisfactory, and there are competing
proposals for equivalence of symmetric lenses and spans.  In the case
of spans of monadic lenses, the right notion of equivalence seem even
less obvious.  We compare three, increasingly coarse, equivalences
of spans based on isomorphism (following~\citet{abousaleh15mpc}), span
equivalence (following~\citet{johnson14bx}), and bisimulation
(following~\citet{symlens} and \citet{abousaleh15bx}).  In addition, we show a (we
think surprising) result: in the pure case, span equivalence and
bisimulation equivalence coincide.

In this paper we employ Haskell-like notation to describe and compare
formalisms, with a few conventions: we write function composition \ensuremath{\Varid{f}\hsdot{\mathbin{\cdot}}{\mathrel{.}}\Varid{g}} with a centred dot, and use a lowered dot for field lookup
\ensuremath{\Varid{x}\mathord{.}\Varid{f}}, in contrast to Haskell's notation \ensuremath{\Varid{f}\;\Varid{x}}.  Throughout the
paper, we introduce a number of different representations of lenses,
and rather than pedantically disambiguating them all, we freely
redefine identifiers as we go.  We assume familiarity with common uses
of monads in Haskell to encapsulate
effects (following \citet{DBLP:conf/afp/Wadler95}), and with the
\ensuremath{\mathbf{do}}-notation (following Wadler's
monad comprehensions~\citet{DBLP:journals/mscs/Wadler92}).
Although some of these ideas are present or implicit in recent
papers~\citep{symlens,johnson14bx,cheney14bx2,abousaleh15mpc,abousaleh15bx},
this paper reflects our desire to clarify these ideas and expose them
in their clearest form --- a desire that is strongly influenced by
Wadler's work on a wide variety of related
topics~\citep{DBLP:journals/mscs/Wadler92,DBLP:conf/fp/KingW92,DBLP:conf/afp/Wadler95},
and by our interactions with him as a colleague.

\section{Asymmetric monadic lenses}\label{sec:asymmetric}

Recall that a \emph{lens}~\citep{lens-toplas,foster10ssgip} is a pair of functions, usually called
\emph{get} and \emph{put}:
\begin{hscode}\SaveRestoreHook
\column{B}{@{}>{\hspre}l<{\hspost}@{}}%
\column{28}{@{}>{\hspre}l<{\hspost}@{}}%
\column{E}{@{}>{\hspre}l<{\hspost}@{}}%
\>[B]{}\mathbf{data}\;\alpha\mathbin{\leadsto}\beta\mathrel{=}\Conid{Lens}\;\{\mskip1.5mu {}\<[28]%
\>[28]{}\Varid{get}\mathbin{::}\alpha\to \beta,\Varid{put}\mathbin{::}\alpha\to \beta\to \alpha\mskip1.5mu\}{}\<[E]%
\ColumnHook
\end{hscode}\resethooks
satisfying (at least) the following \emph{well-behavedness} laws:
\begin{hscode}\SaveRestoreHook
\column{B}{@{}>{\hspre}l<{\hspost}@{}}%
\column{3}{@{}>{\hspre}l<{\hspost}@{}}%
\column{12}{@{}>{\hspre}c<{\hspost}@{}}%
\column{12E}{@{}l@{}}%
\column{16}{@{}>{\hspre}l<{\hspost}@{}}%
\column{31}{@{}>{\hspre}l<{\hspost}@{}}%
\column{E}{@{}>{\hspre}l<{\hspost}@{}}%
\>[3]{}\mathsf{(GetPut)}{}\<[12]%
\>[12]{}\quad{}\<[12E]%
\>[16]{}\Varid{put}\;\Varid{a}\;(\Varid{get}\;\Varid{a}){}\<[31]%
\>[31]{}\mathrel{=}\Varid{a}{}\<[E]%
\\
\>[3]{}\mathsf{(PutGet)}{}\<[12]%
\>[12]{}\quad{}\<[12E]%
\>[16]{}\Varid{get}\;(\Varid{put}\;\Varid{a}\;\Varid{b}){}\<[31]%
\>[31]{}\mathrel{=}\Varid{b}{}\<[E]%
\ColumnHook
\end{hscode}\resethooks
The idea is that a lens of type \ensuremath{\Conid{A}\mathbin{\leadsto}\Conid{B}} maintains a source of type
\ensuremath{\Conid{A}}, providing a view of type \ensuremath{\Conid{B}} onto it; the well-behavedness laws
capture the intuition that the view faithfully reflects the source: if
we ``get'' a \ensuremath{\Varid{b}} from a source \ensuremath{\Varid{a}} and then ``put'' the same \ensuremath{\Varid{b}} value
back into \ensuremath{\Varid{a}}, this leaves \ensuremath{\Varid{a}} unchanged; and if we ``put'' a \ensuremath{\Varid{b}} into
a source \ensuremath{\Varid{a}} and then ``get'' from the result, we get \ensuremath{\Varid{b}} itself.
Lenses are often
equipped with a \ensuremath{\Varid{create}} function
\begin{hscode}\SaveRestoreHook
\column{B}{@{}>{\hspre}l<{\hspost}@{}}%
\column{28}{@{}>{\hspre}l<{\hspost}@{}}%
\column{69}{@{}>{\hspre}l<{\hspost}@{}}%
\column{E}{@{}>{\hspre}l<{\hspost}@{}}%
\>[B]{}\mathbf{data}\;\alpha\mathbin{\leadsto}\beta\mathrel{=}\Conid{Lens}\;\{\mskip1.5mu {}\<[28]%
\>[28]{}\Varid{get}\mathbin{::}\alpha\to \beta,\Varid{put}\mathbin{::}\alpha\to \beta\to \alpha,{}\<[69]%
\>[69]{}\Varid{create}\mathbin{::}\beta\to \alpha\mskip1.5mu\}{}\<[E]%
\ColumnHook
\end{hscode}\resethooks
 satisfying an additional law:
\begin{hscode}\SaveRestoreHook
\column{B}{@{}>{\hspre}l<{\hspost}@{}}%
\column{5}{@{}>{\hspre}l<{\hspost}@{}}%
\column{14}{@{}>{\hspre}c<{\hspost}@{}}%
\column{14E}{@{}l@{}}%
\column{18}{@{}>{\hspre}l<{\hspost}@{}}%
\column{34}{@{}>{\hspre}l<{\hspost}@{}}%
\column{E}{@{}>{\hspre}l<{\hspost}@{}}%
\>[5]{}\mathsf{(CreateGet)}{}\<[14]%
\>[14]{}\quad{}\<[14E]%
\>[18]{}\Varid{get}\;(\Varid{create}\;\Varid{b}){}\<[34]%
\>[34]{}\mathrel{=}\Varid{b}{}\<[E]%
\ColumnHook
\end{hscode}\resethooks
When the distinction is important, we use the term \emph{full} for
well-behaved lenses equipped with a \ensuremath{\Varid{create}} operation.  It is easy to show that
the source and view types of a full lens must either both be empty or
both non-empty, and that the \ensuremath{\Varid{get}} operation of a full lens is
surjective.

Lenses have been investigated extensively; see for example \citet{foster10ssgip} for a recent tutorial overview.  For the
purposes of this paper, we just recall the definition of
\emph{composition} of lenses:
\begin{hscode}\SaveRestoreHook
\column{B}{@{}>{\hspre}l<{\hspost}@{}}%
\column{3}{@{}>{\hspre}l<{\hspost}@{}}%
\column{28}{@{}>{\hspre}l<{\hspost}@{}}%
\column{E}{@{}>{\hspre}l<{\hspost}@{}}%
\>[3]{}(\mathbin{;})\mathbin{::}(\alpha\mathbin{\leadsto}\beta)\to (\beta\mathbin{\leadsto}\gamma)\to (\alpha\mathbin{\leadsto}\gamma){}\<[E]%
\\
\>[3]{}\Varid{l}_{1}\mathbin{;}\Varid{l}_{2}\mathrel{=}\Conid{Lens}\;{}\<[28]%
\>[28]{}(\Varid{l}_{2}\mathord{.}\Varid{get}\hsdot{\mathbin{\cdot}}{\mathrel{.}}\Varid{l}_{1}\mathord{.}\Varid{get})\;{}\<[E]%
\\
\>[28]{}(\lambda \Varid{a}\;\Varid{c}\to \Varid{l}_{1}\mathord{.}\Varid{put}\;\Varid{a}\;(\Varid{l}_{2}\mathord{.}\Varid{put}\;(\Varid{l}_{1}\mathord{.}\Varid{get}\;\Varid{a})\;\Varid{c}))\;{}\<[E]%
\\
\>[28]{}(\Varid{l}_{1}\mathord{.}\Varid{create}\hsdot{\mathbin{\cdot}}{\mathrel{.}}\Varid{l}_{1}\mathord{.}\Varid{create}){}\<[E]%
\ColumnHook
\end{hscode}\resethooks
which preserves well-behavedness.  


\subsection{A naive approach}\label{sec:naive-lens}

As a first attempt, consider simply adding a monadic effect \ensuremath{\mu} to
the result types of both \ensuremath{\Varid{get}} and \ensuremath{\Varid{put}}.  
\begin{hscode}\SaveRestoreHook
\column{B}{@{}>{\hspre}l<{\hspost}@{}}%
\column{35}{@{}>{\hspre}l<{\hspost}@{}}%
\column{E}{@{}>{\hspre}l<{\hspost}@{}}%
\>[B]{}\mathbf{data}\;\monadic{\alpha\rightsquigarrow_0\beta}{\mu}\mathrel{=}MLens_0\;\{\mskip1.5mu {}\<[35]%
\>[35]{}\Varid{mget}\mathbin{::}\alpha\to \mu\;\beta,\Varid{mput}\mathbin{::}\alpha\to \beta\to \mu\;\alpha\mskip1.5mu\}{}\<[E]%
\ColumnHook
\end{hscode}\resethooks
Such an approach has been considered and discussed in some recent
Haskell libraries and online discussions~\citep{reddit}.  A natural
question arises immediately: what laws should a lens \ensuremath{\Varid{l}\mathbin{::}\monadic{\Conid{A}\rightsquigarrow_0\Conid{B}}{\Conid{M}}} satisfy?  The following generalisations of the laws appear natural:
\begin{hscode}\SaveRestoreHook
\column{B}{@{}>{\hspre}l<{\hspost}@{}}%
\column{3}{@{}>{\hspre}l<{\hspost}@{}}%
\column{12}{@{}>{\hspre}l<{\hspost}@{}}%
\column{46}{@{}>{\hspre}l<{\hspost}@{}}%
\column{E}{@{}>{\hspre}l<{\hspost}@{}}%
\>[3]{}\mathsf{(MGetPut_0)}{}\<[12]%
\>[12]{}\quad\mathbf{do}\;\{\mskip1.5mu \Varid{b}\leftarrow \Varid{mget}\;\Varid{a};\Varid{mput}\;\Varid{a}\;\Varid{b}\mskip1.5mu\}{}\<[46]%
\>[46]{}\mathrel{=}\Varid{return}\;\Varid{a}{}\<[E]%
\\
\>[3]{}\mathsf{(MPutGet_0)}{}\<[12]%
\>[12]{}\quad\mathbf{do}\;\{\mskip1.5mu \Varid{a'}\leftarrow \Varid{mput}\;\Varid{a}\;\Varid{b};\Varid{mget}\;\Varid{a'}\mskip1.5mu\}{}\<[46]%
\>[46]{}\mathrel{=}\mathbf{do}\;\{\mskip1.5mu \Varid{a'}\leftarrow \Varid{mput}\;\Varid{a}\;\Varid{b};\Varid{return}\;\Varid{b}\mskip1.5mu\}{}\<[E]%
\ColumnHook
\end{hscode}\resethooks
that is, if we ``get'' \ensuremath{\Varid{b}} from \ensuremath{\Varid{a}} and then ``put'' the same \ensuremath{\Varid{b}}
value back into \ensuremath{\Varid{a}}, this has the same effect as just returning \ensuremath{\Varid{a}}
(and doing nothing else), and if we ``put'' a value $b$ and then
``get'' the result, this has the same effect as just returning $b$
after doing the ``put''.  
The obvious generalisation of composition from the pure case for these
operations is:
\begin{hscode}\SaveRestoreHook
\column{B}{@{}>{\hspre}l<{\hspost}@{}}%
\column{3}{@{}>{\hspre}l<{\hspost}@{}}%
\column{31}{@{}>{\hspre}l<{\hspost}@{}}%
\column{E}{@{}>{\hspre}l<{\hspost}@{}}%
\>[3]{}(\mathbin{;})\mathbin{::}\monadic{\alpha\rightsquigarrow_0\beta}{\mu}\to \monadic{\beta\rightsquigarrow_0\gamma}{\mu}\to \monadic{\alpha\rightsquigarrow_0\gamma}{\mu}{}\<[E]%
\\
\>[3]{}\Varid{l}_{1}\mathbin{;}\Varid{l}_{2}\mathrel{=}MLens_0\;{}\<[31]%
\>[31]{}(\lambda \Varid{a}\to \mathbf{do}\;\{\mskip1.5mu \Varid{b}\leftarrow \Varid{l}_{1}\mathord{.}\Varid{mget}\;\Varid{a};\Varid{l}_{2}\mathord{.}\Varid{mget}\;\Varid{b}\mskip1.5mu\})\;{}\<[E]%
\\
\>[31]{}(\lambda \Varid{a}\;\Varid{c}\to \mathbf{do}\;\{\mskip1.5mu \Varid{b}\leftarrow \Varid{l}_{1}\mathord{.}\Varid{mget}\;\Varid{a};\Varid{b'}\leftarrow \Varid{l}_{2}\mathord{.}\Varid{mput}\;\Varid{b}\;\Varid{c};\Varid{l}_{1}\mathord{.}\Varid{mput}\;\Varid{a}\;\Varid{b'}\mskip1.5mu\}){}\<[E]%
\ColumnHook
\end{hscode}\resethooks

This proposal has at least two apparent problems.  First, the
\ensuremath{\mathsf{(MGetPut_0)}} law appears to sharply constrain \ensuremath{\Varid{mget}}: indeed, if \ensuremath{\Varid{mget}\;\Varid{a}}
has an irreversible side-effect then \ensuremath{\mathsf{(MGetPut_0)}} 
cannot hold.  This
suggests that \ensuremath{\Varid{mget}} must either be pure, or have side-effects that
are reversible by \ensuremath{\Varid{mput}}, ruling out behaviours such as performing I/O
during \ensuremath{\Varid{mget}}.  Second, it appears difficult to compose these
structures in a way that preserves the laws, unless we again make
fairly draconian assumptions about \ensuremath{\mu}.  
In order to show \ensuremath{\mathsf{(MGetPut_0)}} for the composition \ensuremath{\Varid{l}_{1}\mathbin{;}\Varid{l}_{2}}, it seems necessary to be able to commute \ensuremath{\Varid{l}_{2}\mathord{.}\Varid{mget}} with \ensuremath{\Varid{l}_{1}\mathord{.}\Varid{mget}}
and we also need to know that doing \ensuremath{\Varid{l}_{1}\mathord{.}\Varid{mget}} 
twice is the same as doing it
just once. Likewise, to show \ensuremath{\mathsf{(MPutGet_0)}} we need to commute \ensuremath{\Varid{l}_{2}\mathord{.}\Varid{mget}} with
\ensuremath{\Varid{l}_{1}\mathord{.}\Varid{mput}}.  


\subsection{Monadic put-lenses}

\citet{pacheco14pepm} proposed a variant of lenses
called \emph{monadic putback-oriented lenses}.  For the purposes of
this paper, the putback-orientation of their approach is irrelevant:
we focus on their use of monads, and we provide a slightly simplified
version of their definition:
\begin{hscode}\SaveRestoreHook
\column{B}{@{}>{\hspre}l<{\hspost}@{}}%
\column{35}{@{}>{\hspre}l<{\hspost}@{}}%
\column{E}{@{}>{\hspre}l<{\hspost}@{}}%
\>[B]{}\mathbf{data}\;\monadic{\alpha\rightsquigarrow_1\beta}{\mu}\mathrel{=}MLens_1\;\{\mskip1.5mu {}\<[35]%
\>[35]{}\Varid{mget}\mathbin{::}\alpha\to \beta,\Varid{mput}\mathbin{::}\alpha\to \beta\to \mu\;\alpha\mskip1.5mu\}{}\<[E]%
\ColumnHook
\end{hscode}\resethooks
The main difference from their version is that we remove the \ensuremath{\Conid{Maybe}}
type constructors from the return type of \ensuremath{\Varid{mget}} and the first
argument of \ensuremath{\Varid{mput}}.
Pacheco et al. state laws for these monadic lenses.  First, they
assume that the monad \ensuremath{\mu} has a \emph{monad membership} operation 
\begin{hscode}\SaveRestoreHook
\column{B}{@{}>{\hspre}l<{\hspost}@{}}%
\column{3}{@{}>{\hspre}l<{\hspost}@{}}%
\column{E}{@{}>{\hspre}l<{\hspost}@{}}%
\>[3]{}(\in)\mathbin{::}\alpha\to \mu\;\alpha\to \Conid{Bool}{}\<[E]%
\ColumnHook
\end{hscode}\resethooks
satisfying the following two laws:
\begin{hscode}\SaveRestoreHook
\column{B}{@{}>{\hspre}l<{\hspost}@{}}%
\column{3}{@{}>{\hspre}l<{\hspost}@{}}%
\column{17}{@{}>{\hspre}l<{\hspost}@{}}%
\column{41}{@{}>{\hspre}c<{\hspost}@{}}%
\column{41E}{@{}l@{}}%
\column{46}{@{}>{\hspre}l<{\hspost}@{}}%
\column{58}{@{}>{\hspre}l<{\hspost}@{}}%
\column{E}{@{}>{\hspre}l<{\hspost}@{}}%
\>[3]{}({\in}\text{-ID}){}\<[17]%
\>[17]{}\quad\Varid{x}\in\Varid{return}\;\Varid{x}{}\<[41]%
\>[41]{}\Leftrightarrow{}\<[41E]%
\>[46]{}\Conid{True}{}\<[E]%
\\
\>[3]{}({\in}\text{-}{\bind }){}\<[17]%
\>[17]{}\quad\Varid{y}\in(\Varid{m}\bind \Varid{f}){}\<[41]%
\>[41]{}\Leftrightarrow{}\<[41E]%
\>[46]{}\exists \Varid{x}\hsforall \hsdot{\mathbin{\cdot}}{\mathrel{.}}{}\<[58]%
\>[58]{}\Varid{x}\in\Varid{m}\mathrel{\wedge}\Varid{y}\in(\Varid{f}\;\Varid{x}){}\<[E]%
\ColumnHook
\end{hscode}\resethooks
Then the laws for \ensuremath{MLens_1} 
(adapted from \citet[Prop. 3, p49]{pacheco14pepm}) 
are as follows:
\begin{hscode}\SaveRestoreHook
\column{B}{@{}>{\hspre}l<{\hspost}@{}}%
\column{3}{@{}>{\hspre}l<{\hspost}@{}}%
\column{14}{@{}>{\hspre}c<{\hspost}@{}}%
\column{14E}{@{}l@{}}%
\column{18}{@{}>{\hspre}l<{\hspost}@{}}%
\column{40}{@{}>{\hspre}l<{\hspost}@{}}%
\column{E}{@{}>{\hspre}l<{\hspost}@{}}%
\>[3]{}\mathsf{(MGetPut_1)}{}\<[14]%
\>[14]{}\quad{}\<[14E]%
\>[18]{}\Varid{v}\mathrel{=}\Varid{mget}\;\Varid{s}{}\<[40]%
\>[40]{}\Longrightarrow\Varid{mput}\;\Varid{s}\;\Varid{v}\mathrel{=}\Varid{return}\;\Varid{s}{}\<[E]%
\\
\>[3]{}\mathsf{(MPutGet_1)}{}\<[14]%
\>[14]{}\quad{}\<[14E]%
\>[18]{}\Varid{s'}\in\Varid{mput}\;\Varid{s}\;\Varid{v'}{}\<[40]%
\>[40]{}\Longrightarrow\Varid{v'}\mathrel{=}\Varid{mget}\;\Varid{s'}{}\<[E]%
\ColumnHook
\end{hscode}\resethooks
In the first law we correct an apparent typo in the original paper,
as well as removing the \ensuremath{\Conid{Just}} constructors from both laws.  By making \ensuremath{\Varid{mget}}
pure, this definition avoids the immediate problems with composition
discussed above, and Pacheco et al. outline a proof that their laws
are preserved by composition.  However, it is not obvious how to
generalise their approach beyond monads that admit a sensible \ensuremath{\in}
operation.

Many interesting monads do have a sensible \ensuremath{\in} operation
(e.g. \ensuremath{\Conid{Maybe}}, \ensuremath{[\mskip1.5mu \mskip1.5mu]}).  Pacheco et al. suggest that \ensuremath{\in} can be
defined for any monad as $x \in m \equiv (\exists h: h\,m = x)$, where
$h$ is what they call a ``(polymorphic) algebra for the monad at hand,
essentially, a
 of type \ensuremath{\Varid{m}\;\Varid{a}\to \Varid{a}} for any type \ensuremath{\Varid{a}}.''
However, this
definition doesn't appear satisfactory for monads such as \ensuremath{\Conid{IO}}, for
which there is no such (pure) function: the \ensuremath{({\in}\text{-ID})} law can
never hold in this case.  It is not clear that we can define a useful
\ensuremath{\in} operation directly for \ensuremath{\Conid{IO}} either: given that \ensuremath{\Varid{m}\mathbin{::}\Conid{IO}\;\Varid{a}}
could ultimately return any \ensuremath{\Varid{a}}-value, 
it seems safe, if perhaps
overly conservative, to define \ensuremath{\Varid{x}\in\Varid{m}\mathrel{=}\Conid{True}} for any \ensuremath{\Varid{x}} and
\ensuremath{\Varid{m}}. This satisfies the \ensuremath{\in} laws, at least, if we make a
simplifying assumption that all types are inhabited, and indeed, it
seems to be the only thing we could write in Haskell that would
satisfy the laws, since we have no way of looking inside the monadic
computation \ensuremath{\Varid{m}\mathbin{::}\Conid{IO}\;\Varid{a}} to find out what its eventual return value
is. But then the precondition of the \ensuremath{\mathsf{(MPutGet_1)}} law is always true, which
forces the view space to be trivial.  These complications suggest, at
least, that it would be advantageous to find a definition of monadic
lenses that makes sense, and is preserved under composition, for any
monad.

\subsection{Monadic lenses}

We propose the following definition of monadic lenses for any monad \ensuremath{\Conid{M}}:
\begin{definition}[monadic lens]
A \emph{monadic lens} from source type \ensuremath{\Conid{A}} to view type \ensuremath{\Conid{B}} in
which the put operation may have effects from monad \ensuremath{\Conid{M}} (or
``\ensuremath{\Conid{M}}-lens from \ensuremath{\Conid{A}} to \ensuremath{\Conid{B}}''), is represented by the type \ensuremath{\monadic{\Conid{A}\rightsquigarrow\Conid{B}}{\Conid{M}}}, where 
\begin{hscode}\SaveRestoreHook
\column{B}{@{}>{\hspre}l<{\hspost}@{}}%
\column{33}{@{}>{\hspre}l<{\hspost}@{}}%
\column{E}{@{}>{\hspre}l<{\hspost}@{}}%
\>[B]{}\mathbf{data}\;\monadic{\alpha\rightsquigarrow\beta}{\mu}\mathrel{=}\Conid{MLens}\;\{\mskip1.5mu {}\<[33]%
\>[33]{}\Varid{mget}\mathbin{::}\alpha\to \beta,\Varid{mput}\mathbin{::}\alpha\to \beta\to \mu\;\alpha\mskip1.5mu\}{}\<[E]%
\ColumnHook
\end{hscode}\resethooks
(dropping the \ensuremath{\mu} from the return type of \ensuremath{\Varid{mget}}, compared to the definition in Section~\ref{sec:naive-lens}).
We say that \ensuremath{\Conid{M}}-lens \ensuremath{\Varid{l}} is \emph{well-behaved} if it satisfies
\savecolumns
\begin{hscode}\SaveRestoreHook
\column{B}{@{}>{\hspre}l<{\hspost}@{}}%
\column{11}{@{}>{\hspre}c<{\hspost}@{}}%
\column{11E}{@{}l@{}}%
\column{15}{@{}>{\hspre}l<{\hspost}@{}}%
\column{E}{@{}>{\hspre}l<{\hspost}@{}}%
\>[B]{}\mathsf{(MGetPut)}{}\<[11]%
\>[11]{}\quad{}\<[11E]%
\>[15]{}\mathbf{do}\;\{\mskip1.5mu \Varid{l}\mathord{.}\Varid{mput}\;\Varid{a}\;(\Varid{l}\mathord{.}\Varid{mget}\;\Varid{a})\mskip1.5mu\}\mathrel{=}\Varid{return}\;\Varid{a}{}\<[E]%
\\
\>[B]{}\mathsf{(MPutGet)}{}\<[11]%
\>[11]{}\quad{}\<[11E]%
\>[15]{}\mathbf{do}\;\{\mskip1.5mu \Varid{a'}\leftarrow \Varid{l}\mathord{.}\Varid{mput}\;\Varid{a}\;\Varid{b};\Varid{k}\;\Varid{a'}\;(\Varid{l}\mathord{.}\Varid{mget}\;\Varid{a'})\mskip1.5mu\}{}\<[E]%
\\
\>[15]{}\mathrel{=}\mathbf{do}\;\{\mskip1.5mu \Varid{a'}\leftarrow \Varid{l}\mathord{.}\Varid{mput}\;\Varid{a}\;\Varid{b};\Varid{k}\;\Varid{a'}\;\Varid{b}\mskip1.5mu\}{}\<[E]%
\ColumnHook
\end{hscode}\resethooks
\endswithdisplay
\end{definition}
Note that in \ensuremath{\mathsf{(MPutGet)}}, we use a continuation \ensuremath{\Varid{k}\mathbin{::}\alpha\to \beta\to \mu\;\gamma} to quantify over all possible subsequent computations in which
\ensuremath{\Varid{a'}} and \ensuremath{\Varid{l}\mathord{.}\Varid{mget}\;\Varid{a'}} 
might appear.  In fact, using the laws of monads
and simply-typed lambda calculus we can prove this law from just the
special case \ensuremath{\Varid{k}\mathrel{=}\lambda \Varid{a}\;\Varid{b}\to \Varid{return}\;(\Varid{a},\Varid{b})}, so in the sequel when we
prove \ensuremath{\mathsf{(MPutGet)}} we may just prove this case while using the strong
form freely in the proof.

The ordinary asymmetric lenses are exactly the monadic lenses over
\ensuremath{\mu\mathrel{=}\Conid{Id}}; the laws then specialise to the standard
equational laws.  Monadic lenses where \ensuremath{\mu\mathrel{=}\Conid{Id}} are called
\emph{pure}, and we may refer to ordinary lenses as pure lenses also.
\begin{definition}
We can also define an operation that lifts a pure lens to a monadic lens:
\begin{hscode}\SaveRestoreHook
\column{B}{@{}>{\hspre}l<{\hspost}@{}}%
\column{E}{@{}>{\hspre}l<{\hspost}@{}}%
\>[B]{}\Varid{lens2mlens}\mathbin{::}\Conid{Monad}\;\mu\Rightarrow \alpha\mathbin{\leadsto}\beta\to \monadic{\alpha\rightsquigarrow\beta}{\mu}{}\<[E]%
\\
\>[B]{}\Varid{lens2mlens}\;\Varid{l}\mathrel{=}\Conid{MLens}\;(\Varid{l}\mathord{.}\Varid{get})\;(\lambda \Varid{a}\;\Varid{b}\to \Varid{return}\;(\Varid{l}\mathord{.}\Varid{put}\;\Varid{a}\;\Varid{b})){}\<[E]%
\ColumnHook
\end{hscode}\resethooks
\endswithdisplay
\end{definition}
\begin{lemma}\label{lem:lens2mlens-wb}
  If \ensuremath{\Varid{l}\mathbin{::}\Conid{Lens}\;\alpha\;\beta} is well-behaved, then so is \ensuremath{\Varid{lens2mlens}\;\Varid{l}}.
\end{lemma}
\begin{example}
  To illustrate, some simple pure lenses include:
\begin{hscode}\SaveRestoreHook
\column{B}{@{}>{\hspre}l<{\hspost}@{}}%
\column{E}{@{}>{\hspre}l<{\hspost}@{}}%
\>[B]{}id_{\mathrm{l}}\mathbin{::}\alpha\mathbin{\leadsto}\alpha{}\<[E]%
\\
\>[B]{}id_{\mathrm{l}}\mathrel{=}\Conid{Lens}\;(\lambda \Varid{a}\to \Varid{a})\;(\lambda \anonymous \;\Varid{a}\to \Varid{a}){}\<[E]%
\\[\blanklineskip]%
\>[B]{}fst_{\mathrm{l}}\mathbin{::}(\alpha,\beta)\mathbin{\leadsto}\alpha{}\<[E]%
\\
\>[B]{}fst_{\mathrm{l}}\mathrel{=}\Conid{MLens}\;\Varid{fst}\;(\lambda (\Varid{s}_{1},\Varid{s}_{2})\;\Varid{s}_{1}'\to (\Varid{s}_{1}',\Varid{s}_{2})){}\<[E]%
\ColumnHook
\end{hscode}\resethooks
Many more examples of pure lenses are to be found in the
literature~\citep{lens-toplas,foster10ssgip}, all of which lift to
well-behaved monadic lenses.
\end{example}





As more interesting examples, we present asymmetric versions of the
partial and logging lenses presented by \citet{abousaleh15mpc}.  Pure lenses are usually defined using
total functions, which means that  \ensuremath{\Varid{get}} must be surjective whenever
\ensuremath{\Conid{A}} is nonempty, and \ensuremath{\Varid{put}}
must be defined for all source and view pairs.  One way to accommodate
partiality is to adjust the return type of \ensuremath{\Varid{get}} to \ensuremath{\Conid{Maybe}\;\Varid{b}} or give
\ensuremath{\Varid{put}} the return type \ensuremath{\Conid{Maybe}\;\Varid{a}} to allow for failure if we attempt to
put a \ensuremath{\Varid{b}}-value that is not in the range of \ensuremath{\Varid{get}}. In either case, the
laws need to be adjusted somehow.  Monadic lenses allow for partiality
without requiring such an ad hoc change. A trivial example is
 \begin{hscode}\SaveRestoreHook
\column{B}{@{}>{\hspre}l<{\hspost}@{}}%
\column{4}{@{}>{\hspre}l<{\hspost}@{}}%
\column{26}{@{}>{\hspre}l<{\hspost}@{}}%
\column{E}{@{}>{\hspre}l<{\hspost}@{}}%
\>[4]{}\Varid{constMLens}\mathbin{::}\beta\to \monadic{\alpha\rightsquigarrow\beta}{\Conid{Maybe}}{}\<[E]%
\\
\>[4]{}\Varid{constMLens}\;\Varid{b}\mathrel{=}\Conid{MLens}\;{}\<[26]%
\>[26]{}(\Varid{const}\;\Varid{b})\;{}\<[E]%
\\
\>[26]{}(\lambda \Varid{a}\;\Varid{b'}\to \mathbf{if}\;\Varid{b}\equals\Varid{b'}\;\mathbf{then}\;\Conid{Just}\;\Varid{a}\;\mathbf{else}\;\Conid{Nothing}){}\<[E]%
\ColumnHook
\end{hscode}\resethooks
which is well-behaved because both sides of \ensuremath{\mathsf{(MPutGet)}} fail if the view
is changed to a value different from \ensuremath{\Varid{b}}.  Of course, this example also
illustrates that the \ensuremath{\Varid{mget}} function of a monadic lens need not be surjective.

As a more interesting example, consider:
  \begin{hscode}\SaveRestoreHook
\column{B}{@{}>{\hspre}l<{\hspost}@{}}%
\column{5}{@{}>{\hspre}l<{\hspost}@{}}%
\column{22}{@{}>{\hspre}l<{\hspost}@{}}%
\column{32}{@{}>{\hspre}l<{\hspost}@{}}%
\column{E}{@{}>{\hspre}l<{\hspost}@{}}%
\>[5]{}\Varid{absLens}\mathbin{::}\monadic{\Conid{Int}\rightsquigarrow\Conid{Int}}{\Conid{Maybe}}{}\<[E]%
\\
\>[5]{}\Varid{absLens}\mathrel{=}\Conid{MLens}\;{}\<[22]%
\>[22]{}\Varid{abs}\;{}\<[E]%
\\
\>[22]{}(\lambda \Varid{a}\;\Varid{b}\to {}\<[32]%
\>[32]{}\mathbf{if}\;\Varid{b}\mathbin{<}\mathrm{0}{}\<[E]%
\\
\>[32]{}\mathbf{then}\;\Conid{Nothing}{}\<[E]%
\\
\>[32]{}\mathbf{else}\;\Conid{Just}\;(\mathbf{if}\;\Varid{a}\mathbin{<}\mathrm{0}\;\mathbf{then}\mathbin{-}\Varid{b}\;\mathbf{else}\;\Varid{b})){}\<[E]%
\ColumnHook
\end{hscode}\resethooks
In the \ensuremath{\Varid{mget}} direction, this lens maps a source number to its absolute value;
in the reverse direction, it fails if the view \ensuremath{\Varid{b}} is negative, and otherwise uses
the sign of the previous source \ensuremath{\Varid{a}} to determine the sign of the updated source.

The following \emph{logging lens} takes a pure lens \ensuremath{\Varid{l}} and, whenever
the source value \ensuremath{\Varid{a}} changes, records the previous \ensuremath{\Varid{a}} value.
  \begin{hscode}\SaveRestoreHook
\column{B}{@{}>{\hspre}l<{\hspost}@{}}%
\column{5}{@{}>{\hspre}l<{\hspost}@{}}%
\column{18}{@{}>{\hspre}l<{\hspost}@{}}%
\column{26}{@{}>{\hspre}l<{\hspost}@{}}%
\column{47}{@{}>{\hspre}l<{\hspost}@{}}%
\column{E}{@{}>{\hspre}l<{\hspost}@{}}%
\>[5]{}\Varid{logLens}\mathbin{::}\Conid{Eq}\;\alpha\Rightarrow {}\<[26]%
\>[26]{}\alpha\mathbin{\leadsto}\beta\to \monadic{\alpha\rightsquigarrow\beta}{\Conid{Writer}\;\alpha}{}\<[E]%
\\
\>[5]{}\Varid{logLens}\;\Varid{l}\mathrel{=}\Conid{MLens}\;(\Varid{l}\mathord{.}\Varid{get})\;(\lambda \Varid{a}\;\Varid{b}\to {}\<[E]%
\\
\>[5]{}\hsindent{13}{}\<[18]%
\>[18]{}\mathbf{let}\;\Varid{a'}\mathrel{=}\Varid{l}\mathord{.}\Varid{put}\;\Varid{a}\;\Varid{b}\;\mathbf{in}\;\mathbf{do}\;\{\mskip1.5mu {}\<[47]%
\>[47]{}\mathbf{if}\;\Varid{a}\notequals\Varid{a'}\;\mathbf{then}\;\Varid{tell}\;\Varid{a}\;\mathbf{else}\;\Varid{return}\;();\Varid{return}\;\Varid{a'}\mskip1.5mu\}){}\<[E]%
\ColumnHook
\end{hscode}\resethooks
We presented a number of more involved examples of effectful symmetric bx in 
\citep{abousaleh15mpc}.  They show how monadic lenses can employ user
interaction, state, or nondeterminism to restore consistency.  Most of
these examples are equivalently definable as \emph{spans} of monadic
lenses, which we will discuss in the next section.


In practical use, it is usually also
necessary to equip lenses with an \emph{initialisation} mechanism.
Indeed, as already mentioned, Pacheco et al.'s monadic put-lenses make
the \ensuremath{\alpha} argument optional (using \ensuremath{\Conid{Maybe}}), to allow for
initialisation when only a \ensuremath{\beta} is available; we chose to
exclude this from our version of monadic lenses above.  

We propose the following alternative:
\begin{hscode}\SaveRestoreHook
\column{B}{@{}>{\hspre}l<{\hspost}@{}}%
\column{3}{@{}>{\hspre}l<{\hspost}@{}}%
\column{35}{@{}>{\hspre}l<{\hspost}@{}}%
\column{E}{@{}>{\hspre}l<{\hspost}@{}}%
\>[3]{}\mathbf{data}\;\monadic{\alpha\rightsquigarrow\beta}{\mu}\mathrel{=}\Conid{MLens}\;\{\mskip1.5mu {}\<[35]%
\>[35]{}\Varid{mget}\mathbin{::}\alpha\to \beta,{}\<[E]%
\\
\>[35]{}\Varid{mput}\mathbin{::}\alpha\to \beta\to \mu\;\alpha,{}\<[E]%
\\
\>[35]{}\Varid{mcreate}\mathbin{::}\beta\to \mu\;\alpha\mskip1.5mu\}{}\<[E]%
\ColumnHook
\end{hscode}\resethooks
and we consider such initialisable monadic lenses to be well-behaved
when they satisfy the following additional law:
\begin{hscode}\SaveRestoreHook
\column{B}{@{}>{\hspre}l<{\hspost}@{}}%
\column{11}{@{}>{\hspre}c<{\hspost}@{}}%
\column{11E}{@{}l@{}}%
\column{15}{@{}>{\hspre}l<{\hspost}@{}}%
\column{E}{@{}>{\hspre}l<{\hspost}@{}}%
\>[B]{}\mathsf{(MCreateGet)}{}\<[11]%
\>[11]{}\quad{}\<[11E]%
\>[15]{}\mathbf{do}\;\{\mskip1.5mu \Varid{a}\leftarrow \Varid{mcreate}\;\Varid{b};\Varid{k}\;\Varid{a}\;(\Varid{mget}\;\Varid{a})\mskip1.5mu\}\mathrel{=}\mathbf{do}\;\{\mskip1.5mu \Varid{a}\leftarrow \Varid{mcreate}\;\Varid{b};\Varid{k}\;\Varid{a}\;\Varid{b}\mskip1.5mu\}{}\<[E]%
\ColumnHook
\end{hscode}\resethooks
As with \ensuremath{\mathsf{(MPutGet)}}, this property follows from the special case \ensuremath{\Varid{k}\mathrel{=}\lambda \Varid{x}\;\Varid{y}\to \Varid{return}\;(\Varid{x},\Varid{y})},
and we will use this fact freely.

This approach, in our view, helps keep the \ensuremath{\mathsf{(GetPut)}} and \ensuremath{\mathsf{(PutGet)}}
laws simple and clear, and avoids the need to wrap \ensuremath{\Varid{mput}}'s first
argument in \ensuremath{\Conid{Just}} whenever it is called.  

Next, we consider composition of monadic lenses.
\begin{hscode}\SaveRestoreHook
\column{B}{@{}>{\hspre}l<{\hspost}@{}}%
\column{3}{@{}>{\hspre}l<{\hspost}@{}}%
\column{28}{@{}>{\hspre}l<{\hspost}@{}}%
\column{E}{@{}>{\hspre}l<{\hspost}@{}}%
\>[B]{}(\mathbin{;})\mathbin{::}\Conid{Monad}\;\mu\Rightarrow \monadic{\alpha\rightsquigarrow\beta}{\mu}\to \monadic{\beta\rightsquigarrow\gamma}{\mu}\to \monadic{\alpha\rightsquigarrow\gamma}{\mu}{}\<[E]%
\\
\>[B]{}\Varid{l}_{1}\mathbin{;}\Varid{l}_{2}\mathrel{=}\Conid{MLens}\;{}\<[28]%
\>[28]{}(\Varid{l}_{2}\mathord{.}\Varid{mget}\hsdot{\mathbin{\cdot}}{\mathrel{.}}\Varid{l}_{1}\mathord{.}\Varid{mget})\;\Varid{mput}\;\Varid{mcreate}\;\mathbf{where}{}\<[E]%
\\
\>[B]{}\hsindent{3}{}\<[3]%
\>[3]{}\Varid{mput}\;\Varid{a}\;\Varid{c}\mathrel{=}\mathbf{do}\;\{\mskip1.5mu \Varid{b}\leftarrow \Varid{l}_{2}\mathord{.}\Varid{mput}\;(\Varid{l}_{1}\mathord{.}\Varid{mget}\;\Varid{a})\;\Varid{c};\Varid{l}_{1}\mathord{.}\Varid{mput}\;\Varid{a}\;\Varid{b}\mskip1.5mu\}{}\<[E]%
\\
\>[B]{}\hsindent{3}{}\<[3]%
\>[3]{}\Varid{mcreate}\;\Varid{c}\mathrel{=}\mathbf{do}\;\{\mskip1.5mu \Varid{b}\leftarrow \Varid{l}_{2}\mathord{.}\Varid{mcreate};\Varid{l}_{1}\mathord{.}\Varid{mcreate}\mskip1.5mu\}{}\<[E]%
\ColumnHook
\end{hscode}\resethooks
Note that we consider only the simple case in which the lenses share a
common monad \ensuremath{\mu}.  Composing lenses with effects in different monads
would require determining how to compose the monads themselves, which is nontrivial~\citep{DBLP:conf/fp/KingW92,Jones&Duponcheel93:Composing}.
\begin{theorem}\label{thm:mlens-comp-wb}
  If \ensuremath{\Varid{l}_{1}\mathbin{::}\monadic{\Conid{A}\rightsquigarrow\Conid{B}}{\Conid{M}},\Varid{l}_{2}\mathbin{::}\monadic{\Conid{B}\rightsquigarrow\Conid{C}}{\Conid{M}}} are well-behaved,
  then so is \ensuremath{\Varid{l}_{1}\mathbin{;}\Varid{l}_{2}}.  
\end{theorem}



\section{Symmetric monadic lenses and spans}\label{sec:symmetric}

\citet{symlens} proposed \emph{symmetric
  lenses} that use a \emph{complement} to store (at least) the
information that is not present in both views. 
\begin{hscode}\SaveRestoreHook
\column{B}{@{}>{\hspre}l<{\hspost}@{}}%
\column{33}{@{}>{\hspre}l<{\hspost}@{}}%
\column{39}{@{}>{\hspre}l<{\hspost}@{}}%
\column{42}{@{}>{\hspre}l<{\hspost}@{}}%
\column{E}{@{}>{\hspre}l<{\hspost}@{}}%
\>[B]{}\mathbf{data}\;\alpha\mathbin{\stackrel{\gamma}{\longleftrightarrow}}\beta\mathrel{=}\Conid{SLens}\;\{\mskip1.5mu {}\<[33]%
\>[33]{}\Varid{put}_\mathrm{R}{}\<[42]%
\>[42]{}\mathbin{::}(\alpha,\gamma)\to (\beta,\gamma),{}\<[E]%
\\
\>[33]{}\Varid{put}_\mathrm{L}{}\<[39]%
\>[39]{}\mathbin{::}(\beta,\gamma)\to (\alpha,\gamma),{}\<[E]%
\\
\>[33]{}\Varid{missing}{}\<[42]%
\>[42]{}\mathbin{::}\gamma\mskip1.5mu\}{}\<[E]%
\ColumnHook
\end{hscode}\resethooks
Informally, \ensuremath{\Varid{put}_\mathrm{R}} turns an \ensuremath{\alpha} into a \ensuremath{\beta}, modifying a complement \ensuremath{\gamma} as it goes, and symmetrically for \ensuremath{\Varid{put}_\mathrm{L}}; and \ensuremath{\Varid{missing}} is an initial complement, to get the ball rolling.
Well-behavedness for symmetric lenses amounts to the following
equational laws:
\begin{hscode}\SaveRestoreHook
\column{B}{@{}>{\hspre}l<{\hspost}@{}}%
\column{13}{@{}>{\hspre}c<{\hspost}@{}}%
\column{13E}{@{}l@{}}%
\column{20}{@{}>{\hspre}l<{\hspost}@{}}%
\column{22}{@{}>{\hspre}c<{\hspost}@{}}%
\column{22E}{@{}l@{}}%
\column{25}{@{}>{\hspre}l<{\hspost}@{}}%
\column{53}{@{}>{\hspre}l<{\hspost}@{}}%
\column{E}{@{}>{\hspre}l<{\hspost}@{}}%
\>[B]{}\mathsf{(PutRL)}{}\<[13]%
\>[13]{}\quad{}\<[13E]%
\>[20]{}\mathbf{let}\;(\Varid{b},\Varid{c'})\mathrel{=}\Varid{sl}\mathord{.}\Varid{put}_\mathrm{R}\;(\Varid{a},\Varid{c})\;\mathbf{in}\;\Varid{sl}\mathord{.}\Varid{put}_\mathrm{L}\;(\Varid{b},\Varid{c'}){}\<[E]%
\\
\>[20]{}\hsindent{2}{}\<[22]%
\>[22]{}\mathrel{=}{}\<[22E]%
\>[25]{}\mathbf{let}\;(\Varid{b},\Varid{c'})\mathrel{=}\Varid{sl}\mathord{.}\Varid{put}_\mathrm{R}\;(\Varid{a},\Varid{c})\;\mathbf{in}\;(\Varid{a},\Varid{c'}){}\<[E]%
\\
\>[B]{}\mathsf{(PutLR)}{}\<[13]%
\>[13]{}\quad{}\<[13E]%
\>[20]{}\mathbf{let}\;(\Varid{a},\Varid{c'})\mathrel{=}\Varid{sl}\mathord{.}\Varid{put}_\mathrm{L}\;(\Varid{b},\Varid{c})\;\mathbf{in}\;\Varid{sl}\mathord{.}\Varid{put}_\mathrm{R}\;(\Varid{a},\Varid{c'}){}\<[E]%
\\
\>[20]{}\hsindent{2}{}\<[22]%
\>[22]{}\mathrel{=}{}\<[22E]%
\>[25]{}\mathbf{let}\;(\Varid{a},\Varid{c'})\mathrel{=}\Varid{sl}\mathord{.}\Varid{put}_\mathrm{L}\;(\Varid{b},\Varid{c})\;{}\<[53]%
\>[53]{}\mathbf{in}\;(\Varid{b},\Varid{c'}){}\<[E]%
\ColumnHook
\end{hscode}\resethooks
Furthermore, the composition of two symmetric lenses preserves
well-behavedness, and can be defined as
follows:
\begin{hscode}\SaveRestoreHook
\column{B}{@{}>{\hspre}l<{\hspost}@{}}%
\column{3}{@{}>{\hspre}l<{\hspost}@{}}%
\column{5}{@{}>{\hspre}l<{\hspost}@{}}%
\column{25}{@{}>{\hspre}l<{\hspost}@{}}%
\column{30}{@{}>{\hspre}l<{\hspost}@{}}%
\column{E}{@{}>{\hspre}l<{\hspost}@{}}%
\>[3]{}(\mathbin{;})\mathbin{::}(\alpha\mathbin{\stackrel{\sigma_1}{\longleftrightarrow}}\beta)\to (\beta\mathbin{\stackrel{\sigma_2}{\longleftrightarrow}}\gamma)\to (\alpha\mathbin{\stackrel{(\sigma_1,\sigma_2)}{\longleftrightarrow}}\gamma){}\<[E]%
\\
\>[3]{}\Varid{l}_{1}\mathbin{;}\Varid{l}_{2}\mathrel{=}\Conid{SLens}\;\Varid{put}_\mathrm{R}\;\Varid{put}_\mathrm{L}\;(\Varid{l}_{1}\mathord{.}\Varid{missing},\Varid{l}_{2}\mathord{.}\Varid{missing})\;\mathbf{where}{}\<[E]%
\\
\>[3]{}\hsindent{2}{}\<[5]%
\>[5]{}\Varid{put}_\mathrm{R}\;(\Varid{a},(\Varid{s}_{1},\Varid{s}_{2}))\mathrel{=}{}\<[25]%
\>[25]{}\mathbf{let}\;{}\<[30]%
\>[30]{}(\Varid{b},\Varid{s}_{1}')\mathrel{=}\Varid{put}_\mathrm{R}\;(\Varid{a},\Varid{s}_{1}){}\<[E]%
\\
\>[30]{}(\Varid{c},\Varid{s}_{2}')\mathrel{=}\Varid{put}_\mathrm{R}\;(\Varid{b},\Varid{s}_{2}){}\<[E]%
\\
\>[25]{}\mathbf{in}\;(\Varid{c},(\Varid{s}_{1}',\Varid{s}_{2}')){}\<[E]%
\\
\>[3]{}\hsindent{2}{}\<[5]%
\>[5]{}\Varid{put}_\mathrm{L}\;(\Varid{c},(\Varid{s}_{1},\Varid{s}_{2}))\mathrel{=}{}\<[25]%
\>[25]{}\mathbf{let}\;{}\<[30]%
\>[30]{}(\Varid{b},\Varid{s}_{2}')\mathrel{=}\Varid{put}_\mathrm{L}\;(\Varid{c},\Varid{s}_{2}){}\<[E]%
\\
\>[30]{}(\Varid{a},\Varid{s}_{1}')\mathrel{=}\Varid{put}_\mathrm{L}\;(\Varid{b},\Varid{s}_{1}){}\<[E]%
\\
\>[25]{}\mathbf{in}\;(\Varid{a},(\Varid{s}_{1}',\Varid{s}_{2}')){}\<[E]%
\ColumnHook
\end{hscode}\resethooks



We can define an \emph{identity} symmetric lens as follows:
\begin{hscode}\SaveRestoreHook
\column{B}{@{}>{\hspre}l<{\hspost}@{}}%
\column{3}{@{}>{\hspre}l<{\hspost}@{}}%
\column{E}{@{}>{\hspre}l<{\hspost}@{}}%
\>[3]{}id_{\mathrm{sl}}\mathbin{::}\alpha\mathbin{\stackrel{()}{\longleftrightarrow}}\alpha{}\<[E]%
\\
\>[3]{}id_{\mathrm{sl}}\mathrel{=}\Conid{SLens}\;\Varid{id}\;\Varid{id}\;(){}\<[E]%
\ColumnHook
\end{hscode}\resethooks
It is natural to wonder whether symmetric lens composition satisfies
identity and associativity laws making symmetric lenses into a
category.  This is complicated by the fact that the complement types of
the composition \ensuremath{id_{\mathrm{sl}};\Varid{sl}} and of \ensuremath{\Varid{sl}} differ, so it is not even 
type-correct to ask whether \ensuremath{id_{\mathrm{sl}};\Varid{sl}} and \ensuremath{\Varid{sl}} are equal.  To make
it possible to relate the behaviour of symmetric lenses with different
complement types, \HPWlong{} defined equivalence of symmetric lenses
as follows:
\begin{definition}
  Suppose \ensuremath{\Conid{R}\subseteq\Conid{C}_{1} \times \Conid{C}_{2}}.  Then \ensuremath{\Varid{f} \sim_{\Conid{R}} \Varid{g}} 
  means that for
  all \ensuremath{\Varid{c}_{1},\Varid{c}_{2},\Varid{x}}, if \ensuremath{(\Varid{c}_{1},\Varid{c}_{2})\in\Conid{R}} and \ensuremath{(\Varid{y},\Varid{c}_{1}')\mathrel{=}\Varid{f}\;(\Varid{x},\Varid{c}_{1})} and \ensuremath{(\Varid{y'},\Varid{c}_{2}')\mathrel{=}\Varid{g}\;(\Varid{y},\Varid{c}_{2})}, then \ensuremath{\Varid{y}\mathrel{=}\Varid{y'}} and \ensuremath{(\Varid{c}_{1}',\Varid{c}_{2}')\in\Conid{R}}.
\end{definition}
\begin{definition}[Symmetric lens equivalence]\label{def:hpw-equiv}
  Two symmetric lenses \ensuremath{\Varid{sl}_{1}\mathbin{::}\Conid{X}\mathbin{\stackrel{\Conid{C}_{1}}{\longleftrightarrow}}\Conid{Y}} and \ensuremath{\Varid{sl}_{2}\mathbin{::}\Conid{X}\mathbin{\stackrel{\Conid{C}_{2}}{\longleftrightarrow}}\Conid{Y}}
  are considered \emph{equivalent} (\ensuremath{\Varid{sl}_{1}\equiv_{\mathrm{sl}} \Varid{sl}_{2}}) if there is a
  relation \ensuremath{\Conid{R}\subseteq\Conid{C}_{1} \times \Conid{C}_{2}} such that 
  \begin{enumerate}
\item     \ensuremath{(\Varid{sl}_{1}\mathord{.}\Varid{missing},\Varid{sl}_{2}\mathord{.}\Varid{missing})\in\Conid{R}}, 
\item \ensuremath{\Varid{sl}_{1}\mathord{.}\Varid{put}_\mathrm{R} \sim_{\Conid{R}} \Varid{sl}_{2}\mathord{.}\Varid{put}_\mathrm{R}}, 
and 
\item \ensuremath{\Varid{sl}_{1}\mathord{.}\Varid{put}_\mathrm{L} \sim_{\Conid{R}} \Varid{sl}_{2}\mathord{.}\Varid{put}_\mathrm{L}}.
  \end{enumerate}\endswithdisplay
\end{definition}
\HPWlong  show that \ensuremath{\equiv_{\mathrm{sl}} } is an equivalence relation; moreover it is sufficiently strong to validate identity, associativity and congruence laws:
\begin{theorem}
  [\citep{symlens}]
  If \ensuremath{\Varid{sl}_{1}\mathbin{::}\Conid{X}\mathbin{\stackrel{\Conid{C}_{1}}{\longleftrightarrow}}\Conid{Y}} and \ensuremath{\Varid{sl}_{2}\mathbin{::}\Conid{Y}\mathbin{\stackrel{\Conid{C}_{2}}{\longleftrightarrow}}\Conid{Z}} are
  well-behaved, then so is \ensuremath{\Varid{sl}_{1}\mathbin{;}\Varid{sl}_{2}}.  In addition,
  composition satisfies the laws:
\begin{hscode}\SaveRestoreHook
\column{B}{@{}>{\hspre}l<{\hspost}@{}}%
\column{14}{@{}>{\hspre}l<{\hspost}@{}}%
\column{18}{@{}>{\hspre}l<{\hspost}@{}}%
\column{E}{@{}>{\hspre}l<{\hspost}@{}}%
\>[B]{}\mathsf{(Identity)}{}\<[14]%
\>[14]{}\quad{}\<[18]%
\>[18]{}\Varid{sl}\mathbin{;}id_{\mathrm{sl}}\equiv_{\mathrm{sl}} \Varid{sl}\equiv_{\mathrm{sl}} id_{\mathrm{sl}}\mathbin{;}\Varid{sl}{}\<[E]%
\\
\>[B]{}\mathsf{(Assoc)}{}\<[14]%
\>[14]{}\quad\Varid{sl}_{1}\mathbin{;}(\Varid{sl}_{2}\mathbin{;}\Varid{sl}_{3})\equiv_{\mathrm{sl}} (\Varid{sl}_{1}\mathbin{;}\Varid{sl}_{2})\mathbin{;}\Varid{sl}_{3}{}\<[E]%
\\
\>[B]{}\mathsf{(Cong)}{}\<[14]%
\>[14]{}\quad\Varid{sl}_{1}\equiv_{\mathrm{sl}} \Varid{sl}_{1}'\mathrel{\wedge}\Varid{sl}_{2}\equiv_{\mathrm{sl}} \Varid{sl}_{2}'\Longrightarrow\Varid{sl}_{1}\mathbin{;}\Varid{sl}_{2}\equiv_{\mathrm{sl}} \Varid{sl}_{1}'\mathbin{;}\Varid{sl}_{2}'{}\<[E]%
\ColumnHook
\end{hscode}\resethooks
\endswithdisplay
\end{theorem}

\subsection{Naive monadic symmetric lenses}

We now consider an obvious
monadic generalisation of symmetric lenses, in which the \ensuremath{\Varid{put}_\mathrm{L}} and \ensuremath{\Varid{put}_\mathrm{R}} functions are allowed to have effects in some monad \ensuremath{\Conid{M}}:
\begin{definition}
  A \emph{monadic symmetric lens} from \ensuremath{\Conid{A}} to \ensuremath{\Conid{B}} with complement type \ensuremath{\Conid{C}}
  and effects \ensuremath{\Conid{M}} consists of two functions converting \ensuremath{\Conid{A}} to \ensuremath{\Conid{B}} and vice
  versa, each also operating on \ensuremath{\Conid{C}} and possibly having effects in \ensuremath{\Conid{M}}, and
  a complement value \ensuremath{\Varid{missing}} used for initialisation: 
\begin{hscode}\SaveRestoreHook
\column{B}{@{}>{\hspre}l<{\hspost}@{}}%
\column{39}{@{}>{\hspre}l<{\hspost}@{}}%
\column{49}{@{}>{\hspre}l<{\hspost}@{}}%
\column{E}{@{}>{\hspre}l<{\hspost}@{}}%
\>[B]{}\mathbf{data}\;\monadic{\alpha\mathbin{\overset{\gamma}{\longleftrightarrow}}\beta}{\mu}\mathrel{=}\Conid{SMLens}\;\{\mskip1.5mu {}\<[39]%
\>[39]{}\Varid{mput}_\mathrm{R}{}\<[49]%
\>[49]{}\mathbin{::}(\alpha,\gamma)\to \mu\;(\beta,\gamma),{}\<[E]%
\\
\>[39]{}\Varid{mput}_\mathrm{L}{}\<[49]%
\>[49]{}\mathbin{::}(\beta,\gamma)\to \mu\;(\alpha,\gamma),{}\<[E]%
\\
\>[39]{}\Varid{missing}{}\<[49]%
\>[49]{}\mathbin{::}\gamma\mskip1.5mu\}{}\<[E]%
\ColumnHook
\end{hscode}\resethooks
Such a lens \ensuremath{\Varid{sl}} is called \emph{well-behaved} if:
\begin{hscode}\SaveRestoreHook
\column{B}{@{}>{\hspre}l<{\hspost}@{}}%
\column{14}{@{}>{\hspre}c<{\hspost}@{}}%
\column{14E}{@{}l@{}}%
\column{21}{@{}>{\hspre}l<{\hspost}@{}}%
\column{23}{@{}>{\hspre}c<{\hspost}@{}}%
\column{23E}{@{}l@{}}%
\column{26}{@{}>{\hspre}l<{\hspost}@{}}%
\column{E}{@{}>{\hspre}l<{\hspost}@{}}%
\>[B]{}\mathsf{(PutRLM)}{}\<[14]%
\>[14]{}\quad{}\<[14E]%
\>[21]{}\mathbf{do}\;\{\mskip1.5mu (\Varid{b},\Varid{c'})\leftarrow \Varid{sl}\mathord{.}\Varid{mput}_\mathrm{R}\;(\Varid{a},\Varid{c});\Varid{sl}\mathord{.}\Varid{mput}_\mathrm{L}\;(\Varid{b},\Varid{c'})\mskip1.5mu\}{}\<[E]%
\\
\>[21]{}\hsindent{2}{}\<[23]%
\>[23]{}\mathrel{=}{}\<[23E]%
\>[26]{}\mathbf{do}\;\{\mskip1.5mu (\Varid{b},\Varid{c'})\leftarrow \Varid{sl}\mathord{.}\Varid{mput}_\mathrm{R}\;(\Varid{a},\Varid{c});\Varid{return}\;(\Varid{a},\Varid{c'})\mskip1.5mu\}{}\<[E]%
\\
\>[B]{}\mathsf{(PutLRM)}{}\<[14]%
\>[14]{}\quad{}\<[14E]%
\>[21]{}\mathbf{do}\;\{\mskip1.5mu (\Varid{a},\Varid{c'})\leftarrow \Varid{sl}\mathord{.}\Varid{mput}_\mathrm{L}\;(\Varid{b},\Varid{c});\Varid{sl}\mathord{.}\Varid{mput}_\mathrm{R}\;(\Varid{a},\Varid{c'})\mskip1.5mu\}{}\<[E]%
\\
\>[21]{}\hsindent{2}{}\<[23]%
\>[23]{}\mathrel{=}{}\<[23E]%
\>[26]{}\mathbf{do}\;\{\mskip1.5mu (\Varid{a},\Varid{c'})\leftarrow \Varid{sl}\mathord{.}\Varid{mput}_\mathrm{L}\;(\Varid{b},\Varid{c});\Varid{return}\;(\Varid{b},\Varid{c'})\mskip1.5mu\}{}\<[E]%
\ColumnHook
\end{hscode}\resethooks
\endswithdisplay
\end{definition}

The above monadic generalisation of symmetric lenses appears natural,
but it turns out to have some idiosyncrasies, similar to those of the
naive version of monadic lenses we considered in
Section~\ref{sec:naive-lens}.


\paragraph{Composition and well-behavedness}

Consider the following candidate definition of composition for
monadic symmetric lenses:
\begin{hscode}\SaveRestoreHook
\column{B}{@{}>{\hspre}l<{\hspost}@{}}%
\column{9}{@{}>{\hspre}l<{\hspost}@{}}%
\column{16}{@{}>{\hspre}l<{\hspost}@{}}%
\column{27}{@{}>{\hspre}l<{\hspost}@{}}%
\column{29}{@{}>{\hspre}l<{\hspost}@{}}%
\column{35}{@{}>{\hspre}l<{\hspost}@{}}%
\column{E}{@{}>{\hspre}l<{\hspost}@{}}%
\>[B]{}(\mathbin{;})\mathbin{::}{}\<[16]%
\>[16]{}\Conid{Monad}\;\mu\Rightarrow {}\<[29]%
\>[29]{}\monadic{\alpha\mathbin{\overset{\sigma_1}{\longleftrightarrow}}\beta}{\mu}\to \monadic{\beta\mathbin{\overset{\sigma_2}{\longleftrightarrow}}\gamma}{\mu}\to \monadic{\alpha\mathbin{\overset{(\sigma_1,\sigma_2)}{\longleftrightarrow}}\gamma}{\mu}{}\<[E]%
\\
\>[B]{}\Varid{sl}_{1}\mathbin{;}\Varid{sl}_{2}\mathrel{=}\Conid{SMLens}\;\Varid{put}_\mathrm{R}\;\Varid{put}_\mathrm{L}\;\Varid{missing}\;\mathbf{where}{}\<[E]%
\\
\>[B]{}\hsindent{9}{}\<[9]%
\>[9]{}\Varid{put}_\mathrm{R}\;(\Varid{a},(\Varid{s}_{1},\Varid{s}_{2})){}\<[27]%
\>[27]{}\mathrel{=}\mathbf{do}\;\{\mskip1.5mu {}\<[35]%
\>[35]{}(\Varid{b},\Varid{s}_{1}')\leftarrow \Varid{sl}_{1}\mathord{.}\Varid{mput}_\mathrm{R}\;(\Varid{a},\Varid{s}_{1});{}\<[E]%
\\
\>[35]{}(\Varid{c},\Varid{s}_{2}')\leftarrow \Varid{sl}_{2}\mathord{.}\Varid{mput}_\mathrm{R}\;(\Varid{b},\Varid{s}_{2});{}\<[E]%
\\
\>[35]{}\Varid{return}\;(\Varid{c},(\Varid{s}_{1}',\Varid{s}_{2}'))\mskip1.5mu\}{}\<[E]%
\\
\>[B]{}\hsindent{9}{}\<[9]%
\>[9]{}\Varid{put}_\mathrm{L}\;(\Varid{c},(\Varid{s}_{1},\Varid{s}_{2})){}\<[27]%
\>[27]{}\mathrel{=}\mathbf{do}\;\{\mskip1.5mu {}\<[35]%
\>[35]{}(\Varid{b},\Varid{s}_{2}')\leftarrow \Varid{sl}_{2}\mathord{.}\Varid{mput}_\mathrm{L}\;(\Varid{c},\Varid{s}_{2});{}\<[E]%
\\
\>[35]{}(\Varid{a},\Varid{s}_{1}')\leftarrow \Varid{sl}_{1}\mathord{.}\Varid{mput}_\mathrm{L}\;(\Varid{b},\Varid{s}_{1});{}\<[E]%
\\
\>[35]{}\Varid{return}\;(\Varid{a},(\Varid{s}_{1}',\Varid{s}_{2}'))\mskip1.5mu\}{}\<[E]%
\\
\>[B]{}\hsindent{9}{}\<[9]%
\>[9]{}\Varid{missing}{}\<[27]%
\>[27]{}\mathrel{=}(\Varid{sl}_{1}\mathord{.}\Varid{missing},\Varid{sl}_{2}\mathord{.}\Varid{missing}){}\<[E]%
\ColumnHook
\end{hscode}\resethooks
which seems to be the obvious generalisation of pure symmetric lens
composition to the monadic case.  However, it does not always preserve
well-behavedness.
\begin{example}\label{ex:counterexample2}
Consider the following construction:
  \begin{hscode}\SaveRestoreHook
\column{B}{@{}>{\hspre}l<{\hspost}@{}}%
\column{5}{@{}>{\hspre}l<{\hspost}@{}}%
\column{16}{@{}>{\hspre}c<{\hspost}@{}}%
\column{16E}{@{}l@{}}%
\column{20}{@{}>{\hspre}l<{\hspost}@{}}%
\column{E}{@{}>{\hspre}l<{\hspost}@{}}%
\>[5]{}\Varid{setBool}{}\<[16]%
\>[16]{}\mathbin{::}{}\<[16E]%
\>[20]{}\Conid{Bool}\to \monadic{()\mathbin{\overset{()}{\longleftrightarrow}}()}{\Conid{State}\;\Conid{Bool}}{}\<[E]%
\\
\>[5]{}\Varid{setBool}\;\Varid{b}{}\<[16]%
\>[16]{}\mathrel{=}{}\<[16E]%
\>[20]{}\Conid{SMLens}\;\Varid{m}\;\Varid{m}\;()\;\mathbf{where}\;\Varid{m}\;\anonymous \mathrel{=}\mathbf{do}\;\{\mskip1.5mu \set\;\Varid{b};\Varid{return}\;((),())\mskip1.5mu\}{}\<[E]%
\ColumnHook
\end{hscode}\resethooks
  The lens \ensuremath{\Varid{setBool}\;\Conid{True}} has no effect on the complement or values,
  but sets the state to \ensuremath{\Conid{True}}.  Both \ensuremath{\Varid{setBool}\;\Conid{True}} and
  \ensuremath{\Varid{setBool}\;\Conid{False}} are well-behaved, but their composition (in either
  direction) is not: \ensuremath{\mathsf{(PutRLM)}} fails for
  \ensuremath{\Varid{setBool}\;\Conid{True};\Varid{setBool}\;\Conid{False}}
  because \ensuremath{\Varid{setBool}\;\Conid{True}} and \ensuremath{\Varid{setBool}\;\Conid{False}} share a single \ensuremath{\Conid{Bool}}
  state value.
\end{example}

\begin{proposition}\label{prop:setBool-wb-comp-not-wb}
  \ensuremath{\Varid{setBool}\;\Varid{b}} is well-behaved for \ensuremath{\Varid{b}\in\{\mskip1.5mu \Conid{True},\Conid{False}\mskip1.5mu\}}, but 
  \ensuremath{\Varid{setBool}\;\Conid{True}\mathbin{;}\Varid{setBool}\;\Conid{False}} is not well-behaved. 
\end{proposition}

Composition does preserve well-behavedness for commutative monads,
i.e. those for which 
\begin{hscode}\SaveRestoreHook
\column{B}{@{}>{\hspre}l<{\hspost}@{}}%
\column{E}{@{}>{\hspre}l<{\hspost}@{}}%
\>[B]{}\mathbf{do}\;\{\mskip1.5mu \Varid{a}\leftarrow \Varid{x};\Varid{b}\leftarrow \Varid{y};\Varid{return}\;(\Varid{a},\Varid{b})\mskip1.5mu\}\mathrel{=}\mathbf{do}\;\{\mskip1.5mu \Varid{b}\leftarrow \Varid{y};\Varid{a}\leftarrow \Varid{x};\Varid{return}\;(\Varid{a},\Varid{b})\mskip1.5mu\}{}\<[E]%
\ColumnHook
\end{hscode}\resethooks
but this rules out many interesting
monads, such as \ensuremath{\Conid{State}} and \ensuremath{\Conid{IO}}.






\subsection{Entangled state monads}


The types of the \ensuremath{\Varid{mput}_\mathrm{R}} and \ensuremath{\Varid{mput}_\mathrm{L}} operations of symmetric lenses
can be seen (modulo mild reordering) as stateful operations in the
\emph{state monad} \ensuremath{\Conid{State}\;\gamma\;\alpha\mathrel{=}\gamma\to (\alpha,\gamma)}, where the state
\ensuremath{\gamma\mathrel{=}\Conid{C}}.  This observation was also anticipated by Hofmann et al.
In a sequence of papers, we considered generalising these operations and their
laws to an arbitrary monad~\citep{cheney14bx2,abousaleh15mpc,abousaleh15bx}.  
In our
initial workshop paper, we proposed the following definition:
\begin{hscode}\SaveRestoreHook
\column{B}{@{}>{\hspre}l<{\hspost}@{}}%
\column{3}{@{}>{\hspre}l<{\hspost}@{}}%
\column{35}{@{}>{\hspre}l<{\hspost}@{}}%
\column{E}{@{}>{\hspre}l<{\hspost}@{}}%
\>[3]{}\mathbf{data}\;\monadic{\alpha\britishrail\beta}{\mu}\mathrel{=}\Conid{SetBX}\;\{\mskip1.5mu {}\<[35]%
\>[35]{}\get{L}\mathbin{::}\mu\;\alpha,\set{L}\mathbin{::}\alpha\to \mu\;(),{}\<[E]%
\\
\>[35]{}\get{R}\mathbin{::}\mu\;\beta,\set{R}\mathbin{::}\beta\to \mu\;()\mskip1.5mu\}{}\<[E]%
\ColumnHook
\end{hscode}\resethooks
subject to a subset of the  \ensuremath{\Conid{State}} monad laws~\citep{Plotkin&Power2002:Notions}, such as:
\begin{hscode}\SaveRestoreHook
\column{B}{@{}>{\hspre}l<{\hspost}@{}}%
\column{3}{@{}>{\hspre}l<{\hspost}@{}}%
\column{14}{@{}>{\hspre}l<{\hspost}@{}}%
\column{41}{@{}>{\hspre}l<{\hspost}@{}}%
\column{E}{@{}>{\hspre}l<{\hspost}@{}}%
\>[3]{}\mathsf{(Get_LSet_L)}{}\<[14]%
\>[14]{}\quad\mathbf{do}\;\{\mskip1.5mu \Varid{a}\leftarrow \get{L};\set{L}\;\Varid{a}\mskip1.5mu\}{}\<[41]%
\>[41]{}\mathrel{=}\Varid{return}\;(){}\<[E]%
\\
\>[3]{}\mathsf{(Set_LGet_L)}{}\<[14]%
\>[14]{}\quad\mathbf{do}\;\{\mskip1.5mu \set{L}\;\Varid{a};\get{L}\mskip1.5mu\}{}\<[41]%
\>[41]{}\mathrel{=}\mathbf{do}\;\{\mskip1.5mu \set{L}\;\Varid{a};\Varid{return}\;\Varid{a}\mskip1.5mu\}{}\<[E]%
\ColumnHook
\end{hscode}\resethooks
This presentation makes clear that
bidirectionality can be viewed as a state effect in which two ``views'' of
some common state are \emph{entangled}.  That is, rather than storing
a pair of views, each independently variable, they are entangled, in
the sense that a change to either may also change the other.
Accordingly, the entangled state monad operations do \emph{not}
satisfy all of the usual laws of state: for example, the \ensuremath{\set{L}} and
\ensuremath{\set{R}} operations do not commute.


However, one difficulty with the entangled state monad formalism is
that, as discussed in Section~\ref{sec:naive-lens}, effectful
\ensuremath{\Varid{mget}} operations cause problems for composition.  It turned out to be nontrivial
to define a satisfactory notion of composition, even for the well-behaved
special case where \ensuremath{\mu\mathrel{=}\Conid{StateT}\;\sigma\;\nu} for some \ensuremath{\nu}, where \ensuremath{\Conid{StateT}} is
the \emph{state monad transformer}~\citep{liang95popl}, i.e. \ensuremath{\Conid{StateT}\;\sigma\;\nu\;\alpha\mathrel{=}\sigma\to \nu\;(\alpha,\sigma)}.
We formulated the definition of \emph{monadic lenses} given earlier in
this paper in the process of exploring this design space.



\subsection{Spans of monadic lenses}

\citet{symlens} showed that a symmetric lens is
equivalent to a \emph{span} of two ordinary lenses, and later work by
\citet{johnson14bx} investigated such \emph{spans
  of lenses} in greater depth.  Accordingly, we propose the following
definition:

\begin{definition}[Monadic lens spans]
  A span of monadic lenses (``\ensuremath{\Conid{M}}-lens span'') is a pair of \ensuremath{\Conid{M}}-lenses having the same source:
\begin{hscode}\SaveRestoreHook
\column{B}{@{}>{\hspre}l<{\hspost}@{}}%
\column{5}{@{}>{\hspre}l<{\hspost}@{}}%
\column{73}{@{}>{\hspre}l<{\hspost}@{}}%
\column{E}{@{}>{\hspre}l<{\hspost}@{}}%
\>[5]{}\mathbf{type}\;\monadic{\alpha\mathbin{\reflectbox{$\rightsquigarrow$}}\sigma\rightsquigarrow\beta}{\mu}\mathrel{=}\Conid{Span}\;\{\mskip1.5mu \Varid{left}\mathbin{::}\monadic{\sigma\rightsquigarrow\alpha}{\mu},\Varid{right}{}\<[73]%
\>[73]{}\mathbin{::}\monadic{\sigma\rightsquigarrow\beta}{\mu}\mskip1.5mu\}{}\<[E]%
\ColumnHook
\end{hscode}\resethooks
We say that an \ensuremath{\Conid{M}}-lens span is well-behaved if both of its components are.
\end{definition}
We first note that we can extend either leg of a span with a monadic lens (preserving well-behavedness if the arguments are well-behaved):
\begin{hscode}\SaveRestoreHook
\column{B}{@{}>{\hspre}l<{\hspost}@{}}%
\column{3}{@{}>{\hspre}l<{\hspost}@{}}%
\column{13}{@{}>{\hspre}l<{\hspost}@{}}%
\column{E}{@{}>{\hspre}l<{\hspost}@{}}%
\>[3]{}(\triangleleft){}\<[13]%
\>[13]{}\mathbin{::}\Conid{Monad}\;\mu\Rightarrow \monadic{\alpha_1\rightsquigarrow\alpha_2}{\mu}\to \monadic{\alpha_1\mathbin{\reflectbox{$\rightsquigarrow$}}\sigma\rightsquigarrow\beta}{\mu}\to \monadic{\alpha_2\mathbin{\reflectbox{$\rightsquigarrow$}}\sigma\rightsquigarrow\beta}{\mu}{}\<[E]%
\\
\>[3]{}\Varid{ml}\triangleleft\Varid{sp}{}\<[13]%
\>[13]{}\mathrel{=}\Conid{Span}\;(\Varid{sp}\mathord{.}\Varid{left}\mathbin{;}\Varid{ml})\;(\Varid{sp}\mathord{.}\Varid{right}){}\<[E]%
\\
\>[3]{}(\triangleright){}\<[13]%
\>[13]{}\mathbin{::}\Conid{Monad}\;\mu\Rightarrow \monadic{\alpha\mathbin{\reflectbox{$\rightsquigarrow$}}\sigma\rightsquigarrow\beta_1}{\mu}\to \monadic{\beta_1\rightsquigarrow\beta_2}{\mu}\to \monadic{\alpha\mathbin{\reflectbox{$\rightsquigarrow$}}\sigma\rightsquigarrow\beta_2}{\mu}{}\<[E]%
\\
\>[3]{}\Varid{sp}\triangleright\Varid{ml}{}\<[13]%
\>[13]{}\mathrel{=}\Conid{Span}\;\Varid{sp}\mathord{.}\Varid{left}\;(\Varid{sp}\mathord{.}\Varid{right}\mathbin{;}\Varid{ml}){}\<[E]%
\ColumnHook
\end{hscode}\resethooks

To define composition, the basic idea is as follows.  Given two spans
\ensuremath{\monadic{\Conid{A}\mathbin{\reflectbox{$\rightsquigarrow$}}\Conid{S}_{1}\rightsquigarrow\Conid{B}}{\Conid{M}}} and \ensuremath{\monadic{\Conid{B}\mathbin{\reflectbox{$\rightsquigarrow$}}\Conid{S}_{2}\rightsquigarrow\Conid{C}}{\Conid{M}}}
with a common type \ensuremath{\Conid{B}} ``in the middle'', we want to form a single
span from \ensuremath{\Conid{A}} to \ensuremath{\Conid{C}}.  The obvious thing to try is to form a pullback
of the two monadic lenses from \ensuremath{\Conid{S}_{1}} and \ensuremath{\Conid{S}_{2}} to the common type \ensuremath{\Conid{B}},
obtaining a span from some common state type \ensuremath{\Conid{S}} to the state types
\ensuremath{\Conid{S}_{1}} and \ensuremath{\Conid{S}_{2}}, and composing with the outer legs.  (See Figure~\ref{fig:span-construction}.)
However, the category of monadic lenses doesn't have pullbacks (as Johnson and Rosebrugh note, this is already the case for ordinary lenses).  Instead, we construct the appropriate span as follows.  
\begin{hscode}\SaveRestoreHook
\column{B}{@{}>{\hspre}l<{\hspost}@{}}%
\column{3}{@{}>{\hspre}l<{\hspost}@{}}%
\column{5}{@{}>{\hspre}l<{\hspost}@{}}%
\column{23}{@{}>{\hspre}c<{\hspost}@{}}%
\column{23E}{@{}l@{}}%
\column{26}{@{}>{\hspre}l<{\hspost}@{}}%
\column{E}{@{}>{\hspre}l<{\hspost}@{}}%
\>[3]{}({\Join})\mathbin{::}\Conid{Monad}\;\mu\Rightarrow \monadic{\sigma_1\rightsquigarrow\beta}{\mu}\to \monadic{\sigma_2\rightsquigarrow\beta}{\mu}\to \monadic{\sigma_1\mathbin{\reflectbox{$\rightsquigarrow$}}(\sigma_1\mathbin{\!\Join\!}\sigma_2)\rightsquigarrow\sigma_2}{\mu}{}\<[E]%
\\
\>[3]{}\Varid{l}_{1}\mathbin{\Join}\Varid{l}_{2}\mathrel{=}\Conid{Span}\;(\Conid{MLens}\;\Varid{fst}\;\Varid{put}_\mathrm{L}\;\Varid{create}_\mathrm{L})\;(\Conid{MLens}\;\Varid{snd}\;\Varid{put}_\mathrm{R}\;\Varid{create}_\mathrm{R})\;\mathbf{where}{}\<[E]%
\\
\>[3]{}\hsindent{2}{}\<[5]%
\>[5]{}\Varid{put}_\mathrm{L}\;(\anonymous ,\Varid{s}_{2})\;\Varid{s}_{1}'{}\<[23]%
\>[23]{}\mathrel{=}{}\<[23E]%
\>[26]{}\mathbf{do}\;\{\mskip1.5mu \Varid{s}_{2}'\leftarrow \Varid{l}_{2}\mathord{.}\Varid{mput}\;\Varid{s}_{2}\;(\Varid{l}_{1}\mathord{.}\Varid{mget}\;\Varid{s}_{1}');\Varid{return}\;(\Varid{s}_{1}',\Varid{s}_{2}')\mskip1.5mu\}{}\<[E]%
\\
\>[3]{}\hsindent{2}{}\<[5]%
\>[5]{}\Varid{create}_\mathrm{L}\;\Varid{s}_{1}{}\<[23]%
\>[23]{}\mathrel{=}{}\<[23E]%
\>[26]{}\mathbf{do}\;\{\mskip1.5mu \Varid{s}_{2}'\leftarrow \Varid{l}_{2}\mathord{.}\Varid{mcreate}\;(\Varid{l}_{1}\mathord{.}\Varid{mget}\;\Varid{s}_{1});\Varid{return}\;(\Varid{s}_{1},\Varid{s}_{2}')\mskip1.5mu\}{}\<[E]%
\\
\>[3]{}\hsindent{2}{}\<[5]%
\>[5]{}\Varid{put}_\mathrm{R}\;(\Varid{s}_{1},\anonymous )\;\Varid{s}_{2}'{}\<[23]%
\>[23]{}\mathrel{=}{}\<[23E]%
\>[26]{}\mathbf{do}\;\{\mskip1.5mu \Varid{s}_{1}'\leftarrow \Varid{l}_{1}\mathord{.}\Varid{mput}\;\Varid{s}_{1}\;(\Varid{l}_{2}\mathord{.}\Varid{mget}\;\Varid{s}_{2}');\Varid{return}\;(\Varid{s}_{1}',\Varid{s}_{2}')\mskip1.5mu\}{}\<[E]%
\\
\>[3]{}\hsindent{2}{}\<[5]%
\>[5]{}\Varid{create}_\mathrm{R}\;\Varid{s}_{1}{}\<[23]%
\>[23]{}\mathrel{=}{}\<[23E]%
\>[26]{}\mathbf{do}\;\{\mskip1.5mu \Varid{s}_{1}'\leftarrow \Varid{l}_{1}\mathord{.}\Varid{mcreate}\;(\Varid{l}_{2}\mathord{.}\Varid{mget}\;\Varid{s}_{2});\Varid{return}\;(\Varid{s}_{1}',\Varid{s}_{2})\mskip1.5mu\}{}\<[E]%
\ColumnHook
\end{hscode}\resethooks
where we write \ensuremath{\Conid{S}_{1}\mathbin{\!\Join\!}\Conid{S}_{2}} for the type of \emph{consistent} state
pairs \ensuremath{\{\mskip1.5mu (\Varid{s}_{1},\Varid{s}_{2})\in\Conid{S}_{1}\times\Conid{S}_{2}\mid \Varid{l}_{1}\mathord{.}\Varid{mget}\;(\Varid{s}_{1})\mathrel{=}\Varid{l}_{2}\mathord{.}\Varid{mget}\;(\Varid{s}_{2})\mskip1.5mu\}}.  In the absence of
dependent types, we represent this type as \ensuremath{(\Conid{S}_{1},\Conid{S}_{2})} in Haskell, and
we need to check that the \ensuremath{\Varid{mput}} and \ensuremath{\Varid{mcreate}} operations respect the consistency invariant.

\begin{lemma}\label{lem:cospan2span}
  If  \ensuremath{\Varid{ml}_{\mathrm{1}}\mathbin{::}\monadic{\Conid{S}_{1}\rightsquigarrow\Conid{B}}{\Conid{M}}} and \ensuremath{\Varid{ml}_{\mathrm{2}}\mathbin{::}\monadic{\Conid{S}_{2}\rightsquigarrow\Conid{B}}{\Conid{M}}} are well-behaved then so is \ensuremath{\Varid{ml}_{\mathrm{1}}\mathbin{\Join}\Varid{ml}_{\mathrm{2}}\mathbin{::}\monadic{\Conid{S}_{1}\mathbin{\reflectbox{$\rightsquigarrow$}}(\Conid{S}_{1}\mathbin{\!\Join\!}\Conid{S}_{2})\rightsquigarrow\Conid{S}_{2}}{\mu}}.
\end{lemma} 
Note that \ensuremath{\mathsf{(MPutGet)}} and \ensuremath{\mathsf{(MCreateGet)}} hold by construction and do not need the corresponding properties for \ensuremath{\Varid{l}_{1}} and \ensuremath{\Varid{l}_{2}}, but these properties are needed to show that consistency is established by \ensuremath{\Varid{mcreate}} and preserved by \ensuremath{\Varid{mput}}.

We can now define composition as follows:
\begin{hscode}\SaveRestoreHook
\column{B}{@{}>{\hspre}l<{\hspost}@{}}%
\column{41}{@{}>{\hspre}l<{\hspost}@{}}%
\column{89}{@{}>{\hspre}l<{\hspost}@{}}%
\column{E}{@{}>{\hspre}l<{\hspost}@{}}%
\>[B]{}(\mathbin{;})\mathbin{::}\Conid{Monad}\;\mu\Rightarrow \monadic{\alpha\mathbin{\reflectbox{$\rightsquigarrow$}}\sigma_1\rightsquigarrow\beta}{\mu}\to \monadic{\beta\mathbin{\reflectbox{$\rightsquigarrow$}}\sigma_2\rightsquigarrow\gamma}{\mu}\to {}\<[89]%
\>[89]{}\monadic{\alpha\mathbin{\reflectbox{$\rightsquigarrow$}}(\sigma_1\mathbin{\!\Join\!}\sigma_2)\rightsquigarrow\gamma}{\mu}{}\<[E]%
\\
\>[B]{}\Varid{sp}_{\mathrm{1}}\mathbin{;}\Varid{sp}_{\mathrm{2}}\mathrel{=}\Varid{sp}_{\mathrm{1}}\mathord{.}\Varid{left}{}\<[41]%
\>[41]{}\triangleleft(\Varid{sp}_{\mathrm{1}}\mathord{.}\Varid{right}\mathbin{\Join}\Varid{sp}_{\mathrm{2}}\mathord{.}\Varid{left})\triangleright\Varid{sp}_{\mathrm{2}}\mathord{.}\Varid{right}{}\<[E]%
\ColumnHook
\end{hscode}\resethooks

\begin{figure}[tb]
  \centering
  \includegraphics{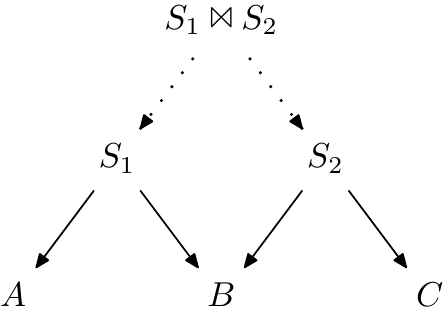}
  \caption{Composing spans of lenses}
  \label{fig:span-construction}
\end{figure}

The well-behavedness of the composition of two well-behaved spans is
immediate because \ensuremath{\triangleleft} and \ensuremath{\triangleright} preserve well-behavedness of their arguments:
\begin{theorem}
  If \ensuremath{\Varid{sp}_{\mathrm{1}}\mathbin{::}\monadic{\Conid{A}\mathbin{\reflectbox{$\rightsquigarrow$}}\Conid{S}_{1}\rightsquigarrow\Conid{B}}{\Conid{M}}} and \ensuremath{\Varid{sp}_{\mathrm{2}}\mathbin{::}\monadic{\Conid{B}\mathbin{\reflectbox{$\rightsquigarrow$}}\Conid{S}_{2}\rightsquigarrow\Conid{C}}{\Conid{M}}}
  are well-behaved spans of monadic lenses, then their composition
  \ensuremath{\Varid{sp}_{\mathrm{1}}\mathbin{;}\Varid{sp}_{\mathrm{2}}} is well-behaved.
\end{theorem}








Given a span of monadic lenses \ensuremath{\Varid{sp}\mathbin{::}\monadic{\Conid{A}\mathbin{\reflectbox{$\rightsquigarrow$}}\Conid{S}\rightsquigarrow\Conid{B}}{\Conid{M}}},
we can  construct a monadic symmetric lens \ensuremath{\Varid{sl}\mathbin{::}\monadic{\Conid{A}\mathbin{\overset{\Conid{Maybe}\;\Conid{S}}{\longleftrightarrow}}\Conid{B}}{\Conid{M}}}  as follows:

\begin{hscode}\SaveRestoreHook
\column{B}{@{}>{\hspre}l<{\hspost}@{}}%
\column{3}{@{}>{\hspre}l<{\hspost}@{}}%
\column{22}{@{}>{\hspre}c<{\hspost}@{}}%
\column{22E}{@{}l@{}}%
\column{25}{@{}>{\hspre}l<{\hspost}@{}}%
\column{29}{@{}>{\hspre}l<{\hspost}@{}}%
\column{E}{@{}>{\hspre}l<{\hspost}@{}}%
\>[B]{}\Varid{span2smlens}\;(\Varid{left},\Varid{right})\mathrel{=}\Conid{SMLens}\;\Varid{mput}_\mathrm{R}\;\Varid{mput}_\mathrm{L}\;\Conid{Nothing}\;\mathbf{where}{}\<[E]%
\\
\>[B]{}\hsindent{3}{}\<[3]%
\>[3]{}\Varid{mput}_\mathrm{R}\;(\Varid{a},\Conid{Just}\;\Varid{s}){}\<[22]%
\>[22]{}\mathrel{=}{}\<[22E]%
\>[25]{}\mathbf{do}\;{}\<[29]%
\>[29]{}\{\mskip1.5mu \Varid{s'}\leftarrow \Varid{left}\mathord{.}\Varid{mput}\;\Varid{s}\;\Varid{a};\Varid{return}\;(\Varid{right}\mathord{.}\Varid{mget}\;\Varid{s'},\Conid{Just}\;\Varid{s'})\mskip1.5mu\}{}\<[E]%
\\
\>[B]{}\hsindent{3}{}\<[3]%
\>[3]{}\Varid{mput}_\mathrm{R}\;(\Varid{a},\Conid{Nothing}){}\<[22]%
\>[22]{}\mathrel{=}{}\<[22E]%
\>[25]{}\mathbf{do}\;{}\<[29]%
\>[29]{}\{\mskip1.5mu \Varid{s'}\leftarrow \Varid{left}\mathord{.}\Varid{mcreate}\;\Varid{a};\Varid{return}\;(\Varid{right}\mathord{.}\Varid{mget}\;\Varid{s'},\Conid{Just}\;\Varid{s'})\mskip1.5mu\}{}\<[E]%
\\
\>[B]{}\hsindent{3}{}\<[3]%
\>[3]{}\Varid{mput}_\mathrm{L}\;(\Varid{b},\Conid{Just}\;\Varid{s}){}\<[22]%
\>[22]{}\mathrel{=}{}\<[22E]%
\>[25]{}\mathbf{do}\;{}\<[29]%
\>[29]{}\{\mskip1.5mu \Varid{s'}\leftarrow \Varid{right}\mathord{.}\Varid{mput}\;\Varid{s}\;\Varid{b};\Varid{return}\;(\Varid{left}\mathord{.}\Varid{mget}\;\Varid{s'},\Conid{Just}\;\Varid{s'})\mskip1.5mu\}{}\<[E]%
\\
\>[B]{}\hsindent{3}{}\<[3]%
\>[3]{}\Varid{mput}_\mathrm{L}\;(\Varid{b},\Conid{Nothing}){}\<[22]%
\>[22]{}\mathrel{=}{}\<[22E]%
\>[25]{}\mathbf{do}\;{}\<[29]%
\>[29]{}\{\mskip1.5mu \Varid{s'}\leftarrow \Varid{right}\mathord{.}\Varid{mcreate}\;\Varid{b};\Varid{return}\;(\Varid{left}\mathord{.}\Varid{mget}\;\Varid{s'},\Conid{Just}\;\Varid{s'})\mskip1.5mu\}{}\<[E]%
\ColumnHook
\end{hscode}\resethooks
Essentially, these operations 
use the span's \ensuremath{\Varid{mput}} and \ensuremath{\Varid{mget}} operations to update one side and
obtain the new view value for the other side, and use the \ensuremath{\Varid{mcreate}}
operations to build the initial \ensuremath{\Conid{S}} state if the complement is
\ensuremath{\Conid{Nothing}}.

Well-behavedness is preserved by the conversion from  monadic lens spans to
\ensuremath{\Conid{SMLens}}, for arbitrary monads \ensuremath{\Conid{M}}:
\begin{theorem}\label{thm:span2smlens-wb}
If \ensuremath{\Varid{sp}\mathbin{::}\monadic{\Conid{A}\mathbin{\reflectbox{$\rightsquigarrow$}}\Conid{S}\rightsquigarrow\Conid{B}}{\Conid{M}}} is well-behaved, then \ensuremath{\Varid{span2smlens}\;\Varid{sp}} is also well-behaved.
\end{theorem}





Given 
\ensuremath{\Varid{sl}\mathbin{::}\monadic{\Conid{A}\mathbin{\overset{\Conid{C}}{\longleftrightarrow}}\Conid{B}}{\Conid{M}}},
let $S \subseteq A \times B \times C$ be the set of \emph{consistent triples}
\ensuremath{(\Varid{a},\Varid{b},\Varid{c})}, that is, those for which \ensuremath{\Varid{sl}\mathord{.}\Varid{mput}_\mathrm{R}\;(\Varid{a},\Varid{c})\mathrel{=}\Varid{return}\;(\Varid{b},\Varid{c})} and \ensuremath{\Varid{sl}\mathord{.}\Varid{mput}_\mathrm{L}\;(\Varid{b},\Varid{c})\mathrel{=}\Varid{return}\;(\Varid{a},\Varid{c})}.
We 
construct \ensuremath{\Varid{sp}\mathbin{::}\monadic{\Conid{A}\mathbin{\reflectbox{$\rightsquigarrow$}}\Conid{S}\rightsquigarrow\Conid{B}}{\Conid{M}}} by
\begin{hscode}\SaveRestoreHook
\column{B}{@{}>{\hspre}l<{\hspost}@{}}%
\column{3}{@{}>{\hspre}l<{\hspost}@{}}%
\column{26}{@{}>{\hspre}l<{\hspost}@{}}%
\column{43}{@{}>{\hspre}l<{\hspost}@{}}%
\column{E}{@{}>{\hspre}l<{\hspost}@{}}%
\>[B]{}\Varid{smlens2span}\;\Varid{sl}\mathrel{=}\Conid{Span}\;(\Conid{MLens}\;\get{L}\;\Varid{put}_\mathrm{L}\;\Varid{create}_\mathrm{L})\;(\Conid{MLens}\;\get{R}\;\Varid{put}_\mathrm{R}\;\Varid{create}_\mathrm{R}){}\<[E]%
\\
\>[B]{}\mathbf{where}{}\<[E]%
\\
\>[B]{}\hsindent{3}{}\<[3]%
\>[3]{}\get{L}\;(\Varid{a},\Varid{b},\Varid{c}){}\<[26]%
\>[26]{}\mathrel{=}\Varid{a}{}\<[E]%
\\
\>[B]{}\hsindent{3}{}\<[3]%
\>[3]{}\Varid{put}_\mathrm{L}\;(\Varid{a},\Varid{b},\Varid{c})\;\Varid{a'}{}\<[26]%
\>[26]{}\mathrel{=}\mathbf{do}\;\{\mskip1.5mu (\Varid{b'},\Varid{c'})\leftarrow \Varid{sl}\mathord{.}\Varid{mput}_\mathrm{R}\;(\Varid{a'},\Varid{c});\Varid{return}\;(\Varid{a'},\Varid{b'},\Varid{c'})\mskip1.5mu\}{}\<[E]%
\\
\>[B]{}\hsindent{3}{}\<[3]%
\>[3]{}\Varid{create}_\mathrm{L}\;\Varid{a}{}\<[26]%
\>[26]{}\mathrel{=}\mathbf{do}\;\{\mskip1.5mu (\Varid{b},\Varid{c})\leftarrow {}\<[43]%
\>[43]{}\Varid{sl}\mathord{.}\Varid{mput}_\mathrm{R}\;(\Varid{a},\Varid{sl}\mathord{.}\Varid{missing});\Varid{return}\;(\Varid{a},\Varid{b},\Varid{c})\mskip1.5mu\}{}\<[E]%
\\
\>[B]{}\hsindent{3}{}\<[3]%
\>[3]{}\get{R}\;(\Varid{a},\Varid{b},\Varid{c}){}\<[26]%
\>[26]{}\mathrel{=}\Varid{b}{}\<[E]%
\\
\>[B]{}\hsindent{3}{}\<[3]%
\>[3]{}\Varid{put}_\mathrm{R}\;(\Varid{a},\Varid{b},\Varid{c})\;\Varid{b'}{}\<[26]%
\>[26]{}\mathrel{=}\mathbf{do}\;\{\mskip1.5mu (\Varid{a'},\Varid{c'})\leftarrow \Varid{sl}\mathord{.}\Varid{mput}_\mathrm{L}\;(\Varid{b'},\Varid{c});\Varid{return}\;(\Varid{a'},\Varid{b'},\Varid{c'})\mskip1.5mu\}{}\<[E]%
\\
\>[B]{}\hsindent{3}{}\<[3]%
\>[3]{}\Varid{create}_\mathrm{R}\;\Varid{b}{}\<[26]%
\>[26]{}\mathrel{=}\mathbf{do}\;\{\mskip1.5mu (\Varid{a},\Varid{c})\leftarrow \Varid{sl}\mathord{.}\Varid{mput}_\mathrm{L}\;(\Varid{b},\Varid{sl}\mathord{.}\Varid{missing});\Varid{return}\;(\Varid{a},\Varid{b},\Varid{c})\mskip1.5mu\}{}\<[E]%
\ColumnHook
\end{hscode}\resethooks

However, \ensuremath{\Varid{smlens2span}} may not preserve well-behavedness even for
simple monads such as \ensuremath{\Conid{Maybe}}, as the following counterexample
illustrates.
\begin{example}\label{ex:counterexamples}
  Consider the following monadic symmetric lens construction:
\begin{hscode}\SaveRestoreHook
\column{B}{@{}>{\hspre}l<{\hspost}@{}}%
\column{E}{@{}>{\hspre}l<{\hspost}@{}}%
\>[B]{}\Varid{fail}\mathbin{::}\monadic{()\mathbin{\overset{()}{\longleftrightarrow}}()}{\Conid{Maybe}}{}\<[E]%
\\
\>[B]{}\Varid{fail}\mathrel{=}\Conid{SMLens}\;\Conid{Nothing}\;\Conid{Nothing}\;(){}\<[E]%
\ColumnHook
\end{hscode}\resethooks


This is well-behaved but \ensuremath{\Varid{smlens2span}\;\Varid{fail}} is not.  In fact, the set
of consistent states of \ensuremath{\Varid{fail}} is empty, and each leg of the
induced span is of the following form:
\begin{hscode}\SaveRestoreHook
\column{B}{@{}>{\hspre}l<{\hspost}@{}}%
\column{E}{@{}>{\hspre}l<{\hspost}@{}}%
\>[B]{}\Varid{failMLens}\mathbin{::}\Conid{MLens}\;\Conid{Maybe}\;\emptyset\;(){}\<[E]%
\\
\>[B]{}\Varid{failMLens}\mathrel{=}\Conid{MLens}\;(\lambda \anonymous \to ())\;(\lambda \anonymous \;()\to \Conid{Nothing})\;(\lambda \anonymous \to \Conid{Nothing}){}\<[E]%
\ColumnHook
\end{hscode}\resethooks
which fails to satisfy \ensuremath{\mathsf{(MGetPut)}}.
\end{example}

For pure symmetric lenses, \ensuremath{\Varid{smlens2span}} does preserve well-behavedness.
\begin{theorem}\label{thm:smlens2span-wb}
  If \ensuremath{\Varid{sl}\mathbin{::}\Conid{SMLens}\;\Conid{Id}\;\Conid{C}\;\Conid{A}\;\Conid{B}} is well-behaved,
  then \ensuremath{\Varid{smlens2span}\;\Varid{sl}} is also well-behaved, with state space \ensuremath{\Conid{S}}
  consisting of the consistent triples of \ensuremath{\Varid{sl}}.
\end{theorem}

To summarise: spans of monadic lenses are closed under composition,
and correspond to well-behaved symmetric monadic lenses.  However,
there are well-behaved symmetric monadic lenses that do not map to
well-behaved spans.  It seems to be an interesting open problem to give a
direct axiomatisation of the symmetric monadic lenses that are
essentially spans of monadic lenses (and are therefore closed under
composition).
 

\section{Equivalence of spans}\label{sec:equiv}

\citet{symlens} introduced a bisimulation-like notion of equivalence
for pure symmetric lenses, in order to validate laws such as identity,
associativity and congruence of composition.  \citet{johnson14bx}
introduced a definition of equivalence of spans and compared it with
symmetric lens equivalence.  We have considered equivalences based on
isomorphism~\citep{abousaleh15mpc} and
bisimulation~\citep{abousaleh15bx}.  In this section we consider and
relate these approaches in the context of spans of \ensuremath{\Conid{M}}-lenses.

\begin{definition}
  [Isomorphism Equivalence] Two \ensuremath{\Conid{M}}-lens spans \ensuremath{\Varid{sp}_{\mathrm{1}}\mathbin{::}\monadic{\Conid{A}\mathbin{\reflectbox{$\rightsquigarrow$}}\Conid{S}_{1}\rightsquigarrow\Conid{B}}{\Conid{M}}} and \ensuremath{\Varid{sp}_{\mathrm{2}}\mathbin{::}\monadic{\Conid{A}\mathbin{\reflectbox{$\rightsquigarrow$}}\Conid{S}_{2}\rightsquigarrow\Conid{B}}{\Conid{M}}} are isomorphic (\ensuremath{\Varid{sp}\equiv_{\mathrm{i}} \Varid{sp'}}) if there is an isomorphism \ensuremath{\Varid{h}\mathbin{::}\Conid{S}_{1}\to \Conid{S}_{2}} on their
  state spaces such that \ensuremath{\Varid{h}\mathbin{;}\Varid{sp}_{\mathrm{2}}\mathord{.}\Varid{left}\mathrel{=}\Varid{sp}_{\mathrm{1}}\mathord{.}\Varid{left}} and
  \ensuremath{\Varid{h}\mathbin{;}\Varid{sp}_{\mathrm{2}}\mathord{.}\Varid{right}\mathrel{=}\Varid{sp}_{\mathrm{1}}\mathord{.}\Varid{right}}.
\end{definition}
Note that any isomorphism \ensuremath{\Varid{h}\mathbin{::}\Conid{S}_{1}\to \Conid{S}_{2}} can be made into a (monadic)
lens; we omit the explicit conversion.

We consider a second definition of equivalence, inspired by \citet{johnson14bx}, which we call \emph{span equivalence}:
\begin{definition}
  [Span Equivalence] Two \ensuremath{\Conid{M}}-lens spans \ensuremath{\Varid{sp}_{\mathrm{1}}\mathbin{::}\monadic{\Conid{A}\mathbin{\reflectbox{$\rightsquigarrow$}}\Conid{S}_{1}\rightsquigarrow\Conid{B}}{\Conid{M}}}
  and \ensuremath{\Varid{sp}_{\mathrm{2}}\mathbin{::}\monadic{\Conid{A}\mathbin{\reflectbox{$\rightsquigarrow$}}\Conid{S}_{2}\rightsquigarrow\Conid{B}}{\Conid{M}}} are related by \ensuremath{\curvearrowright } if
  there is a full lens \ensuremath{\Varid{h}\mathbin{::}\Conid{S}_{1}\mathbin{\leadsto}\Conid{S}_{2}} such that \ensuremath{\Varid{h}\mathbin{;}\Varid{sp}_{\mathrm{2}}\mathord{.}\Varid{left}\mathrel{=}\Varid{sp}_{\mathrm{1}}\mathord{.}\Varid{left}} and \ensuremath{\Varid{h}\mathbin{;}\Varid{sp}_{\mathrm{2}}\mathord{.}\Varid{right}\mathrel{=}\Varid{sp}_{\mathrm{1}}\mathord{.}\Varid{right}}.  The
  equivalence relation \ensuremath{\equiv_{\mathrm{s}} } is the least equivalence relation
  containing \ensuremath{\curvearrowright }.
\end{definition}
One important consideration emphasised by Johnson and Rosebrugh is
the need to avoid making all compatible spans equivalent to the
``trivial'' span \ensuremath{\monadic{\Conid{A}\mathbin{\reflectbox{$\rightsquigarrow$}}\emptyset\rightsquigarrow\Conid{B}}{\Conid{M}}}.
To avoid this problem, they imposed conditions on \ensuremath{\Varid{h}}:
its \ensuremath{\Varid{get}} function must be surjective and \emph{split}, meaning that
there exists a function \ensuremath{\Varid{c}} such that \ensuremath{\Varid{h}\mathord{.}\Varid{get}\hsdot{\mathbin{\cdot}}{\mathrel{.}}\Varid{c}\mathrel{=}\Varid{id}}.  We chose
instead to require \ensuremath{\Varid{h}} to be a full lens.  This is actually slightly
stronger than Johnson and Rosebrugh's definition, at least from a
constructive perspective, because \ensuremath{\Varid{h}} is equipped with a specific
choice of \ensuremath{\Varid{c}\mathrel{=}\Varid{create}} satisfying \ensuremath{\Varid{h}\mathord{.}\Varid{get}\hsdot{\mathbin{\cdot}}{\mathrel{.}}\Varid{c}\mathrel{=}\Varid{id}}, that is, the
\ensuremath{\mathsf{(CreateGet)}} law.

We have defined span equivalence as the reflexive, symmetric,
transitive closure of \ensuremath{\curvearrowright }.  Interestingly, even though span
equivalence allows for an arbitrary sequence of (pure) lenses between the
respective state spaces, it suffices to consider only spans of
lenses.  To prove this, we first state a lemma about the
\ensuremath{({\Join})} operation used in composition.  Its proof is
straightforward equational reasoning.
\begin{lemma}\label{lem:cospan2span-pullback}
  Suppose \ensuremath{\Varid{l}_{1}\mathbin{::}\Conid{A}\mathbin{\leadsto}\Conid{B}} and \ensuremath{\Varid{l}_{2}\mathbin{::}\Conid{C}\mathbin{\leadsto}\Conid{B}} are pure lenses.  Then \ensuremath{(\Varid{l}_{1}\mathbin{\Join}\Varid{l}_{2})\mathord{.}\Varid{left}\mathbin{;}\Varid{l}_{1}\mathrel{=}(\Varid{l}_{1}\mathbin{\Join}\Varid{l}_{2})\mathord{.}\Varid{right}\mathbin{;}\Varid{l}_{2}}.
\end{lemma}

\begin{theorem}\label{thm:span-equivalence}
 Given \ensuremath{\Varid{sp}_{\mathrm{1}}\mathbin{::}\monadic{\Conid{A}\mathbin{\reflectbox{$\rightsquigarrow$}}\Conid{S}_{1}\rightsquigarrow\Conid{B}}{\Conid{M}}}
  and \ensuremath{\Varid{sp}_{\mathrm{2}}\mathbin{::}\monadic{\Conid{A}\mathbin{\reflectbox{$\rightsquigarrow$}}\Conid{S}_{2}\rightsquigarrow\Conid{B}}{\Conid{M}}}, if \ensuremath{\Varid{sp}_{\mathrm{1}}\equiv_{\mathrm{s}} \Varid{sp}_{\mathrm{2}}} then
  there exists \ensuremath{\Varid{sp}\mathbin{::}\Conid{S}_{1}\mathbin{{\reflectbox{$\rightsquigarrow$}}}\Conid{S}\mathbin{{\rightsquigarrow}}\Conid{S}_{2}} such
  that \ensuremath{\Varid{sp}\mathord{.}\Varid{left}\mathbin{;}\Varid{sp}_{\mathrm{1}}\mathord{.}\Varid{left}\mathrel{=}\Varid{sp}\mathord{.}\Varid{right}\mathbin{;}\Varid{sp}_{\mathrm{2}}\mathord{.}\Varid{left}} and \ensuremath{\Varid{sp}\mathord{.}\Varid{left}\mathbin{;}\Varid{sp}_{\mathrm{1}}\mathord{.}\Varid{right}\mathrel{=}\Varid{sp}\mathord{.}\Varid{right}\mathbin{;}\Varid{sp}_{\mathrm{2}}\mathord{.}\Varid{right}}.
\end{theorem}
\begin{proof}
  Let \ensuremath{\Varid{sp}_{\mathrm{1}}} and \ensuremath{\Varid{sp}_{\mathrm{2}}} be given such that \ensuremath{\Varid{sp}_{\mathrm{1}}\equiv_{\mathrm{s}} \Varid{sp}_{\mathrm{2}}}.  The
  proof is by induction on the length of the sequence of
  \ensuremath{\curvearrowright } or \ensuremath{\curvearrowleft } steps linking \ensuremath{\Varid{sp}_{\mathrm{1}}} to \ensuremath{\Varid{sp}_{\mathrm{2}}}.

  If \ensuremath{\Varid{sp}_{\mathrm{1}}\mathrel{=}\Varid{sp}_{\mathrm{2}}} then the result is immediate. If \ensuremath{\Varid{sp}_{\mathrm{1}}\curvearrowright \Varid{sp}_{\mathrm{2}}} then we can complete a span between \ensuremath{\Conid{S}_{1}} and \ensuremath{\Conid{S}_{2}} using
  the identity lens.  For the inductive case, suppose that the result
  holds for sequences of up to $n$ \ensuremath{\curvearrowright } or \ensuremath{\curvearrowleft } steps, and suppose
  \ensuremath{\Varid{sp}_{\mathrm{1}}\equiv_{\mathrm{s}} \Varid{sp}_{\mathrm{2}}} holds in $n$ 
\ensuremath{\curvearrowright } or \ensuremath{\curvearrowleft } steps.  There are two cases, depending on
the direction of the first step.  If \ensuremath{\Varid{sp}_{\mathrm{1}}\curvearrowleft \Varid{sp}_{\mathrm{3}}\equiv_{\mathrm{s}} \Varid{sp}_{\mathrm{2}}} then by induction we
  must have a pure span \ensuremath{\Varid{sp}} between \ensuremath{\Conid{S}_{3}} and \ensuremath{\Conid{S}_{2}} and \ensuremath{\Varid{sp}_{\mathrm{1}}\curvearrowleft \Varid{sp}_{\mathrm{3}}} holds by virtue of a lens \ensuremath{\Varid{h}\mathbin{::}\Conid{S}_{3}\to \Conid{S}_{1}}, so we
  can simply compose \ensuremath{\Varid{h}} with \ensuremath{\Varid{sp}\mathord{.}\Varid{left}} to obtain the required span
  between \ensuremath{\Conid{S}_{1}} and \ensuremath{\Conid{S}_{2}}.  Otherwise, if \ensuremath{\Varid{sp}_{\mathrm{1}}\curvearrowright \Varid{sp}_{\mathrm{3}}\equiv_{\mathrm{s}} \Varid{sp}_{\mathrm{2}}} then by induction we
  must have a pure span \ensuremath{\Varid{sp}} between \ensuremath{\Conid{S}_{3}} and \ensuremath{\Conid{S}_{2}} and we must have a lens \ensuremath{\Varid{h}\mathbin{::}\Conid{S}_{1}\to \Conid{S}_{3}}, so we use Lemma~\ref{lem:cospan2span-pullback} to form a span
  \ensuremath{\Varid{sp}_{\mathrm{0}}\mathbin{::}\Conid{S}_{1}\mathbin{{\reflectbox{$\rightsquigarrow$}}}(\Conid{S}_{1}\mathbin{\!\Join\!}\Conid{S}_{3})\mathbin{{\rightsquigarrow}}\Conid{S}_{3}} and extend \ensuremath{\Varid{sp}_{\mathrm{0}}\mathord{.}\Varid{right}} with
  \ensuremath{\Varid{sp}\mathord{.}\Varid{right}} to form the required span between \ensuremath{\Conid{S}_{1}} and \ensuremath{\Conid{S}_{3}}.
\end{proof}

Thus, span equivalence is a doubly appropriate name for \ensuremath{\equiv_{\mathrm{s}} }:
it is an equivalence of spans witnessed by a (pure) span.

Finally, we consider a third notion of equivalence, inspired by the natural
bisimulation equivalence for coalgebraic bx~\citep{abousaleh15bx}: 
\begin{definition}
  [Base map]
  Given \ensuremath{\Conid{M}}-lenses \ensuremath{\Varid{l}_{1}\mathbin{::}\monadic{\Conid{S}_{1}\rightsquigarrow\Conid{V}}{\Conid{M}}} and \ensuremath{\Varid{l}_{2}\mathbin{::}\monadic{\Conid{S}_{2}\rightsquigarrow\Conid{V}}{\Conid{M}}}, we
  say that \ensuremath{\Varid{h}\mathbin{:}\Conid{S}_{1}\to \Conid{S}_{2}} is a \emph{base map} from \ensuremath{\Varid{l}_{1}} to \ensuremath{\Varid{l}_{2}} if
  \begin{hscode}\SaveRestoreHook
\column{B}{@{}>{\hspre}l<{\hspost}@{}}%
\column{5}{@{}>{\hspre}l<{\hspost}@{}}%
\column{46}{@{}>{\hspre}l<{\hspost}@{}}%
\column{E}{@{}>{\hspre}l<{\hspost}@{}}%
\>[5]{}\Varid{l}_{1}\mathord{.}\Varid{mget}\;\Varid{s}{}\<[46]%
\>[46]{}\mathrel{=}\Varid{l}_{2}\mathord{.}\Varid{mget}\;(\Varid{h}\;\Varid{s}){}\<[E]%
\\
\>[5]{}\mathbf{do}\;\{\mskip1.5mu \Varid{s}\leftarrow \Varid{l}_{1}\mathord{.}\Varid{mput}\;\Varid{s}\;\Varid{v};\Varid{return}\;(\Varid{h}\;\Varid{s})\mskip1.5mu\}{}\<[46]%
\>[46]{}\mathrel{=}\Varid{l}_{2}\mathord{.}\Varid{mput}\;(\Varid{h}\;\Varid{s})\;\Varid{v}{}\<[E]%
\\
\>[5]{}\mathbf{do}\;\{\mskip1.5mu \Varid{s}\leftarrow \Varid{l}_{1}\mathord{.}\Varid{mcreate}\;\Varid{v};\Varid{return}\;(\Varid{h}\;\Varid{s})\mskip1.5mu\}{}\<[46]%
\>[46]{}\mathrel{=}\Varid{l}_{2}\mathord{.}\Varid{mcreate}\;\Varid{v}{}\<[E]%
\ColumnHook
\end{hscode}\resethooks
  Similarly, given two \ensuremath{\Conid{M}}-lens spans \ensuremath{\Varid{sp}_{\mathrm{1}}\mathbin{::}\monadic{\Conid{A}\mathbin{\reflectbox{$\rightsquigarrow$}}\Conid{S}_{1}\rightsquigarrow\Conid{B}}{\Conid{M}}} and
  \ensuremath{\Varid{sp}_{\mathrm{2}}\mathbin{::}\monadic{\Conid{A}\mathbin{\reflectbox{$\rightsquigarrow$}}\Conid{S}_{2}\rightsquigarrow\Conid{B}}{\Conid{M}}} we say that \ensuremath{\Varid{h}\mathbin{::}\Conid{S}_{1}\to \Conid{S}_{2}} is a base
  map from \ensuremath{\Varid{sp}_{\mathrm{1}}}
  to \ensuremath{\Varid{sp}_{\mathrm{2}}} if \ensuremath{\Varid{h}} is a base map from \ensuremath{\Varid{sp}_{\mathrm{1}}\mathord{.}\Varid{left}} to \ensuremath{\Varid{sp}_{\mathrm{2}}\mathord{.}\Varid{left}} and from
  \ensuremath{\Varid{sp}_{\mathrm{1}}\mathord{.}\Varid{right}} to \ensuremath{\Varid{sp}_{\mathrm{2}}\mathord{.}\Varid{right}}.
  \end{definition}

\begin{definition}[Bisimulation equivalence]
  A \emph{bisimulation} of \ensuremath{\Conid{M}}-lens spans \ensuremath{\Varid{sp}_{\mathrm{1}}\mathbin{::}\monadic{\Conid{A}\mathbin{\reflectbox{$\rightsquigarrow$}}\Conid{S}_{1}\rightsquigarrow\Conid{B}}{\Conid{M}}}
  and \ensuremath{\Varid{sp}_{\mathrm{2}}\mathbin{::}\monadic{\Conid{A}\mathbin{\reflectbox{$\rightsquigarrow$}}\Conid{S}_{2}\rightsquigarrow\Conid{B}}{\Conid{M}}} is a \ensuremath{\Conid{M}}-lens span \ensuremath{\Varid{sp}\mathbin{::}\monadic{\Conid{A}\mathbin{\reflectbox{$\rightsquigarrow$}}\Conid{R}\rightsquigarrow\Conid{B}}{\Conid{M}}} where \ensuremath{\Conid{R}\subseteq\Conid{S}_{1} \times \Conid{S}_{2}} and \ensuremath{\Varid{fst}} is a base map from \ensuremath{\Varid{sp}} to
  \ensuremath{\Varid{sp}_{\mathrm{1}}} and \ensuremath{\Varid{snd}} is a base map  from \ensuremath{\Varid{sp}} to \ensuremath{\Varid{sp}_{\mathrm{2}}}.  We write \ensuremath{\Varid{sp}_{\mathrm{1}}\equiv_{\mathrm{b}} \Varid{sp}_{\mathrm{2}}} when there is a bisimulation of spans \ensuremath{\Varid{sp}_{\mathrm{1}}} and
  \ensuremath{\Varid{sp}_{\mathrm{2}}}.
\end{definition}

Figure~\ref{fig:equivalences} illustrates the three equivalences diagrammatically.
\begin{figure}[tb]
  \centering
  \begin{tabular}{ccc}
\includegraphics[width=0.25\textwidth]{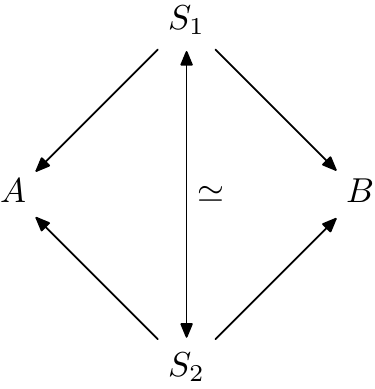}
&
\quad
\includegraphics[width=0.25\textwidth]{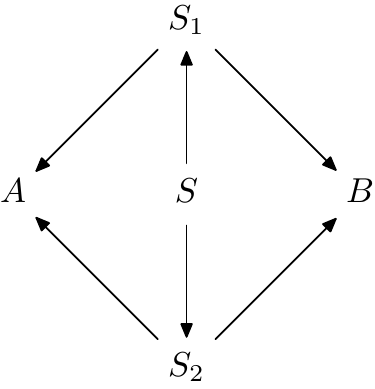}
&
\quad
\includegraphics[width=0.25\textwidth]{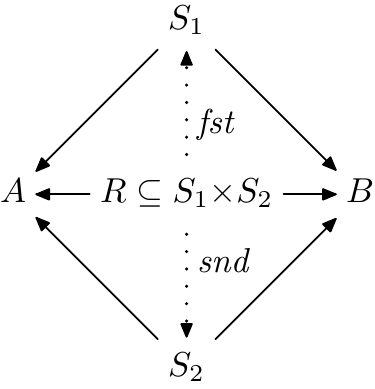}
\\
(a) &\quad
 (b) &\quad
 (c)
  \end{tabular}
  \caption{(a) Isomorphism equivalence \ensuremath{(\equiv_{\mathrm{i}} )}, (b) span
    equivalence \ensuremath{(\equiv_{\mathrm{s}} )}, and (c) bisimulation
    \ensuremath{(\equiv_{\mathrm{b}} )} equivalence.  
    In (c), the dotted arrows are base maps; all other
    arrows are (monadic) lenses.}
  \label{fig:equivalences}
\end{figure}
\begin{proposition}
  Each of the relations \ensuremath{\equiv_{\mathrm{i}} }, \ensuremath{\equiv_{\mathrm{s}} } and \ensuremath{\equiv_{\mathrm{b}} } are
  equivalence relations on compatible spans of \ensuremath{\Conid{M}}-lenses and satisfy
  \ensuremath{\mathsf{(Identity)}}, \ensuremath{\mathsf{(Assoc)}} and \ensuremath{\mathsf{(Cong)}}.
\end{proposition}

\begin{theorem}
  \ensuremath{\Varid{sp}_{\mathrm{1}}\equiv_{\mathrm{i}} \Varid{sp}_{\mathrm{2}}} implies \ensuremath{\Varid{sp}_{\mathrm{1}}\equiv_{\mathrm{s}} \Varid{sp}_{\mathrm{2}}}, but not the converse.
\end{theorem}
\begin{proof}
  The forward direction is obvious; for the reverse direction, consider
  \begin{hscode}\SaveRestoreHook
\column{B}{@{}>{\hspre}l<{\hspost}@{}}%
\column{5}{@{}>{\hspre}l<{\hspost}@{}}%
\column{E}{@{}>{\hspre}l<{\hspost}@{}}%
\>[5]{}\Varid{h}\mathbin{::}\Conid{Bool}\mathbin{\leadsto}(){}\<[E]%
\\
\>[5]{}\Varid{h}\mathrel{=}\Conid{Lens}\;(\lambda \anonymous \to ())\;(\lambda \Varid{a}\;()\to \Varid{a})\;(\lambda ()\to \Conid{True}){}\<[E]%
\\
\>[5]{}\Varid{sp}_{\mathrm{1}}\mathbin{::}\monadic{()\mathbin{\reflectbox{$\rightsquigarrow$}}()\rightsquigarrow()}{\mu}{}\<[E]%
\\
\>[5]{}\Varid{sp}_{\mathrm{1}}\mathrel{=}\Conid{Span}\;\Varid{idMLens}\;\Varid{idMLens}{}\<[E]%
\\
\>[5]{}\Varid{sp}_{\mathrm{2}}\mathrel{=}(\Varid{h}\mathbin{;}\Varid{sp}_{\mathrm{1}}\mathord{.}\Varid{left},\Varid{h}\mathbin{;}\Varid{sp}_{\mathrm{2}}\mathord{.}\Varid{right}){}\<[E]%
\ColumnHook
\end{hscode}\resethooks
Clearly \ensuremath{\Varid{sp}_{\mathrm{1}}\equiv_{\mathrm{s}} \Varid{sp}_{\mathrm{2}}} by definition and all three structures
are well-behaved, but \ensuremath{\Varid{h}} is not an isomorphism: any \ensuremath{\Varid{k}\mathbin{::}()\mathbin{\leadsto}\Conid{Bool}} must satisfy \ensuremath{\Varid{k}\mathord{.}\Varid{get}\;()\mathrel{=}\Conid{True}} or \ensuremath{\Varid{k}\mathord{.}\Varid{get}\;()\mathrel{=}\Conid{False}}, so \ensuremath{(\Varid{h}\mathbin{;}\Varid{k})\mathord{.}\Varid{get}\mathrel{=}\Varid{k}\mathord{.}\Varid{get}\hsdot{\mathbin{\cdot}}{\mathrel{.}}\Varid{h}\mathord{.}\Varid{get}} cannot be the identity function.
\end{proof}

\begin{theorem}\label{thm:jr-implies-bisim}
  Given \ensuremath{\Varid{sp}_{\mathrm{1}}\mathbin{::}\monadic{\Conid{A}\mathbin{\reflectbox{$\rightsquigarrow$}}\Conid{S}_{1}\rightsquigarrow\Conid{B}}{\Conid{M}},\Varid{sp}_{\mathrm{2}}\mathbin{::}\monadic{\Conid{A}\mathbin{\reflectbox{$\rightsquigarrow$}}\Conid{S}_{2}\rightsquigarrow\Conid{B}}{\Conid{M}}}, if \ensuremath{\Varid{sp}_{\mathrm{1}}\equiv_{\mathrm{s}} \Varid{sp}_{\mathrm{2}}} then \ensuremath{\Varid{sp}_{\mathrm{1}}\equiv_{\mathrm{b}} \Varid{sp}_{\mathrm{2}}}.
\end{theorem}
\begin{proof}
  For the forward direction, it suffices to show that a single \ensuremath{\Varid{sp}_{\mathrm{1}}\curvearrowright \Varid{sp}_{\mathrm{2}}} step implies \ensuremath{\Varid{sp}_{\mathrm{1}}\equiv_{\mathrm{b}} \Varid{sp}_{\mathrm{2}}}, which is
  straightforward by taking \ensuremath{\Conid{R}} to be the set of pairs \ensuremath{\{\mskip1.5mu (\Varid{s}_{1},\Varid{s}_{2})\mid \Varid{l}_{1}\mathord{.}\Varid{get}\;\Varid{s}_{1}\mathrel{=}\Varid{s}_{2}\mskip1.5mu\}}, and constructing an appropriate span \ensuremath{\Varid{sp}\mathbin{:}\Conid{A}\mathbin{{\reflectbox{$\rightsquigarrow$}}}\Conid{R}\mathbin{{\rightsquigarrow}}\Conid{B}}.  Since bisimulation equivalence is transitive, it
  follows that \ensuremath{\Varid{sp}_{\mathrm{1}}\equiv_{\mathrm{s}} \Varid{sp}_{\mathrm{2}}} implies \ensuremath{\Varid{sp}_{\mathrm{1}}\equiv_{\mathrm{b}} \Varid{sp}_{\mathrm{2}}}
  as well.
\end{proof}

In the pure case, we can also show a converse:
\begin{theorem}\label{thm:pure-bisim-implies-jr}
  Given \ensuremath{\Varid{sp}_{\mathrm{1}}\mathbin{::}\Conid{A}\mathbin{{\reflectbox{$\rightsquigarrow$}}}\Conid{S}_{1}\mathbin{{\rightsquigarrow}}\Conid{B},\Varid{sp}_{\mathrm{2}}\mathbin{::}\Conid{A}\mathbin{{\reflectbox{$\rightsquigarrow$}}}\Conid{S}_{2}\mathbin{{\rightsquigarrow}}\Conid{B}}, if  \ensuremath{\Varid{sp}_{\mathrm{1}}\equiv_{\mathrm{b}} \Varid{sp}_{\mathrm{2}}} then \ensuremath{\Varid{sp}_{\mathrm{1}}\equiv_{\mathrm{s}} \Varid{sp}_{\mathrm{2}}}.
\end{theorem}
\begin{proof}
  Given \ensuremath{\Conid{R}} and a span \ensuremath{\Varid{sp}\mathbin{::}\Conid{A}\mathbin{{\reflectbox{$\rightsquigarrow$}}}\Conid{R}\mathbin{{\rightsquigarrow}}\Conid{B}} constituting a
  bisimulation \ensuremath{\Varid{sp}_{\mathrm{1}}\equiv_{\mathrm{b}} \Varid{sp}_{\mathrm{2}}}, it suffices to
  construct a span \ensuremath{\Varid{sp'}\mathrel{=}(\Varid{l},\Varid{r})\mathbin{::}\Conid{S}_{1}\mathbin{{\reflectbox{$\rightsquigarrow$}}}\Conid{R}\mathbin{{\rightsquigarrow}}\Conid{S}_{2}} satisfying \ensuremath{\Varid{l}\mathbin{;}\Varid{sp}_{\mathrm{1}}\mathord{.}\Varid{left}\mathrel{=}\Varid{r}\mathbin{;}\Varid{sp}_{\mathrm{2}}\mathord{.}\Varid{left}} and \ensuremath{\Varid{l}\mathbin{;}\Varid{sp}_{\mathrm{1}}\mathord{.}\Varid{right}\mathrel{=}\Varid{r}\mathbin{;}\Varid{sp}_{\mathrm{2}}\mathord{.}\Varid{right}}.
\end{proof}

This result is surprising because the two equivalences come from
rather different perspectives.  Johnson and Rosebrugh introduced a
form of span equivalence, and showed that it implies bisimulation
equivalence.  They did not explicitly address the question of whether
this implication is strict.  However, there are some differences
between their presentation and ours; the most important difference is
the fact that we assume lenses to be equipped with a create function,
while they consider lenses without create functions but sometimes
consider spans of lenses to be ``pointed'', or equipped with
designated initial state values.  Likewise, \citet{abousaleh15bx}
considered bisimulation equivalence for coalgebraic bx over pointed
sets (i.e. sets equipped with designated initial values).  It remains
to be determined whether Theorem~\ref{thm:pure-bisim-implies-jr}
transfers to these settings.


We leave it as an open question to determine whether \ensuremath{\equiv_{\mathrm{b}} } is
equivalent to \ensuremath{\equiv_{\mathrm{s}} } for spans of monadic lenses (we conjecture
that they are not), or whether an
analogous result to Theorem~\ref{thm:pure-bisim-implies-jr} carries over
to symmetric lenses (we conjecture that it does).

\section{Conclusions}

Lenses are a popular and powerful abstraction for bidirectional
transformations.  Although they are most often studied in their
conventional, pure form, practical applications of lenses typically
grapple with side-effects, including exceptions, state, and user
interaction.  Some recent proposals for extending lenses with monadic
effects have been made; our proposal for (asymmetric) monadic lenses
improves on them because \ensuremath{\Conid{M}}-lenses are closed under composition for
any fixed monad
\ensuremath{\Conid{M}}. 
Furthermore, we investigated the symmetric case, and showed that
\emph{spans} of monadic lenses are also closed under composition, while the
obvious generalisation of pure symmetric lenses to incorporate monadic
effects is not closed under composition.  Finally, we presented three
notions of equivalence for spans of monadic lenses, related them, and
proved a new result: bisimulation and span
equivalence coincide for pure spans of lenses.  This last result is somewhat
surprising, given that Johnson and Rosebrugh introduced (what we call)
span equivalence to overcome perceived shortcomings in Hofmann et al.'s
bisimulation-based notion of symmetric lens equivalence.  Further 
investigation is necessary to determine whether this result
generalises.

These results illustrate the benefits of our formulation of monadic
lenses and we hope they will
inspire further research and
appreciation of bidirectional programming with effects.
\section*{Acknowledgements}
The work was supported by the UK EPSRC-funded project \textit{A
  Theory of Least Change for Bidirectional Transformations}
\citep{tlcbx} (EP/K020218/1, EP/K020919/1). 


\bibliographystyle{splncsnat}
\bibliography{../shared/strings,monadicLens}

\newpage
\appendix
\section{Proofs for Section~\ref{sec:asymmetric}}

\restatableTheorem{thm:mlens-comp-wb}
\begin{thm:mlens-comp-wb}
  If \ensuremath{\Varid{l}_{1}\mathbin{::}\monadic{\Conid{A}\rightsquigarrow\Conid{B}}{\Conid{M}}} and \ensuremath{\Varid{l}_{2}\mathbin{::}\monadic{\Conid{B}\rightsquigarrow\Conid{C}}{\Conid{M}}} are well-behaved,
  then so is \ensuremath{\Varid{l}_{1}\mathbin{;}\Varid{l}_{2}}.  
\end{thm:mlens-comp-wb}

\begin{proof}
  Suppose \ensuremath{\Varid{l}_{1}} and \ensuremath{\Varid{l}_{2}} are well-behaved, and let \ensuremath{\Varid{l}\mathrel{=}\Varid{l}_{1}\mathbin{;}\Varid{l}_{2}}.  We reason as follows for
  \ensuremath{\mathsf{(MGetPut)}}:
  \begin{hscode}\SaveRestoreHook
\column{B}{@{}>{\hspre}c<{\hspost}@{}}%
\column{BE}{@{}l@{}}%
\column{4}{@{}>{\hspre}l<{\hspost}@{}}%
\column{6}{@{}>{\hspre}l<{\hspost}@{}}%
\column{10}{@{}>{\hspre}l<{\hspost}@{}}%
\column{47}{@{}>{\hspre}c<{\hspost}@{}}%
\column{47E}{@{}l@{}}%
\column{62}{@{}>{\hspre}c<{\hspost}@{}}%
\column{62E}{@{}l@{}}%
\column{107}{@{}>{\hspre}c<{\hspost}@{}}%
\column{107E}{@{}l@{}}%
\column{E}{@{}>{\hspre}l<{\hspost}@{}}%
\>[4]{}\mathbf{do}\;\{\mskip1.5mu \Varid{l}\mathord{.}\Varid{mput}\;\Varid{a}\;(\Varid{l}\mathord{.}\Varid{mget}\;\Varid{a})\mskip1.5mu\}{}\<[E]%
\\
\>[B]{}\mathrel{=}{}\<[BE]%
\>[6]{}\mbox{\commentbegin  definition   \commentend}{}\<[E]%
\\
\>[B]{}\hsindent{4}{}\<[4]%
\>[4]{}\mathbf{do}\;\{\mskip1.5mu {}\<[10]%
\>[10]{}\Varid{b}\leftarrow \Varid{l}_{2}\mathord{.}\Varid{mput}\;(\Varid{l}_{1}\mathord{.}\Varid{mget}\;\Varid{a})\;(\Varid{l}_{2}\mathord{.}\Varid{mget}\;(\Varid{l}_{1}\mathord{.}\Varid{mget}\;\Varid{a}));\Varid{l}_{1}\mathord{.}\Varid{mput}\;\Varid{a}\;\Varid{b}{}\<[107]%
\>[107]{}\mskip1.5mu\}{}\<[107E]%
\\
\>[B]{}\mathrel{=}{}\<[BE]%
\>[6]{}\mbox{\commentbegin  \ensuremath{\mathsf{(MGetPut)}}  \commentend}{}\<[E]%
\\
\>[B]{}\hsindent{4}{}\<[4]%
\>[4]{}\mathbf{do}\;\{\mskip1.5mu {}\<[10]%
\>[10]{}\Varid{b}\leftarrow \Varid{return}\;(\Varid{l}_{1}\mathord{.}\Varid{mget}\;\Varid{a});\Varid{l}_{1}\mathord{.}\Varid{mput}\;\Varid{a}\;\Varid{b}{}\<[62]%
\>[62]{}\mskip1.5mu\}{}\<[62E]%
\\
\>[B]{}\mathrel{=}{}\<[BE]%
\>[6]{}\mbox{\commentbegin  monad unit  \commentend}{}\<[E]%
\\
\>[B]{}\hsindent{4}{}\<[4]%
\>[4]{}\mathbf{do}\;\{\mskip1.5mu {}\<[10]%
\>[10]{}\Varid{l}_{1}\mathord{.}\Varid{mput}\;\Varid{a}\;(\Varid{l}_{1}\mathord{.}\Varid{mget}\;\Varid{a}){}\<[47]%
\>[47]{}\mskip1.5mu\}{}\<[47E]%
\\
\>[B]{}\mathrel{=}{}\<[BE]%
\>[6]{}\mbox{\commentbegin  \ensuremath{\mathsf{(MGetPut)}}  \commentend}{}\<[E]%
\\
\>[B]{}\hsindent{4}{}\<[4]%
\>[4]{}\Varid{return}\;\Varid{a}{}\<[E]%
\ColumnHook
\end{hscode}\resethooks
For \ensuremath{\mathsf{(MPutGet)}}, the proof is as follows:
\begin{hscode}\SaveRestoreHook
\column{B}{@{}>{\hspre}c<{\hspost}@{}}%
\column{BE}{@{}l@{}}%
\column{4}{@{}>{\hspre}l<{\hspost}@{}}%
\column{6}{@{}>{\hspre}l<{\hspost}@{}}%
\column{10}{@{}>{\hspre}l<{\hspost}@{}}%
\column{E}{@{}>{\hspre}l<{\hspost}@{}}%
\>[4]{}\mathbf{do}\;\{\mskip1.5mu {}\<[10]%
\>[10]{}\Varid{a'}\leftarrow \Varid{l}\mathord{.}\Varid{mput}\;\Varid{a}\;\Varid{c};\Varid{return}\;(\Varid{a'},\Varid{l}\mathord{.}\Varid{mget}\;\Varid{a'})\mskip1.5mu\}{}\<[E]%
\\
\>[B]{}\mathrel{=}{}\<[BE]%
\>[6]{}\mbox{\commentbegin  Definition  \commentend}{}\<[E]%
\\
\>[B]{}\hsindent{4}{}\<[4]%
\>[4]{}\mathbf{do}\;\{\mskip1.5mu {}\<[10]%
\>[10]{}\Varid{b}\leftarrow \Varid{l}_{2}\mathord{.}\Varid{mput}\;(\Varid{l}_{1}\mathord{.}\Varid{mget}\;\Varid{a})\;\Varid{c};{}\<[E]%
\\
\>[10]{}\Varid{a'}\leftarrow \Varid{l}_{1}\mathord{.}\Varid{mput}\;\Varid{a}\;\Varid{b};{}\<[E]%
\\
\>[10]{}\Varid{return}\;(\Varid{a'},\Varid{l}_{2}\mathord{.}\Varid{mget}\;(\Varid{l}_{1}\mathord{.}\Varid{mget}\;\Varid{a'}))\mskip1.5mu\}{}\<[E]%
\\
\>[B]{}\mathrel{=}{}\<[BE]%
\>[6]{}\mbox{\commentbegin  \ensuremath{\mathsf{(MPutGet)}}  \commentend}{}\<[E]%
\\
\>[B]{}\hsindent{4}{}\<[4]%
\>[4]{}\mathbf{do}\;\{\mskip1.5mu {}\<[10]%
\>[10]{}\Varid{b}\leftarrow \Varid{l}_{2}\mathord{.}\Varid{mput}\;(\Varid{l}_{1}\mathord{.}\Varid{mget}\;\Varid{a})\;\Varid{c};{}\<[E]%
\\
\>[10]{}\Varid{a'}\leftarrow \Varid{l}_{1}\mathord{.}\Varid{mput}\;\Varid{a}\;\Varid{b};{}\<[E]%
\\
\>[10]{}\Varid{return}\;(\Varid{a'},\Varid{l}_{2}\mathord{.}\Varid{mget}\;\Varid{b})\mskip1.5mu\}{}\<[E]%
\\
\>[B]{}\mathrel{=}{}\<[BE]%
\>[6]{}\mbox{\commentbegin  \ensuremath{\mathsf{(MPutGet)}}  \commentend}{}\<[E]%
\\
\>[B]{}\hsindent{4}{}\<[4]%
\>[4]{}\mathbf{do}\;\{\mskip1.5mu {}\<[10]%
\>[10]{}\Varid{b}\leftarrow \Varid{l}_{2}\mathord{.}\Varid{mput}\;(\Varid{l}_{1}\mathord{.}\Varid{mget}\;\Varid{a})\;\Varid{c};{}\<[E]%
\\
\>[10]{}\Varid{a'}\leftarrow \Varid{l}_{1}\mathord{.}\Varid{mput}\;\Varid{a}\;\Varid{b};{}\<[E]%
\\
\>[10]{}\Varid{return}\;(\Varid{a'},\Varid{c})\mskip1.5mu\}{}\<[E]%
\\
\>[B]{}\mathrel{=}{}\<[BE]%
\>[6]{}\mbox{\commentbegin  definition  \commentend}{}\<[E]%
\\
\>[B]{}\hsindent{4}{}\<[4]%
\>[4]{}\mathbf{do}\;\{\mskip1.5mu {}\<[10]%
\>[10]{}\Varid{a'}\leftarrow \Varid{l}\mathord{.}\Varid{mput}\;\Varid{a}\;\Varid{c};\Varid{return}\;(\Varid{a'},\Varid{c})\mskip1.5mu\}{}\<[E]%
\ColumnHook
\end{hscode}\resethooks

\endswithdisplay
\end{proof}

\section{Proofs for Section~\ref{sec:symmetric}}

\restatableProposition{prop:setBool-wb-comp-not-wb}
\begin{prop:setBool-wb-comp-not-wb}
  \ensuremath{\Varid{setBool}\;\Varid{x}} is well-behaved for \ensuremath{\Varid{x}\in\{\mskip1.5mu \Conid{True},\Conid{False}\mskip1.5mu\}}, but 
  \ensuremath{\Varid{setBool}\;\Conid{True}\mathbin{;}\Varid{setBool}\;\Conid{False}} is not well-behaved. 
\end{prop:setBool-wb-comp-not-wb}
For the first part: 
\begin{proof}
  Let \ensuremath{\Varid{sl}\mathrel{=}\Varid{setBool}\;\Varid{x}}.  We consider \ensuremath{\mathsf{(PutRLM)}}, and \ensuremath{\mathsf{(PutLRM)}} is symmetric. 
  \begin{hscode}\SaveRestoreHook
\column{B}{@{}>{\hspre}c<{\hspost}@{}}%
\column{BE}{@{}l@{}}%
\column{4}{@{}>{\hspre}l<{\hspost}@{}}%
\column{6}{@{}>{\hspre}l<{\hspost}@{}}%
\column{10}{@{}>{\hspre}l<{\hspost}@{}}%
\column{27}{@{}>{\hspre}l<{\hspost}@{}}%
\column{32}{@{}>{\hspre}l<{\hspost}@{}}%
\column{35}{@{}>{\hspre}l<{\hspost}@{}}%
\column{44}{@{}>{\hspre}l<{\hspost}@{}}%
\column{50}{@{}>{\hspre}l<{\hspost}@{}}%
\column{E}{@{}>{\hspre}l<{\hspost}@{}}%
\>[4]{}\mathbf{do}\;\{\mskip1.5mu {}\<[10]%
\>[10]{}(\Varid{b},\Varid{c'})\leftarrow (\Varid{setBool}\;\Varid{x})\mathord{.}\Varid{mput}_\mathrm{R}\;((),());(\Varid{setBool}\;\Varid{x})\mathord{.}\Varid{mput}_\mathrm{L}\;(\Varid{b},\Varid{c'})\mskip1.5mu\}{}\<[E]%
\\
\>[B]{}\mathrel{=}{}\<[BE]%
\>[6]{}\mbox{\commentbegin  Definition  \commentend}{}\<[E]%
\\
\>[B]{}\hsindent{4}{}\<[4]%
\>[4]{}\mathbf{do}\;\{\mskip1.5mu (\Varid{b},\Varid{c'})\leftarrow \mathbf{do}\;\{\mskip1.5mu \set\;\Varid{x};\Varid{return}\;((),())\mskip1.5mu\};\set\;\Varid{x};\Varid{return}\;((),\Varid{c'})\mskip1.5mu\}{}\<[E]%
\\
\>[B]{}\mathrel{=}{}\<[BE]%
\>[6]{}\mbox{\commentbegin  monad associativity  \commentend}{}\<[E]%
\\
\>[B]{}\hsindent{4}{}\<[4]%
\>[4]{}\mathbf{do}\;\{\mskip1.5mu \set\;\Varid{x};(\Varid{b},\Varid{c'})\leftarrow {}\<[27]%
\>[27]{}\Varid{return}\;((),());\set\;\Varid{x};\Varid{return}\;((),\Varid{c'})\mskip1.5mu\}{}\<[E]%
\\
\>[B]{}\mathrel{=}{}\<[BE]%
\>[6]{}\mbox{\commentbegin  commutativity of \ensuremath{\Varid{return}}  \commentend}{}\<[E]%
\\
\>[B]{}\hsindent{4}{}\<[4]%
\>[4]{}\mathbf{do}\;\{\mskip1.5mu \set\;\Varid{x};\set\;\Varid{x};(\Varid{b},\Varid{c'})\leftarrow {}\<[35]%
\>[35]{}\Varid{return}\;((),());\Varid{return}\;((),\Varid{c'})\mskip1.5mu\}{}\<[E]%
\\
\>[B]{}\mathrel{=}{}\<[BE]%
\>[6]{}\mbox{\commentbegin  \ensuremath{\set\;\Varid{x};\set\;\Varid{x}\mathrel{=}\set\;\Varid{x}}  \commentend}{}\<[E]%
\\
\>[B]{}\hsindent{4}{}\<[4]%
\>[4]{}\mathbf{do}\;\{\mskip1.5mu \set\;\Varid{x};(\Varid{b},\Varid{c'})\leftarrow {}\<[27]%
\>[27]{}\Varid{return}\;((),());{}\<[44]%
\>[44]{}\Varid{return}\;((),\Varid{c'})\mskip1.5mu\}{}\<[E]%
\\
\>[B]{}\mathrel{=}{}\<[BE]%
\>[6]{}\mbox{\commentbegin  monad associativity  \commentend}{}\<[E]%
\\
\>[B]{}\hsindent{4}{}\<[4]%
\>[4]{}\mathbf{do}\;\{\mskip1.5mu (\Varid{b},\Varid{c'})\leftarrow \mathbf{do}\;\{\mskip1.5mu \set\;\Varid{x};{}\<[32]%
\>[32]{}\Varid{return}\;((),())\mskip1.5mu\};{}\<[50]%
\>[50]{}\Varid{return}\;((),\Varid{c'})\mskip1.5mu\}{}\<[E]%
\\
\>[B]{}\mathrel{=}{}\<[BE]%
\>[6]{}\mbox{\commentbegin  Definition  \commentend}{}\<[E]%
\\
\>[B]{}\hsindent{4}{}\<[4]%
\>[4]{}\mathbf{do}\;\{\mskip1.5mu (\Varid{b},\Varid{c'})\leftarrow (\Varid{setBool}\;\Varid{x})\mathord{.}\Varid{mput}_\mathrm{R}\;((),());\Varid{return}\;((),\Varid{c'})\mskip1.5mu\}{}\<[E]%
\ColumnHook
\end{hscode}\resethooks

For the second part, taking \ensuremath{\Varid{sl}\mathrel{=}\Varid{setBool}\;\Conid{True}\mathbin{;}\Varid{setBool}\;\Conid{False}}, we proceed as follows:
  \begin{hscode}\SaveRestoreHook
\column{B}{@{}>{\hspre}c<{\hspost}@{}}%
\column{BE}{@{}l@{}}%
\column{4}{@{}>{\hspre}l<{\hspost}@{}}%
\column{6}{@{}>{\hspre}l<{\hspost}@{}}%
\column{10}{@{}>{\hspre}l<{\hspost}@{}}%
\column{E}{@{}>{\hspre}l<{\hspost}@{}}%
\>[4]{}\mathbf{do}\;\{\mskip1.5mu {}\<[10]%
\>[10]{}(\Varid{c},\Varid{s'})\leftarrow \Varid{sl}\mathord{.}\Varid{mput}_\mathrm{R}\;(\Varid{a},\Varid{s});\Varid{sl}\mathord{.}\Varid{mput}_\mathrm{L}\;(\Varid{c},\Varid{s'})\mskip1.5mu\}{}\<[E]%
\\
\>[B]{}\mathrel{=}{}\<[BE]%
\>[6]{}\mbox{\commentbegin  let \ensuremath{\Varid{s}\mathrel{=}(\Varid{s}_{1},\Varid{s}_{2})} and \ensuremath{\Varid{s'}\mathrel{=}(\Varid{s}_{1}''',\Varid{s}_{2}''')}; definition  \commentend}{}\<[E]%
\\
\>[B]{}\hsindent{4}{}\<[4]%
\>[4]{}\mathbf{do}\;\{\mskip1.5mu {}\<[10]%
\>[10]{}(\Varid{b},\Varid{s}_{1}')\leftarrow (\Varid{setBool}\;\Conid{True})\mathord{.}\Varid{mput}_\mathrm{R}\;(\Varid{a},\Varid{s}_{1});{}\<[E]%
\\
\>[10]{}(\Varid{c},\Varid{s}_{2}')\leftarrow (\Varid{setBool}\;\Conid{False})\mathord{.}\Varid{mput}_\mathrm{R}\;(\Varid{b},\Varid{s}_{2});{}\<[E]%
\\
\>[10]{}(\Varid{c'},(\Varid{s}_{1}'',\Varid{s}_{2}''))\leftarrow \Varid{return}\;(\Varid{c},(\Varid{s}_{1}',\Varid{s}_{2}'));{}\<[E]%
\\
\>[10]{}(\Varid{b'},\Varid{s}_{2}''')\leftarrow (\Varid{setBool}\;\Conid{False})\mathord{.}\Varid{mput}_\mathrm{L}\;(\Varid{c'},\Varid{s}_{2}'');{}\<[E]%
\\
\>[10]{}(\Varid{a'},\Varid{s}_{2}''')\leftarrow (\Varid{setBool}\;\Conid{True})\mathord{.}\Varid{mput}_\mathrm{L}\;(\Varid{b'},\Varid{s}_{1}'');{}\<[E]%
\\
\>[10]{}\Varid{return}\;(\Varid{c},(\Varid{s}_{1}''',\Varid{s}_{2}'''))\mskip1.5mu\}{}\<[E]%
\\
\>[B]{}\mathrel{=}{}\<[BE]%
\>[6]{}\mbox{\commentbegin  monad unit  \commentend}{}\<[E]%
\\
\>[B]{}\hsindent{4}{}\<[4]%
\>[4]{}\mathbf{do}\;\{\mskip1.5mu {}\<[10]%
\>[10]{}(\Varid{b},\Varid{s}_{1}')\leftarrow (\Varid{setBool}\;\Conid{True})\mathord{.}\Varid{mput}_\mathrm{R}\;(\Varid{a},\Varid{s}_{1});{}\<[E]%
\\
\>[10]{}(\Varid{c},\Varid{s}_{2}')\leftarrow (\Varid{setBool}\;\Conid{False})\mathord{.}\Varid{mput}_\mathrm{R}\;(\Varid{b},\Varid{s}_{2});{}\<[E]%
\\
\>[10]{}(\Varid{b'},\Varid{s}_{2}''')\leftarrow (\Varid{setBool}\;\Conid{False})\mathord{.}\Varid{mput}_\mathrm{L}\;(\Varid{c'},\Varid{s}_{2}');{}\<[E]%
\\
\>[10]{}(\Varid{a'},\Varid{s}_{2}''')\leftarrow (\Varid{setBool}\;\Conid{True})\mathord{.}\Varid{mput}_\mathrm{L}\;(\Varid{b'},\Varid{s}_{1}');{}\<[E]%
\\
\>[10]{}\Varid{return}\;(\Varid{c},(\Varid{s}_{1}''',\Varid{s}_{2}'''))\mskip1.5mu\}{}\<[E]%
\\
\>[B]{}\mathrel{=}{}\<[BE]%
\>[6]{}\mbox{\commentbegin  \ensuremath{\mathsf{(PutRLM)}} for \ensuremath{\Varid{setBool}\;\Conid{False}}  \commentend}{}\<[E]%
\\
\>[B]{}\hsindent{4}{}\<[4]%
\>[4]{}\mathbf{do}\;\{\mskip1.5mu {}\<[10]%
\>[10]{}(\Varid{b},\Varid{s}_{1}')\leftarrow (\Varid{setBool}\;\Conid{True})\mathord{.}\Varid{mput}_\mathrm{R}\;(\Varid{a},\Varid{s}_{1});{}\<[E]%
\\
\>[10]{}(\Varid{c},\Varid{s}_{2}')\leftarrow (\Varid{setBool}\;\Conid{False})\mathord{.}\Varid{mput}_\mathrm{R}\;(\Varid{b},\Varid{s}_{2});{}\<[E]%
\\
\>[10]{}(\Varid{b'},\Varid{s}_{2}''')\leftarrow \Varid{return}\;(\Varid{b},\Varid{s}_{2}');{}\<[E]%
\\
\>[10]{}(\Varid{a'},\Varid{s}_{2}''')\leftarrow (\Varid{setBool}\;\Conid{False})\mathord{.}\Varid{mput}_\mathrm{L}\;(\Varid{b'},\Varid{s}_{1}');{}\<[E]%
\\
\>[10]{}\Varid{return}\;(\Varid{c},(\Varid{s}_{1}''',\Varid{s}_{2}'''))\mskip1.5mu\}{}\<[E]%
\\
\>[B]{}\mathrel{=}{}\<[BE]%
\>[6]{}\mbox{\commentbegin  monad unit  \commentend}{}\<[E]%
\\
\>[B]{}\hsindent{4}{}\<[4]%
\>[4]{}\mathbf{do}\;\{\mskip1.5mu {}\<[10]%
\>[10]{}(\Varid{b},\Varid{s}_{1}')\leftarrow (\Varid{setBool}\;\Conid{True})\mathord{.}\Varid{mput}_\mathrm{R}\;(\Varid{a},\Varid{s}_{1});{}\<[E]%
\\
\>[10]{}(\Varid{c},\Varid{s}_{2}')\leftarrow (\Varid{setBool}\;\Conid{False})\mathord{.}\Varid{mput}_\mathrm{R}\;(\Varid{b},\Varid{s}_{2});{}\<[E]%
\\
\>[10]{}(\Varid{a'},\Varid{s}_{2}''')\leftarrow (\Varid{setBool}\;\Conid{True})\mathord{.}\Varid{mput}_\mathrm{L}\;(\Varid{b},\Varid{s}_{1}');{}\<[E]%
\\
\>[10]{}\Varid{return}\;(\Varid{c},(\Varid{s}_{1}''',\Varid{s}_{2}'))\mskip1.5mu\}{}\<[E]%
\ColumnHook
\end{hscode}\resethooks
However, we cannot simplify this any further.  Moreover, it should be clear that the shared state will be \ensuremath{\Conid{True}} after this operation is performed.
 Considering the other side of the desired equation:
\begin{hscode}\SaveRestoreHook
\column{B}{@{}>{\hspre}c<{\hspost}@{}}%
\column{BE}{@{}l@{}}%
\column{4}{@{}>{\hspre}l<{\hspost}@{}}%
\column{6}{@{}>{\hspre}l<{\hspost}@{}}%
\column{10}{@{}>{\hspre}l<{\hspost}@{}}%
\column{E}{@{}>{\hspre}l<{\hspost}@{}}%
\>[4]{}\mathbf{do}\;\{\mskip1.5mu {}\<[10]%
\>[10]{}(\Varid{c},\Varid{s'})\leftarrow \Varid{sl}\mathord{.}\Varid{mput}_\mathrm{R}\;(\Varid{a},\Varid{s});\Varid{sl}\mathord{.}\Varid{mput}_\mathrm{L}\;(\Varid{c},\Varid{s''})\mskip1.5mu\}{}\<[E]%
\\
\>[B]{}\mathrel{=}{}\<[BE]%
\>[6]{}\mbox{\commentbegin   let \ensuremath{\Varid{s}\mathrel{=}(\Varid{s}_{1},\Varid{s}_{2})} and \ensuremath{\Varid{s'}\mathrel{=}(\Varid{s}_{1}''',\Varid{s}_{2}''')}; Definition  \commentend}{}\<[E]%
\\
\>[B]{}\hsindent{4}{}\<[4]%
\>[4]{}\mathbf{do}\;\{\mskip1.5mu {}\<[10]%
\>[10]{}(\Varid{b},\Varid{s}_{1}')\leftarrow (\Varid{setBool}\;\Conid{True})\mathord{.}\Varid{mput}_\mathrm{R}\;(\Varid{a},\Varid{s}_{1});{}\<[E]%
\\
\>[10]{}(\Varid{c},\Varid{s}_{2}')\leftarrow (\Varid{setBool}\;\Conid{False})\mathord{.}\Varid{mput}_\mathrm{R}\;(\Varid{b},\Varid{s}_{2});{}\<[E]%
\\
\>[10]{}(\Varid{c'},(\Varid{s}_{1}'',\Varid{s}_{2}''))\leftarrow \Varid{return}\;(\Varid{c},(\Varid{s}_{1}',\Varid{s}_{2}'));{}\<[E]%
\\
\>[10]{}\Varid{return}\;(\Varid{c'},(\Varid{s}_{1}'',\Varid{s}_{2}''))\mskip1.5mu\}{}\<[E]%
\\
\>[B]{}\mathrel{=}{}\<[BE]%
\>[6]{}\mbox{\commentbegin  Monad unit  \commentend}{}\<[E]%
\\
\>[B]{}\hsindent{4}{}\<[4]%
\>[4]{}\mathbf{do}\;\{\mskip1.5mu {}\<[10]%
\>[10]{}(\Varid{b},\Varid{s}_{1}')\leftarrow (\Varid{setBool}\;\Conid{True})\mathord{.}\Varid{mput}_\mathrm{R}\;(\Varid{a},\Varid{s}_{1});{}\<[E]%
\\
\>[10]{}(\Varid{c},\Varid{s}_{2}')\leftarrow (\Varid{setBool}\;\Conid{False})\mathord{.}\Varid{mput}_\mathrm{R}\;(\Varid{b},\Varid{s}_{2});{}\<[E]%
\\
\>[10]{}\Varid{return}\;(\Varid{c},(\Varid{s}_{1}',\Varid{s}_{2}'))\mskip1.5mu\}{}\<[E]%
\ColumnHook
\end{hscode}\resethooks
it should be clear that the shared state will be \ensuremath{\Conid{False}} after this operation is performed.  Therefore, \ensuremath{\mathsf{(PutRLM)}} is not satisfied by \ensuremath{\Varid{sl}}.
\end{proof}

\restatableLemma{lem:cospan2span}
\begin{lem:cospan2span}
  If  \ensuremath{\Varid{ml}_{\mathrm{1}}\mathbin{::}\monadic{\sigma_1\rightsquigarrow\beta}{\mu}} and \ensuremath{\Varid{ml}_{\mathrm{2}}\mathbin{::}\monadic{\sigma_2\rightsquigarrow\beta}{\mu}} are well-behaved then so is \ensuremath{\Varid{ml}_{\mathrm{1}}\mathbin{\Join}\Varid{ml}_{\mathrm{2}}\mathbin{::}\monadic{\sigma_1\mathbin{\reflectbox{$\rightsquigarrow$}}(\sigma_1\mathbin{\!\Join\!}\sigma_2)\rightsquigarrow\sigma_2}{\mu}}.
\end{lem:cospan2span}
\begin{proof}
  It suffices to consider the two lenses \ensuremath{\Varid{l}_{1}\mathrel{=}\Conid{MLens}\;\Varid{fst}\;\Varid{put}_\mathrm{L}\;\Varid{create}_\mathrm{L}} and
  \ensuremath{\Varid{l}_{2}\mathrel{=}\Conid{MLens}\;\Varid{snd}\;\Varid{put}_\mathrm{R}\;\Varid{create}_\mathrm{R}} in isolation.  Moreover, the two cases are
  completely symmetric, so we only show the first. 

For \ensuremath{\mathsf{(MGetPut)}}, we show:
\begin{hscode}\SaveRestoreHook
\column{B}{@{}>{\hspre}c<{\hspost}@{}}%
\column{BE}{@{}l@{}}%
\column{4}{@{}>{\hspre}l<{\hspost}@{}}%
\column{6}{@{}>{\hspre}l<{\hspost}@{}}%
\column{E}{@{}>{\hspre}l<{\hspost}@{}}%
\>[4]{}\mathbf{do}\;\{\mskip1.5mu \Varid{l}_{1}\mathord{.}\Varid{mput}\;(\Varid{s}_{1},\Varid{s}_{2})\;(\Varid{l}_{1}\mathord{.}\Varid{mget}\;(\Varid{s}_{1},\Varid{s}_{2}))\mskip1.5mu\}{}\<[E]%
\\
\>[B]{}\mathrel{=}{}\<[BE]%
\>[6]{}\mbox{\commentbegin  definition  \commentend}{}\<[E]%
\\
\>[B]{}\hsindent{4}{}\<[4]%
\>[4]{}\mathbf{do}\;\{\mskip1.5mu \Varid{put}_\mathrm{L}\;(\Varid{s}_{1},\Varid{s}_{2})\;(\Varid{fst}\;(\Varid{s}_{1},\Varid{s}_{2}))\mskip1.5mu\}{}\<[E]%
\\
\>[B]{}\mathrel{=}{}\<[BE]%
\>[6]{}\mbox{\commentbegin  definition of \ensuremath{\Varid{put}_\mathrm{L}} and \ensuremath{\Varid{fst}}  \commentend}{}\<[E]%
\\
\>[B]{}\hsindent{4}{}\<[4]%
\>[4]{}\mathbf{do}\;\{\mskip1.5mu \Varid{s}_{2}'\leftarrow \Varid{ml}_{\mathrm{2}}\mathord{.}\Varid{mput}\;\Varid{s}_{2}\;(\Varid{ml}_{\mathrm{1}}\mathord{.}\Varid{mget}\;\Varid{s}_{1})\mskip1.5mu\}{}\<[E]%
\\
\>[B]{}\mathrel{=}{}\<[BE]%
\>[6]{}\mbox{\commentbegin  \ensuremath{(\Varid{s}_{1},\Varid{s}_{2})} consistent  \commentend}{}\<[E]%
\\
\>[B]{}\hsindent{4}{}\<[4]%
\>[4]{}\mathbf{do}\;\{\mskip1.5mu \Varid{s}_{2}'\leftarrow \Varid{ml}_{\mathrm{2}}\mathord{.}\Varid{mput}\;\Varid{s}_{2}\;(\Varid{ml}_{\mathrm{2}}\mathord{.}\Varid{mget}\;\Varid{s}_{2})\mskip1.5mu\}{}\<[E]%
\\
\>[B]{}\mathrel{=}{}\<[BE]%
\>[6]{}\mbox{\commentbegin  \ensuremath{\mathsf{(MGetPut)}}  \commentend}{}\<[E]%
\\
\>[B]{}\hsindent{4}{}\<[4]%
\>[4]{}\Varid{return}\;\Varid{s}{}\<[E]%
\ColumnHook
\end{hscode}\resethooks

The proof for \ensuremath{\mathsf{(MPutGet)}} goes as follows.  Note that it holds by
construction, without appealing to well-behavedness of \ensuremath{\Varid{ml}_{\mathrm{1}}} or \ensuremath{\Varid{ml}_{\mathrm{2}}}.
\begin{hscode}\SaveRestoreHook
\column{B}{@{}>{\hspre}c<{\hspost}@{}}%
\column{BE}{@{}l@{}}%
\column{4}{@{}>{\hspre}l<{\hspost}@{}}%
\column{6}{@{}>{\hspre}l<{\hspost}@{}}%
\column{10}{@{}>{\hspre}l<{\hspost}@{}}%
\column{E}{@{}>{\hspre}l<{\hspost}@{}}%
\>[4]{}\mathbf{do}\;\{\mskip1.5mu (\Varid{s}_{1}',\Varid{s}_{2}')\leftarrow \Varid{l}_{1}\mathord{.}\Varid{mput}\;(\Varid{s}_{1},\Varid{s}_{2})\;\Varid{a};\Varid{return}\;((\Varid{s}_{1}',\Varid{s}_{2}'),\Varid{l}_{1}\mathord{.}\Varid{mget}\;(\Varid{s}_{1}',\Varid{s}_{2}'))\mskip1.5mu\}{}\<[E]%
\\
\>[B]{}\mathrel{=}{}\<[BE]%
\>[6]{}\mbox{\commentbegin  definition  \commentend}{}\<[E]%
\\
\>[B]{}\hsindent{4}{}\<[4]%
\>[4]{}\mathbf{do}\;\{\mskip1.5mu (\Varid{s}_{1}',\Varid{s}_{2}')\leftarrow \Varid{put}_\mathrm{R}\;(\Varid{s}_{1},\Varid{s}_{2})\;\Varid{a};\Varid{return}\;((\Varid{s}_{1}',\Varid{s}_{2}'),\Varid{fst}\;(\Varid{s}_{1}',\Varid{s}_{2}'))\mskip1.5mu\}{}\<[E]%
\\
\>[B]{}\mathrel{=}{}\<[BE]%
\>[6]{}\mbox{\commentbegin  definition  \commentend}{}\<[E]%
\\
\>[B]{}\hsindent{4}{}\<[4]%
\>[4]{}\mathbf{do}\;\{\mskip1.5mu {}\<[10]%
\>[10]{}\Varid{s}_{2}''\leftarrow \Varid{ml}_{\mathrm{2}}\mathord{.}\Varid{mput}\;\Varid{s}_{2}\;(\Varid{ml}_{\mathrm{1}}\mathord{.}\Varid{mget}\;\Varid{a});(\Varid{s}_{1}',\Varid{s}_{2}')\leftarrow \Varid{return}\;(\Varid{a},\Varid{s}_{2}'');{}\<[E]%
\\
\>[10]{}\Varid{return}\;((\Varid{s}_{1}',\Varid{s}_{2}'),\Varid{fst}\;(\Varid{s}_{1}',\Varid{s}_{2}'))\mskip1.5mu\}{}\<[E]%
\\
\>[B]{}\mathrel{=}{}\<[BE]%
\>[6]{}\mbox{\commentbegin  definition of \ensuremath{\Varid{fst}}  \commentend}{}\<[E]%
\\
\>[B]{}\hsindent{4}{}\<[4]%
\>[4]{}\mathbf{do}\;\{\mskip1.5mu {}\<[10]%
\>[10]{}\Varid{s}_{2}''\leftarrow \Varid{ml}_{\mathrm{2}}\mathord{.}\Varid{mput}\;\Varid{s}_{2}\;(\Varid{ml}_{\mathrm{1}}\mathord{.}\Varid{mget}\;\Varid{a});(\Varid{s}_{1}',\Varid{s}_{2}')\leftarrow \Varid{return}\;(\Varid{a},\Varid{s}_{2}'');{}\<[E]%
\\
\>[10]{}\Varid{return}\;((\Varid{s}_{1}',\Varid{s}_{2}'),\Varid{s}_{1}'){}\<[E]%
\\
\>[B]{}\mathrel{=}{}\<[BE]%
\>[6]{}\mbox{\commentbegin  monad laws  \commentend}{}\<[E]%
\\
\>[B]{}\hsindent{4}{}\<[4]%
\>[4]{}\mathbf{do}\;\{\mskip1.5mu {}\<[10]%
\>[10]{}\Varid{s}_{2}''\leftarrow \Varid{ml}_{\mathrm{2}}\mathord{.}\Varid{mput}\;\Varid{s}_{2}\;(\Varid{ml}_{\mathrm{1}}\mathord{.}\Varid{mget}\;\Varid{a});(\Varid{s}_{1}',\Varid{s}_{2}')\leftarrow \Varid{return}\;(\Varid{a},\Varid{s}_{2}'');{}\<[E]%
\\
\>[10]{}\Varid{return}\;((\Varid{s}_{1}',\Varid{s}_{2}'),\Varid{a})\mskip1.5mu\}{}\<[E]%
\\
\>[B]{}\mathrel{=}{}\<[BE]%
\>[6]{}\mbox{\commentbegin  definition  \commentend}{}\<[E]%
\\
\>[B]{}\hsindent{4}{}\<[4]%
\>[4]{}\mathbf{do}\;\{\mskip1.5mu (\Varid{s}_{1}',\Varid{s}_{2}')\leftarrow \Varid{put}_\mathrm{L}\;(\Varid{s}_{1},\Varid{s}_{2})\;\Varid{a};\Varid{return}\;((\Varid{s}_{1}',\Varid{s}_{2}'),\Varid{a})\mskip1.5mu\}{}\<[E]%
\\
\>[B]{}\mathrel{=}{}\<[BE]%
\>[6]{}\mbox{\commentbegin  definition  \commentend}{}\<[E]%
\\
\>[B]{}\hsindent{4}{}\<[4]%
\>[4]{}\mathbf{do}\;\{\mskip1.5mu (\Varid{s}_{1}',\Varid{s}_{2}')\leftarrow \Varid{l}_{1}\mathord{.}\Varid{mput}\;(\Varid{s}_{1},\Varid{s}_{2})\;\Varid{a};\Varid{return}\;((\Varid{s}_{1}',\Varid{s}_{2}'),\Varid{a})\mskip1.5mu\}{}\<[E]%
\ColumnHook
\end{hscode}\resethooks
The proof for \ensuremath{\mathsf{(MCreateGet)}} is similar.  

Finally, we show that \ensuremath{\Varid{put}_\mathrm{L}\mathbin{::}(\sigma_1\mathbin{\!\Join\!}\sigma_2)\to \sigma_1\to \mu\;(\sigma_1\mathbin{\!\Join\!}\sigma_2)}, and in particular, that it maintains the consistency
invariant on the state space \ensuremath{\sigma_1\mathbin{\!\Join\!}\sigma_2}.
Assume that \ensuremath{(\Varid{s}_{1},\Varid{s}_{2})\mathbin{::}\sigma_1\mathbin{\!\Join\!}\sigma_2} and \ensuremath{\Varid{s}_{1}'\mathbin{::}\sigma_1} are
given.  Thus, \ensuremath{\Varid{ml}_{\mathrm{1}}\mathord{.}\Varid{mget}\;\Varid{s}_{1}\mathrel{=}\Varid{ml}_{\mathrm{2}}\mathord{.}\Varid{mget}\;\Varid{s}_{2}}.  We must show that any value returned by \ensuremath{\Varid{put}_\mathrm{L}} also satisfies this consistency criterion.   By definition, 
\begin{hscode}\SaveRestoreHook
\column{B}{@{}>{\hspre}l<{\hspost}@{}}%
\column{E}{@{}>{\hspre}l<{\hspost}@{}}%
\>[B]{}\Varid{put}_\mathrm{L}\;(\Varid{s}_{1},\Varid{s}_{2})\;\Varid{s}_{1}'\mathrel{=}\mathbf{do}\;\{\mskip1.5mu \Varid{s}_{2}'\leftarrow \Varid{ml}_{\mathrm{2}}\mathord{.}\Varid{mput}\;\Varid{s}_{2}\;(\Varid{ml}_{\mathrm{1}}\mathord{.}\Varid{mget}\;\Varid{s}_{1}');\Varid{return}\;(\Varid{s}_{1}',\Varid{s}_{2}')\mskip1.5mu\}{}\<[E]%
\ColumnHook
\end{hscode}\resethooks
By \ensuremath{\mathsf{(MPutGet)}}, any \ensuremath{\Varid{s}_{2}'} resulting from \ensuremath{\Varid{ml}_{\mathrm{2}}\mathord{.}\Varid{mput}\;\Varid{s}_{2}\;(\Varid{ml}_{\mathrm{1}}\mathord{.}\Varid{mget}\;\Varid{s}_{1}')} will satisfy \ensuremath{\Varid{ml}_{\mathrm{2}}\mathord{.}\Varid{mget}\;\Varid{s}_{2}'\mathrel{=}\Varid{ml}_{\mathrm{1}}\mathord{.}\Varid{mget}\;\Varid{s}_{1}'}.  The proof that
\ensuremath{\Varid{create}_\mathrm{L}\mathbin{::}\sigma_1\to \mu\;(\sigma_1\mathbin{\!\Join\!}\sigma_2)} is similar, but simpler.
\end{proof}

\restatableTheorem{thm:span2smlens-wb}
\begin{thm:span2smlens-wb}
  If \ensuremath{\Varid{sp}\mathbin{::}\monadic{\Conid{A}\mathbin{\reflectbox{$\rightsquigarrow$}}\Conid{S}\rightsquigarrow\Conid{B}}{\Conid{M}}} is well-behaved, then \ensuremath{\Varid{span2smlens}\;\Varid{sp}} is also well-behaved.
\end{thm:span2smlens-wb}
\begin{proof}
Let \ensuremath{\Varid{sl}\mathrel{=}\Varid{span2smlens}\;\Varid{sp}}.  We need to show that the laws \ensuremath{\mathsf{(PutRLM)}}
and \ensuremath{\mathsf{(PutLRM)}} hold.  We show \ensuremath{\mathsf{(PutRLM)}}, and \ensuremath{\mathsf{(PutLRM)}} is
symmetric.

We need to show that
\begin{hscode}\SaveRestoreHook
\column{B}{@{}>{\hspre}c<{\hspost}@{}}%
\column{BE}{@{}l@{}}%
\column{4}{@{}>{\hspre}l<{\hspost}@{}}%
\column{10}{@{}>{\hspre}l<{\hspost}@{}}%
\column{E}{@{}>{\hspre}l<{\hspost}@{}}%
\>[4]{}\mathbf{do}\;\{\mskip1.5mu {}\<[10]%
\>[10]{}(\Varid{b'},\Varid{mc'})\leftarrow \Varid{sl}\mathord{.}\Varid{mput}_\mathrm{R}\;(\Varid{a},\Varid{mc});\Varid{sl}\mathord{.}\Varid{mput}_\mathrm{L}\;(\Varid{b'},\Varid{mc'})\mskip1.5mu\}{}\<[E]%
\\
\>[B]{}\mathrel{=}{}\<[BE]%
\\
\>[B]{}\hsindent{4}{}\<[4]%
\>[4]{}\mathbf{do}\;\{\mskip1.5mu {}\<[10]%
\>[10]{}(\Varid{b'},\Varid{mc'})\leftarrow \Varid{sl}\mathord{.}\Varid{mput}_\mathrm{R}\;(\Varid{a},\Varid{mc});\Varid{return}\;(\Varid{a},\Varid{mc'})\mskip1.5mu\}{}\<[E]%
\ColumnHook
\end{hscode}\resethooks
There are two cases, depending on whether the initial state \ensuremath{\Varid{mc}} is \ensuremath{\Conid{Nothing}} or \ensuremath{\Conid{Just}\;\Varid{c}} for some \ensuremath{\Varid{c}}.  

If \ensuremath{\Varid{mc}\mathrel{=}\Conid{Nothing}} then we reason as follows:
\begin{hscode}\SaveRestoreHook
\column{B}{@{}>{\hspre}c<{\hspost}@{}}%
\column{BE}{@{}l@{}}%
\column{4}{@{}>{\hspre}l<{\hspost}@{}}%
\column{6}{@{}>{\hspre}l<{\hspost}@{}}%
\column{10}{@{}>{\hspre}l<{\hspost}@{}}%
\column{E}{@{}>{\hspre}l<{\hspost}@{}}%
\>[4]{}\mathbf{do}\;\{\mskip1.5mu {}\<[10]%
\>[10]{}(\Varid{b'},\Varid{mc'})\leftarrow \Varid{sl}\mathord{.}\Varid{mput}_\mathrm{R}\;(\Varid{a},\Conid{Nothing});\Varid{sl}\mathord{.}\Varid{mput}_\mathrm{L}\;(\Varid{b'},\Varid{mc'})\mskip1.5mu\}{}\<[E]%
\\
\>[B]{}\mathrel{=}{}\<[BE]%
\>[6]{}\mbox{\commentbegin  Definition  \commentend}{}\<[E]%
\\
\>[B]{}\hsindent{4}{}\<[4]%
\>[4]{}\mathbf{do}\;\{\mskip1.5mu {}\<[10]%
\>[10]{}\Varid{s'}\leftarrow \Varid{sp}\mathord{.}\Varid{left}\mathord{.}\Varid{mcreate}\;\Varid{a};(\Varid{b'},\Varid{mc'})\leftarrow \Varid{return}\;(\Varid{sp}\mathord{.}\Varid{right}\mathord{.}\Varid{mget}\;\Varid{s'},\Conid{Just}\;\Varid{s'});{}\<[E]%
\\
\>[10]{}\Varid{sl}\mathord{.}\Varid{mput}_\mathrm{L}\;(\Varid{b'},\Varid{mc'})\mskip1.5mu\}{}\<[E]%
\\
\>[B]{}\mathrel{=}{}\<[BE]%
\>[6]{}\mbox{\commentbegin  monad unit  \commentend}{}\<[E]%
\\
\>[B]{}\hsindent{4}{}\<[4]%
\>[4]{}\mathbf{do}\;\{\mskip1.5mu {}\<[10]%
\>[10]{}\Varid{s'}\leftarrow \Varid{sp}\mathord{.}\Varid{left}\mathord{.}\Varid{mcreate}\;\Varid{a};\Varid{sl}\mathord{.}\Varid{mput}_\mathrm{L}\;(\Varid{sp}\mathord{.}\Varid{right}\mathord{.}\Varid{mget}\;\Varid{s'},\Conid{Just}\;\Varid{s'})\mskip1.5mu\}{}\<[E]%
\\
\>[B]{}\mathrel{=}{}\<[BE]%
\>[6]{}\mbox{\commentbegin  definition  \commentend}{}\<[E]%
\\
\>[B]{}\hsindent{4}{}\<[4]%
\>[4]{}\mathbf{do}\;\{\mskip1.5mu {}\<[10]%
\>[10]{}\Varid{s'}\leftarrow \Varid{sp}\mathord{.}\Varid{left}\mathord{.}\Varid{mcreate}\;\Varid{a};\Varid{s''}\leftarrow \Varid{sp}\mathord{.}\Varid{right}\mathord{.}\Varid{mput}\;\Varid{s'}\;(\Varid{sp}\mathord{.}\Varid{right}\mathord{.}\Varid{mget}\;\Varid{s'});{}\<[E]%
\\
\>[10]{}\Varid{return}\;(\Varid{sp}\mathord{.}\Varid{left}\mathord{.}\Varid{mget}\;\Varid{s''},\Conid{Just}\;\Varid{s''})\mskip1.5mu\}{}\<[E]%
\\
\>[B]{}\mathrel{=}{}\<[BE]%
\>[6]{}\mbox{\commentbegin  \ensuremath{\mathsf{(MGetPut)}}  \commentend}{}\<[E]%
\\
\>[B]{}\hsindent{4}{}\<[4]%
\>[4]{}\mathbf{do}\;\{\mskip1.5mu {}\<[10]%
\>[10]{}\Varid{s'}\leftarrow \Varid{sp}\mathord{.}\Varid{left}\mathord{.}\Varid{mcreate}\;\Varid{a};\Varid{s''}\leftarrow \Varid{return}\;\Varid{s'};{}\<[E]%
\\
\>[10]{}\Varid{return}\;(\Varid{sp}\mathord{.}\Varid{left}\mathord{.}\Varid{mget}\;\Varid{s''},\Conid{Just}\;\Varid{s''})\mskip1.5mu\}{}\<[E]%
\\
\>[B]{}\mathrel{=}{}\<[BE]%
\>[6]{}\mbox{\commentbegin  monad unit  \commentend}{}\<[E]%
\\
\>[B]{}\hsindent{4}{}\<[4]%
\>[4]{}\mathbf{do}\;\{\mskip1.5mu {}\<[10]%
\>[10]{}\Varid{s'}\leftarrow \Varid{sp}\mathord{.}\Varid{left}\mathord{.}\Varid{mcreate}\;\Varid{a};\Varid{return}\;(\Varid{sp}\mathord{.}\Varid{left}\mathord{.}\Varid{mget}\;\Varid{s'},\Conid{Just}\;\Varid{s'})\mskip1.5mu\}{}\<[E]%
\\
\>[B]{}\mathrel{=}{}\<[BE]%
\>[6]{}\mbox{\commentbegin  \ensuremath{\mathsf{(MCreateGet)}}  \commentend}{}\<[E]%
\\
\>[B]{}\hsindent{4}{}\<[4]%
\>[4]{}\mathbf{do}\;\{\mskip1.5mu {}\<[10]%
\>[10]{}\Varid{s'}\leftarrow \Varid{sp}\mathord{.}\Varid{left}\mathord{.}\Varid{mcreate}\;\Varid{a};\Varid{return}\;(\Varid{a},\Conid{Just}\;\Varid{s'})\mskip1.5mu\}{}\<[E]%
\\
\>[B]{}\mathrel{=}{}\<[BE]%
\>[6]{}\mbox{\commentbegin  monad unit  \commentend}{}\<[E]%
\\
\>[B]{}\hsindent{4}{}\<[4]%
\>[4]{}\mathbf{do}\;\{\mskip1.5mu {}\<[10]%
\>[10]{}\Varid{s'}\leftarrow \Varid{sp}\mathord{.}\Varid{left}\mathord{.}\Varid{mcreate}\;\Varid{a};(\Varid{b'},\Varid{mc'})\leftarrow (\Varid{sp}\mathord{.}\Varid{right}\mathord{.}\Varid{get}\;\Varid{s'},\Conid{Just}\;\Varid{s'});{}\<[E]%
\\
\>[10]{}\Varid{return}\;(\Varid{a},\Varid{mc'})\mskip1.5mu\}{}\<[E]%
\\
\>[B]{}\mathrel{=}{}\<[BE]%
\>[6]{}\mbox{\commentbegin  Definition  \commentend}{}\<[E]%
\\
\>[B]{}\hsindent{4}{}\<[4]%
\>[4]{}\mathbf{do}\;\{\mskip1.5mu {}\<[10]%
\>[10]{}(\Varid{b'},\Varid{mc'})\leftarrow \Varid{sl}\mathord{.}\Varid{mput}_\mathrm{R}\;(\Varid{a},\Conid{Nothing});\Varid{return}\;(\Varid{a},\Varid{mc'})\mskip1.5mu\}{}\<[E]%
\ColumnHook
\end{hscode}\resethooks

If \ensuremath{\Varid{mc}\mathrel{=}\Conid{Just}\;\Varid{c}} then we reason as follows:
\begin{hscode}\SaveRestoreHook
\column{B}{@{}>{\hspre}c<{\hspost}@{}}%
\column{BE}{@{}l@{}}%
\column{4}{@{}>{\hspre}l<{\hspost}@{}}%
\column{6}{@{}>{\hspre}l<{\hspost}@{}}%
\column{10}{@{}>{\hspre}l<{\hspost}@{}}%
\column{37}{@{}>{\hspre}l<{\hspost}@{}}%
\column{E}{@{}>{\hspre}l<{\hspost}@{}}%
\>[4]{}\mathbf{do}\;\{\mskip1.5mu {}\<[10]%
\>[10]{}(\Varid{b'},\Varid{mc'})\leftarrow \Varid{sl}\mathord{.}\Varid{mput}_\mathrm{R}\;(\Varid{a},\Conid{Just}\;\Varid{s});\Varid{sl}\mathord{.}\Varid{mput}_\mathrm{L}\;(\Varid{b'},\Varid{mc'})\mskip1.5mu\}{}\<[E]%
\\
\>[B]{}\mathrel{=}{}\<[BE]%
\>[6]{}\mbox{\commentbegin  Definition    \commentend}{}\<[E]%
\\
\>[B]{}\hsindent{4}{}\<[4]%
\>[4]{}\mathbf{do}\;\{\mskip1.5mu {}\<[10]%
\>[10]{}\Varid{s'}\leftarrow \Varid{sp}\mathord{.}\Varid{left}\mathord{.}\Varid{mput}\;\Varid{s}\;\Varid{a};(\Varid{b'},\Varid{mc'})\leftarrow (\Varid{sp}\mathord{.}\Varid{right}\mathord{.}\Varid{mget}\;\Varid{s'},\Conid{Just}\;\Varid{s'});{}\<[E]%
\\
\>[10]{}\Varid{sl}\mathord{.}\Varid{mput}_\mathrm{L}\;(\Varid{b'},\Varid{mc'})\mskip1.5mu\}{}\<[E]%
\\
\>[B]{}\mathrel{=}{}\<[BE]%
\>[6]{}\mbox{\commentbegin  monad unit   \commentend}{}\<[E]%
\\
\>[B]{}\hsindent{4}{}\<[4]%
\>[4]{}\mathbf{do}\;\{\mskip1.5mu {}\<[10]%
\>[10]{}\Varid{s'}\leftarrow \Varid{sp}\mathord{.}\Varid{left}\mathord{.}\Varid{mput}\;\Varid{s}\;\Varid{a};\Varid{sl}\mathord{.}\Varid{mput}_\mathrm{L}\;(\Varid{sp}\mathord{.}\Varid{right}\mathord{.}\Varid{mget}\;\Varid{s'},\Conid{Just}\;\Varid{s'})\mskip1.5mu\}{}\<[E]%
\\
\>[B]{}\mathrel{=}{}\<[BE]%
\>[6]{}\mbox{\commentbegin  definition   \commentend}{}\<[E]%
\\
\>[B]{}\hsindent{4}{}\<[4]%
\>[4]{}\mathbf{do}\;\{\mskip1.5mu {}\<[10]%
\>[10]{}\Varid{s'}\leftarrow \Varid{sp}\mathord{.}\Varid{left}\mathord{.}\Varid{mput}\;\Varid{s}\;\Varid{a};\Varid{s''}\leftarrow \Varid{sp}\mathord{.}\Varid{right}\mathord{.}\Varid{mput}\;\Varid{s'}\;(\Varid{sp}\mathord{.}\Varid{right}\mathord{.}\Varid{mget}\;\Varid{s'});{}\<[E]%
\\
\>[10]{}\Varid{return}\;(\Varid{sp}\mathord{.}\Varid{left}\mathord{.}\Varid{mget}\;\Varid{s''},\Conid{Just}\;\Varid{s''})\mskip1.5mu\}{}\<[E]%
\\
\>[B]{}\mathrel{=}{}\<[BE]%
\>[6]{}\mbox{\commentbegin  \ensuremath{\mathsf{(MGetPut)}}   \commentend}{}\<[E]%
\\
\>[B]{}\hsindent{4}{}\<[4]%
\>[4]{}\mathbf{do}\;\{\mskip1.5mu {}\<[10]%
\>[10]{}\Varid{s'}\leftarrow \Varid{sp}\mathord{.}\Varid{left}\mathord{.}\Varid{mput}\;\Varid{s}\;\Varid{a};\Varid{s''}\leftarrow \Varid{return}\;\Varid{s'};{}\<[E]%
\\
\>[10]{}\Varid{return}\;(\Varid{sp}\mathord{.}\Varid{left}\mathord{.}\Varid{mget}\;\Varid{s''},\Conid{Just}\;\Varid{s''})\mskip1.5mu\}{}\<[E]%
\\
\>[B]{}\mathrel{=}{}\<[BE]%
\>[6]{}\mbox{\commentbegin  monad unit  \commentend}{}\<[E]%
\\
\>[B]{}\hsindent{4}{}\<[4]%
\>[4]{}\mathbf{do}\;\{\mskip1.5mu {}\<[10]%
\>[10]{}\Varid{s'}\leftarrow \Varid{sp}\mathord{.}\Varid{left}\mathord{.}\Varid{mput}\;\Varid{s}\;\Varid{a};{}\<[37]%
\>[37]{}\Varid{return}\;(\Varid{sp}\mathord{.}\Varid{left}\mathord{.}\Varid{mget}\;\Varid{s'},\Conid{Just}\;\Varid{s'})\mskip1.5mu\}{}\<[E]%
\\
\>[B]{}\mathrel{=}{}\<[BE]%
\>[6]{}\mbox{\commentbegin  \ensuremath{\mathsf{(MPutGet)}}  \commentend}{}\<[E]%
\\
\>[B]{}\hsindent{4}{}\<[4]%
\>[4]{}\mathbf{do}\;\{\mskip1.5mu {}\<[10]%
\>[10]{}\Varid{s'}\leftarrow \Varid{sp}\mathord{.}\Varid{left}\mathord{.}\Varid{mput}\;\Varid{s}\;\Varid{a};{}\<[37]%
\>[37]{}\Varid{return}\;(\Varid{a},\Conid{Just}\;\Varid{s'})\mskip1.5mu\}{}\<[E]%
\\
\>[B]{}\mathrel{=}{}\<[BE]%
\>[6]{}\mbox{\commentbegin  monad unit  \commentend}{}\<[E]%
\\
\>[B]{}\hsindent{4}{}\<[4]%
\>[4]{}\mathbf{do}\;\{\mskip1.5mu {}\<[10]%
\>[10]{}\Varid{s'}\leftarrow \Varid{sp}\mathord{.}\Varid{left}\mathord{.}\Varid{mput}\;\Varid{s}\;\Varid{a};(\Varid{b'},\Varid{mc'})\leftarrow (\Varid{sp}\mathord{.}\Varid{right}\mathord{.}\Varid{get}\;\Varid{s'},\Conid{Just}\;\Varid{s'});{}\<[E]%
\\
\>[10]{}\Varid{return}\;(\Varid{a},\Varid{mc'})\mskip1.5mu\}{}\<[E]%
\\
\>[B]{}\mathrel{=}{}\<[BE]%
\>[6]{}\mbox{\commentbegin  Definition   \commentend}{}\<[E]%
\\
\>[B]{}\hsindent{4}{}\<[4]%
\>[4]{}\mathbf{do}\;\{\mskip1.5mu {}\<[10]%
\>[10]{}(\Varid{b'},\Varid{mc'})\leftarrow \Varid{sl}\mathord{.}\Varid{mput}_\mathrm{R}\;(\Varid{a},\Conid{Just}\;\Varid{c});\Varid{return}\;(\Varid{a},\Varid{mc'})\mskip1.5mu\}{}\<[E]%
\ColumnHook
\end{hscode}\resethooks
\endswithdisplay
\end{proof}


\restatableTheorem{thm:smlens2span-wb}
\begin{thm:smlens2span-wb}
    If \ensuremath{\Varid{sl}\mathbin{::}\Conid{SMLens}\;\Conid{Id}\;\Conid{C}\;\Conid{A}\;\Conid{B}} is well-behaved,
  then \ensuremath{\Varid{smlens2span}\;\Varid{sl}} is also well-behaved, with state space \ensuremath{\Conid{S}}
  consisting of the consistent triples of \ensuremath{\Varid{sl}}.
\end{thm:smlens2span-wb}
\begin{proof}
First we show that, given a symmetric lens \ensuremath{\Varid{sl}}, the operations of
\ensuremath{\Varid{sp}\mathrel{=}\Varid{smlens2span}\;\Varid{sl}} preserve consistency of the state.  Assume \ensuremath{(\Varid{a},\Varid{b},\Varid{c})}
is consistent. To show that \ensuremath{\Varid{sp}\mathord{.}\Varid{left}\mathord{.}\Varid{mput}\;(\Varid{a},\Varid{b},\Varid{c})\;\Varid{a'}} 
is consistent for any
\ensuremath{\Varid{a'}}, we have to show that \ensuremath{(\Varid{a'},\Varid{b'},\Varid{c'})} is consistent, where \ensuremath{\Varid{a'}} is
arbitrary and \ensuremath{\Varid{return}\;(\Varid{b'},\Varid{c'})\mathrel{=}\Varid{sl}\mathord{.}\Varid{mput}_\mathrm{R}\;(\Varid{a'},\Varid{c})}. For one half of consistency, 
we have:
\begin{hscode}\SaveRestoreHook
\column{B}{@{}>{\hspre}c<{\hspost}@{}}%
\column{BE}{@{}l@{}}%
\column{4}{@{}>{\hspre}l<{\hspost}@{}}%
\column{5}{@{}>{\hspre}l<{\hspost}@{}}%
\column{E}{@{}>{\hspre}l<{\hspost}@{}}%
\>[4]{}\Varid{sl}\mathord{.}\Varid{mput}_\mathrm{L}\;(\Varid{b'},\Varid{c'}){}\<[E]%
\\
\>[B]{}\mathrel{=}{}\<[BE]%
\>[5]{}\mbox{\commentbegin  \ensuremath{\Varid{sl}\mathord{.}\Varid{mput}_\mathrm{R}\;(\Varid{a'},\Varid{c})\mathrel{=}\Varid{return}\;(\Varid{b'},\Varid{c'})}, and \ensuremath{\mathsf{(PutRLM)}}  \commentend}{}\<[E]%
\\
\>[B]{}\hsindent{4}{}\<[4]%
\>[4]{}\Varid{return}\;(\Varid{a'},\Varid{c'}){}\<[E]%
\ColumnHook
\end{hscode}\resethooks
The proof that \ensuremath{\Varid{sl}\mathord{.}\Varid{mput}_\mathrm{R}\;(\Varid{a'},\Varid{c'})\mathrel{=}\Varid{return}\;(\Varid{b'},\Varid{c'})} is symmetric.
\begin{hscode}\SaveRestoreHook
\column{B}{@{}>{\hspre}c<{\hspost}@{}}%
\column{BE}{@{}l@{}}%
\column{3}{@{}>{\hspre}l<{\hspost}@{}}%
\column{5}{@{}>{\hspre}l<{\hspost}@{}}%
\column{E}{@{}>{\hspre}l<{\hspost}@{}}%
\>[3]{}\Varid{sl}\mathord{.}\Varid{mput}_\mathrm{R}\;(\Varid{a'},\Varid{c'}){}\<[E]%
\\
\>[B]{}\mathrel{=}{}\<[BE]%
\>[5]{}\mbox{\commentbegin  above, and \ensuremath{\mathsf{(PutLRM)}}  \commentend}{}\<[E]%
\\
\>[B]{}\hsindent{3}{}\<[3]%
\>[3]{}\Varid{return}\;(\Varid{b'},\Varid{c'}){}\<[E]%
\ColumnHook
\end{hscode}\resethooks
as required.
The proof that \ensuremath{\Varid{sp}\mathord{.}\Varid{right}\mathord{.}\Varid{mput}\;(\Varid{a},\Varid{b},\Varid{c})\;\Varid{b'}} is consistent is dual.

We will now show that \ensuremath{\Varid{sp}\mathrel{=}\Varid{smlens2span}\;\Varid{sl}} is a well-behaved span for any
symmetric lens \ensuremath{\Varid{sl}}.
For \ensuremath{\mathsf{(MGetPut)}}, we proceed as follows:
\begin{hscode}\SaveRestoreHook
\column{B}{@{}>{\hspre}c<{\hspost}@{}}%
\column{BE}{@{}l@{}}%
\column{4}{@{}>{\hspre}l<{\hspost}@{}}%
\column{6}{@{}>{\hspre}l<{\hspost}@{}}%
\column{E}{@{}>{\hspre}l<{\hspost}@{}}%
\>[4]{}\Varid{sp}\mathord{.}\Varid{left}\mathord{.}\Varid{mput}\;(\Varid{a},\Varid{b},\Varid{c})\;(\Varid{sp}\mathord{.}\Varid{left}\mathord{.}\Varid{mget}\;(\Varid{a},\Varid{b},\Varid{c})){}\<[E]%
\\
\>[B]{}\mathrel{=}{}\<[BE]%
\>[6]{}\mbox{\commentbegin  Definition  \commentend}{}\<[E]%
\\
\>[B]{}\hsindent{4}{}\<[4]%
\>[4]{}\mathbf{do}\;\{\mskip1.5mu (\Varid{b'},\Varid{c'})\leftarrow \Varid{sl}\mathord{.}\Varid{mput}_\mathrm{R}\;(\Varid{a},\Varid{c});\Varid{return}\;(\Varid{a},\Varid{b'},\Varid{c'})\mskip1.5mu\}{}\<[E]%
\\
\>[B]{}\mathrel{=}{}\<[BE]%
\>[6]{}\mbox{\commentbegin  Consistency of \ensuremath{(\Varid{a},\Varid{b},\Varid{c})}  \commentend}{}\<[E]%
\\
\>[B]{}\hsindent{4}{}\<[4]%
\>[4]{}\mathbf{do}\;\{\mskip1.5mu (\Varid{b'},\Varid{c'})\leftarrow \Varid{return}\;(\Varid{b},\Varid{c});\Varid{return}\;(\Varid{a},\Varid{b'},\Varid{c'})\mskip1.5mu\}{}\<[E]%
\\
\>[B]{}\mathrel{=}{}\<[BE]%
\>[6]{}\mbox{\commentbegin  monad unit  \commentend}{}\<[E]%
\\
\>[B]{}\hsindent{4}{}\<[4]%
\>[4]{}\Varid{return}\;(\Varid{a},\Varid{b},\Varid{c}){}\<[E]%
\ColumnHook
\end{hscode}\resethooks

For \ensuremath{\mathsf{(MPutGet)}}, we have:
\begin{hscode}\SaveRestoreHook
\column{B}{@{}>{\hspre}c<{\hspost}@{}}%
\column{BE}{@{}l@{}}%
\column{4}{@{}>{\hspre}l<{\hspost}@{}}%
\column{6}{@{}>{\hspre}l<{\hspost}@{}}%
\column{10}{@{}>{\hspre}l<{\hspost}@{}}%
\column{40}{@{}>{\hspre}l<{\hspost}@{}}%
\column{E}{@{}>{\hspre}l<{\hspost}@{}}%
\>[4]{}\mathbf{do}\;\{\mskip1.5mu \Varid{s'}\leftarrow \Varid{sp}\mathord{.}\Varid{left}\mathord{.}\Varid{put}\;(\Varid{a},\Varid{b},\Varid{c})\;\Varid{a'};\Varid{return}\;(\Varid{s'},\Varid{sp}\mathord{.}\Varid{left}\mathord{.}\Varid{mget}\;\Varid{s'})\mskip1.5mu\}{}\<[E]%
\\
\>[B]{}\mathrel{=}{}\<[BE]%
\>[6]{}\mbox{\commentbegin  Definition  \commentend}{}\<[E]%
\\
\>[B]{}\hsindent{4}{}\<[4]%
\>[4]{}\mathbf{do}\;\{\mskip1.5mu {}\<[10]%
\>[10]{}(\Varid{b'},\Varid{c'})\leftarrow \Varid{sl}\mathord{.}\Varid{mput}_\mathrm{R}\;(\Varid{a'},\Varid{c});\Varid{s'}\leftarrow \Varid{return}\;(\Varid{a'},\Varid{b'},\Varid{c'});{}\<[E]%
\\
\>[10]{}\Varid{return}\;(\Varid{s'},\Varid{sp}\mathord{.}\Varid{left}\mathord{.}\Varid{mget}\;\Varid{s'})\mskip1.5mu\}{}\<[E]%
\\
\>[B]{}\mathrel{=}{}\<[BE]%
\>[6]{}\mbox{\commentbegin  monad unit  \commentend}{}\<[E]%
\\
\>[B]{}\hsindent{4}{}\<[4]%
\>[4]{}\mathbf{do}\;\{\mskip1.5mu {}\<[10]%
\>[10]{}(\Varid{b'},\Varid{c'})\leftarrow \Varid{sl}\mathord{.}\Varid{mput}_\mathrm{R}\;(\Varid{a'},\Varid{c});{}\<[E]%
\\
\>[10]{}\Varid{return}\;((\Varid{a'},\Varid{b'},\Varid{c'}),\Varid{sp}\mathord{.}\Varid{left}\mathord{.}\Varid{mget}\;(\Varid{a'},\Varid{b'},\Varid{c'}))\mskip1.5mu\}{}\<[E]%
\\
\>[B]{}\mathrel{=}{}\<[BE]%
\>[6]{}\mbox{\commentbegin  Definition  \commentend}{}\<[E]%
\\
\>[B]{}\hsindent{4}{}\<[4]%
\>[4]{}\mathbf{do}\;\{\mskip1.5mu (\Varid{b'},\Varid{c'})\leftarrow \Varid{sl}\mathord{.}\Varid{mput}_\mathrm{R}\;(\Varid{a'},\Varid{c});{}\<[40]%
\>[40]{}\Varid{return}\;((\Varid{a'},\Varid{b'},\Varid{c'}),\Varid{a'})\mskip1.5mu\}{}\<[E]%
\\
\>[B]{}\mathrel{=}{}\<[BE]%
\>[6]{}\mbox{\commentbegin  monad unit  \commentend}{}\<[E]%
\\
\>[B]{}\hsindent{4}{}\<[4]%
\>[4]{}\mathbf{do}\;\{\mskip1.5mu (\Varid{b'},\Varid{c'})\leftarrow \Varid{sl}\mathord{.}\Varid{mput}_\mathrm{R}\;(\Varid{a'},\Varid{c});\Varid{s'}\leftarrow \Varid{return}\;(\Varid{a'},\Varid{b'},\Varid{c'});\Varid{return}\;(\Varid{s'},\Varid{a'})\mskip1.5mu\}{}\<[E]%
\\
\>[B]{}\mathrel{=}{}\<[BE]%
\>[6]{}\mbox{\commentbegin  Definition  \commentend}{}\<[E]%
\\
\>[B]{}\hsindent{4}{}\<[4]%
\>[4]{}\mathbf{do}\;\{\mskip1.5mu {}\<[10]%
\>[10]{}\Varid{s'}\leftarrow \Varid{sp}\mathord{.}\Varid{left}\mathord{.}\Varid{put}\;(\Varid{a},\Varid{b},\Varid{c})\;\Varid{a'};\Varid{return}\;(\Varid{s'},\Varid{a'})\mskip1.5mu\}{}\<[E]%
\ColumnHook
\end{hscode}\resethooks

The proof for \ensuremath{\mathsf{(MCreateGet)}} is similar. 
\begin{hscode}\SaveRestoreHook
\column{B}{@{}>{\hspre}c<{\hspost}@{}}%
\column{BE}{@{}l@{}}%
\column{4}{@{}>{\hspre}l<{\hspost}@{}}%
\column{6}{@{}>{\hspre}l<{\hspost}@{}}%
\column{10}{@{}>{\hspre}l<{\hspost}@{}}%
\column{39}{@{}>{\hspre}l<{\hspost}@{}}%
\column{E}{@{}>{\hspre}l<{\hspost}@{}}%
\>[4]{}\mathbf{do}\;\{\mskip1.5mu \Varid{s}\leftarrow \Varid{sp}\mathord{.}\Varid{left}\mathord{.}\Varid{create}\;\Varid{a};\Varid{return}\;(\Varid{s},\Varid{sp}\mathord{.}\Varid{left}\mathord{.}\Varid{mget}\;\Varid{s})\mskip1.5mu\}{}\<[E]%
\\
\>[B]{}\mathrel{=}{}\<[BE]%
\>[6]{}\mbox{\commentbegin  Definition  \commentend}{}\<[E]%
\\
\>[B]{}\hsindent{4}{}\<[4]%
\>[4]{}\mathbf{do}\;\{\mskip1.5mu {}\<[10]%
\>[10]{}(\Varid{b'},\Varid{c'})\leftarrow \Varid{sl}\mathord{.}\Varid{mput}_\mathrm{R}\;(\Varid{a},\Varid{c});\Varid{s}\leftarrow \Varid{return}\;(\Varid{a},\Varid{b'},\Varid{c'});{}\<[E]%
\\
\>[10]{}\Varid{return}\;(\Varid{s},\Varid{sp}\mathord{.}\Varid{left}\mathord{.}\Varid{mget}\;\Varid{s})\mskip1.5mu\}{}\<[E]%
\\
\>[B]{}\mathrel{=}{}\<[BE]%
\>[6]{}\mbox{\commentbegin  monad unit  \commentend}{}\<[E]%
\\
\>[B]{}\hsindent{4}{}\<[4]%
\>[4]{}\mathbf{do}\;\{\mskip1.5mu {}\<[10]%
\>[10]{}(\Varid{b'},\Varid{c'})\leftarrow \Varid{sl}\mathord{.}\Varid{mput}_\mathrm{R}\;(\Varid{a},\Varid{c});{}\<[E]%
\\
\>[10]{}\Varid{return}\;((\Varid{a},\Varid{b'},\Varid{c'}),\Varid{sp}\mathord{.}\Varid{left}\mathord{.}\Varid{mget}\;(\Varid{a},\Varid{b'},\Varid{c'}))\mskip1.5mu\}{}\<[E]%
\\
\>[B]{}\mathrel{=}{}\<[BE]%
\>[6]{}\mbox{\commentbegin  Definition  \commentend}{}\<[E]%
\\
\>[B]{}\hsindent{4}{}\<[4]%
\>[4]{}\mathbf{do}\;\{\mskip1.5mu (\Varid{b'},\Varid{c'})\leftarrow \Varid{sl}\mathord{.}\Varid{mput}_\mathrm{R}\;(\Varid{a},\Varid{c});{}\<[39]%
\>[39]{}\Varid{return}\;((\Varid{a},\Varid{b'},\Varid{c'}),\Varid{a})\mskip1.5mu\}{}\<[E]%
\\
\>[B]{}\mathrel{=}{}\<[BE]%
\>[6]{}\mbox{\commentbegin  monad unit  \commentend}{}\<[E]%
\\
\>[B]{}\hsindent{4}{}\<[4]%
\>[4]{}\mathbf{do}\;\{\mskip1.5mu (\Varid{b'},\Varid{c'})\leftarrow \Varid{sl}\mathord{.}\Varid{mput}_\mathrm{R}\;(\Varid{a},\Varid{c});\Varid{s}\leftarrow \Varid{return}\;(\Varid{a},\Varid{b'},\Varid{c'});\Varid{return}\;(\Varid{s},\Varid{a})\mskip1.5mu\}{}\<[E]%
\\
\>[B]{}\mathrel{=}{}\<[BE]%
\>[6]{}\mbox{\commentbegin  Definition  \commentend}{}\<[E]%
\\
\>[B]{}\hsindent{4}{}\<[4]%
\>[4]{}\mathbf{do}\;\{\mskip1.5mu {}\<[10]%
\>[10]{}\Varid{s}\leftarrow \Varid{sp}\mathord{.}\Varid{left}\mathord{.}\Varid{create}\;\Varid{a};\Varid{return}\;(\Varid{a},\Varid{s})\mskip1.5mu\}{}\<[E]%
\ColumnHook
\end{hscode}\resethooks

\endswithdisplay
\end{proof}

\section{Proofs for Section~\ref{sec:equiv}}

\restatableLemma{lem:cospan2span-pullback}
\begin{lem:cospan2span-pullback}
  Suppose \ensuremath{\Varid{l}_{1}\mathbin{::}\Conid{A}\mathbin{\leadsto}\Conid{B}} and \ensuremath{\Varid{l}_{2}\mathbin{::}\Conid{C}\mathbin{\leadsto}\Conid{B}} are pure lenses.  Then \ensuremath{(\Varid{l}_{1}\mathbin{\Join}\Varid{l}_{2})\mathord{.}\Varid{left}\mathbin{;}\Varid{l}_{1}\mathrel{=}(\Varid{l}_{1}\mathbin{\Join}\Varid{l}_{2})\mathord{.}\Varid{right}\mathbin{;}\Varid{l}_{2}}.
\end{lem:cospan2span-pullback}
\begin{proof}
  Let \ensuremath{(\Varid{l},\Varid{r})\mathrel{=}\Varid{l}_{1}\mathbin{\Join}\Varid{l}_{2}}.  
  We show that each component of \ensuremath{\Varid{l}\mathbin{;}\Varid{l}_{1}} equals
  the corresponding component of \ensuremath{\Varid{r}\mathbin{;}\Varid{l}_{2}}.

For \ensuremath{\Varid{get}}:
\begin{hscode}\SaveRestoreHook
\column{B}{@{}>{\hspre}c<{\hspost}@{}}%
\column{BE}{@{}l@{}}%
\column{4}{@{}>{\hspre}l<{\hspost}@{}}%
\column{7}{@{}>{\hspre}l<{\hspost}@{}}%
\column{E}{@{}>{\hspre}l<{\hspost}@{}}%
\>[4]{}(\Varid{l}\mathbin{;}\Varid{l}_{1})\mathord{.}\Varid{get}\;(\Varid{a},\Varid{c}){}\<[E]%
\\
\>[B]{}\mathrel{=}{}\<[BE]%
\>[7]{}\mbox{\commentbegin  Definition  \commentend}{}\<[E]%
\\
\>[B]{}\hsindent{4}{}\<[4]%
\>[4]{}\Varid{l}_{1}\mathord{.}\Varid{get}\;(\Varid{l}\mathord{.}\Varid{get}\;(\Varid{a},\Varid{c})){}\<[E]%
\\
\>[B]{}\mathrel{=}{}\<[BE]%
\>[7]{}\mbox{\commentbegin  Definition  \commentend}{}\<[E]%
\\
\>[B]{}\hsindent{4}{}\<[4]%
\>[4]{}\Varid{l}_{1}\mathord{.}\Varid{get}\;\Varid{a}{}\<[E]%
\\
\>[B]{}\mathrel{=}{}\<[BE]%
\>[7]{}\mbox{\commentbegin  Consistency  \commentend}{}\<[E]%
\\
\>[B]{}\hsindent{4}{}\<[4]%
\>[4]{}\Varid{l}_{2}\mathord{.}\Varid{get}\;\Varid{c}{}\<[E]%
\\
\>[B]{}\mathrel{=}{}\<[BE]%
\>[7]{}\mbox{\commentbegin  Definition  \commentend}{}\<[E]%
\\
\>[B]{}\hsindent{4}{}\<[4]%
\>[4]{}\Varid{l}_{2}\mathord{.}\Varid{get}\;(\Varid{r}\mathord{.}\Varid{get}\;(\Varid{a},\Varid{c})){}\<[E]%
\\
\>[B]{}\mathrel{=}{}\<[BE]%
\>[7]{}\mbox{\commentbegin  Definition  \commentend}{}\<[E]%
\\
\>[B]{}\hsindent{4}{}\<[4]%
\>[4]{}(\Varid{r}\mathbin{;}\Varid{l}_{2})\mathord{.}\Varid{get}\;(\Varid{a},\Varid{c}){}\<[E]%
\ColumnHook
\end{hscode}\resethooks

For \ensuremath{\Varid{put}}:
\begin{hscode}\SaveRestoreHook
\column{B}{@{}>{\hspre}c<{\hspost}@{}}%
\column{BE}{@{}l@{}}%
\column{4}{@{}>{\hspre}l<{\hspost}@{}}%
\column{7}{@{}>{\hspre}l<{\hspost}@{}}%
\column{E}{@{}>{\hspre}l<{\hspost}@{}}%
\>[4]{}(\Varid{l}\mathbin{;}\Varid{l}_{1})\mathord{.}\Varid{put}\;(\Varid{a},\Varid{c})\;\Varid{b}{}\<[E]%
\\
\>[B]{}\mathrel{=}{}\<[BE]%
\>[7]{}\mbox{\commentbegin  Definition  \commentend}{}\<[E]%
\\
\>[B]{}\hsindent{4}{}\<[4]%
\>[4]{}\Varid{l}\mathord{.}\Varid{put}\;(\Varid{a},\Varid{c})\;(\Varid{l}_{1}\mathord{.}\Varid{put}\;(\Varid{l}\mathord{.}\Varid{get}\;(\Varid{a},\Varid{c}))\;\Varid{b}){}\<[E]%
\\
\>[B]{}\mathrel{=}{}\<[BE]%
\>[7]{}\mbox{\commentbegin  Definition  \commentend}{}\<[E]%
\\
\>[B]{}\hsindent{4}{}\<[4]%
\>[4]{}\Varid{l}\mathord{.}\Varid{put}\;(\Varid{a},\Varid{c})\;(\Varid{l}_{1}\mathord{.}\Varid{put}\;\Varid{a}\;\Varid{b}){}\<[E]%
\\
\>[B]{}\mathrel{=}{}\<[BE]%
\>[7]{}\mbox{\commentbegin  Definition  \commentend}{}\<[E]%
\\
\>[B]{}\hsindent{4}{}\<[4]%
\>[4]{}\mathbf{let}\;\Varid{a'}\mathrel{=}\Varid{l}_{1}\mathord{.}\Varid{put}\;\Varid{a}\;\Varid{b}\;\mathbf{in}{}\<[E]%
\\
\>[B]{}\hsindent{4}{}\<[4]%
\>[4]{}\mathbf{let}\;\Varid{c'}\mathrel{=}\Varid{l}_{2}\mathord{.}\Varid{put}\;\Varid{c}\;(\Varid{l}_{1}\mathord{.}\Varid{get}\;\Varid{a'})\;\mathbf{in}\;(\Varid{a'},\Varid{c'}){}\<[E]%
\\
\>[B]{}\mathrel{=}{}\<[BE]%
\>[7]{}\mbox{\commentbegin  inline \ensuremath{\mathbf{let}}  \commentend}{}\<[E]%
\\
\>[B]{}\hsindent{4}{}\<[4]%
\>[4]{}(\Varid{l}_{1}\mathord{.}\Varid{put}\;\Varid{a}\;\Varid{b},\Varid{l}_{2}\mathord{.}\Varid{put}\;\Varid{c}\;(\Varid{l}_{1}\mathord{.}\Varid{get}\;(\Varid{l}_{1}\mathord{.}\Varid{put}\;\Varid{a}\;\Varid{b}))){}\<[E]%
\\
\>[B]{}\mathrel{=}{}\<[BE]%
\>[7]{}\mbox{\commentbegin  \ensuremath{\mathsf{(PutGet)}}  \commentend}{}\<[E]%
\\
\>[B]{}\hsindent{4}{}\<[4]%
\>[4]{}(\Varid{l}_{1}\mathord{.}\Varid{put}\;\Varid{a}\;\Varid{b},\Varid{l}_{2}\mathord{.}\Varid{put}\;\Varid{c}\;\Varid{b}){}\<[E]%
\\
\>[B]{}\mathrel{=}{}\<[BE]%
\>[7]{}\mbox{\commentbegin  reverse above steps  \commentend}{}\<[E]%
\\
\>[B]{}\hsindent{4}{}\<[4]%
\>[4]{}(\Varid{r}\mathbin{;}\Varid{l}_{2})\mathord{.}\Varid{put}\;(\Varid{a},\Varid{c})\;\Varid{b}{}\<[E]%
\ColumnHook
\end{hscode}\resethooks

Finally, for \ensuremath{\Varid{create}}:
\begin{hscode}\SaveRestoreHook
\column{B}{@{}>{\hspre}c<{\hspost}@{}}%
\column{BE}{@{}l@{}}%
\column{4}{@{}>{\hspre}l<{\hspost}@{}}%
\column{7}{@{}>{\hspre}l<{\hspost}@{}}%
\column{E}{@{}>{\hspre}l<{\hspost}@{}}%
\>[4]{}(\Varid{l}\mathbin{;}\Varid{l}_{1})\mathord{.}\Varid{create}\;\Varid{b}{}\<[E]%
\\
\>[B]{}\mathrel{=}{}\<[BE]%
\>[7]{}\mbox{\commentbegin  Definition  \commentend}{}\<[E]%
\\
\>[B]{}\hsindent{4}{}\<[4]%
\>[4]{}\Varid{l}\mathord{.}\Varid{create}\;(\Varid{l}_{1}\mathord{.}\Varid{create}\;\Varid{b}){}\<[E]%
\\
\>[B]{}\mathrel{=}{}\<[BE]%
\>[7]{}\mbox{\commentbegin  Definition  \commentend}{}\<[E]%
\\
\>[B]{}\hsindent{4}{}\<[4]%
\>[4]{}\mathbf{let}\;\Varid{c}\mathrel{=}\Varid{l}_{2}\mathord{.}\Varid{create}\;(\Varid{l}_{1}\mathord{.}\Varid{get}\;(\Varid{l}_{1}\mathord{.}\Varid{create}\;\Varid{b}))\;\mathbf{in}\;(\Varid{l}_{1}\mathord{.}\Varid{create}\;\Varid{b},\Varid{c}){}\<[E]%
\\
\>[B]{}\mathrel{=}{}\<[BE]%
\>[7]{}\mbox{\commentbegin  \ensuremath{\mathsf{(CreateGet)}}  \commentend}{}\<[E]%
\\
\>[B]{}\hsindent{4}{}\<[4]%
\>[4]{}\mathbf{let}\;\Varid{c}\mathrel{=}\Varid{l}_{2}\mathord{.}\Varid{create}\;\Varid{b}\;\mathbf{in}\;(\Varid{l}_{1}\mathord{.}\Varid{create}\;\Varid{b},\Varid{c}){}\<[E]%
\\
\>[B]{}\mathrel{=}{}\<[BE]%
\>[7]{}\mbox{\commentbegin  Inline \ensuremath{\mathbf{let}}  \commentend}{}\<[E]%
\\
\>[B]{}\hsindent{4}{}\<[4]%
\>[4]{}(\Varid{l}_{1}\mathord{.}\Varid{create}\;\Varid{b},\Varid{l}_{2}\mathord{.}\Varid{create}\;\Varid{b}){}\<[E]%
\\
\>[B]{}\mathrel{=}{}\<[BE]%
\>[7]{}\mbox{\commentbegin  reverse above steps  \commentend}{}\<[E]%
\\
\>[B]{}\hsindent{4}{}\<[4]%
\>[4]{}(\Varid{r}\mathbin{;}\Varid{l}_{2})\mathord{.}\Varid{create}\;\Varid{b}{}\<[E]%
\ColumnHook
\end{hscode}\resethooks
\end{proof}

\restatableTheorem{thm:jr-implies-bisim}
\begin{thm:jr-implies-bisim}
Given \ensuremath{\Varid{sp}_{\mathrm{1}}\mathbin{::}\monadic{\Conid{A}\mathbin{\reflectbox{$\rightsquigarrow$}}\Conid{S}_{1}\rightsquigarrow\Conid{B}}{\Conid{M}},\Varid{sp}_{\mathrm{2}}\mathbin{::}\monadic{\Conid{A}\mathbin{\reflectbox{$\rightsquigarrow$}}\Conid{S}_{2}\rightsquigarrow\Conid{B}}{\Conid{M}}}, if \ensuremath{\Varid{sp}_{\mathrm{1}}\equiv_{\mathrm{s}} \Varid{sp}_{\mathrm{2}}} then \ensuremath{\Varid{sp}_{\mathrm{1}}\equiv_{\mathrm{b}} \Varid{sp}_{\mathrm{2}}}.
\end{thm:jr-implies-bisim}
\begin{proof}
  We give the details for the case \ensuremath{\Varid{sp}_{\mathrm{1}}\curvearrowright \Varid{sp}_{\mathrm{2}}}.  First,
  write \ensuremath{(\Varid{l}_{1},\Varid{r}_{1})\mathrel{=}\Varid{sp}_{\mathrm{1}}} and \ensuremath{(\Varid{l}_{2},\Varid{r}_{2})\mathrel{=}\Varid{sp}_{\mathrm{2}}}, and suppose \ensuremath{\Varid{l}\mathbin{::}\Conid{S}_{1}\mathbin{\leadsto}\Conid{S}_{2}} is a lens satisfying \ensuremath{\Varid{l}_{1}\mathrel{=}\Varid{l}\mathbin{;}\Varid{l}_{2}} and \ensuremath{\Varid{r}_{1}\mathrel{=}\Varid{l}\mathbin{;}\Varid{r}_{2}}.  

We need to define a bisimulation consisting of a set \ensuremath{\Conid{R}\subseteq\Conid{S}_{1} \times \Conid{S}_{2}} and a span
\ensuremath{\Varid{sp}\mathrel{=}(\Varid{l}_{0},\Varid{r}_{0})\mathbin{::}\monadic{\Conid{A}\mathbin{\reflectbox{$\rightsquigarrow$}}\Conid{R}\rightsquigarrow\Conid{B}}{\Conid{M}}} such that \ensuremath{\Varid{fst}} is a base map
from \ensuremath{\Varid{sp}} to \ensuremath{\Varid{sp}_{\mathrm{1}}}
and \ensuremath{\Varid{snd}} is a base map from \ensuremath{\Varid{sp}} to \ensuremath{\Varid{sp}_{\mathrm{2}}}. We take \ensuremath{\Conid{R}\mathrel{=}\{\mskip1.5mu (\Varid{s}_{1},\Varid{s}_{2})\mid \Varid{s}_{2}\mathrel{=}\Varid{l}\mathord{.}\Varid{get}\;(\Varid{s}_{1})\mskip1.5mu\}} and  proceed as follows:
\begin{hscode}\SaveRestoreHook
\column{B}{@{}>{\hspre}l<{\hspost}@{}}%
\column{5}{@{}>{\hspre}l<{\hspost}@{}}%
\column{27}{@{}>{\hspre}l<{\hspost}@{}}%
\column{E}{@{}>{\hspre}l<{\hspost}@{}}%
\>[5]{}\Varid{l}_{0}{}\<[27]%
\>[27]{}\mathbin{::}\monadic{\Conid{R}\rightsquigarrow\Conid{A}}{\Conid{M}}{}\<[E]%
\\
\>[5]{}\Varid{l}_{0}\mathord{.}\Varid{mget}\;(\Varid{s}_{1},\Varid{s}_{2}){}\<[27]%
\>[27]{}\mathrel{=}\Varid{l}_{1}\mathord{.}\Varid{mget}\;\Varid{s}_{1}{}\<[E]%
\\
\>[5]{}\Varid{l}_{0}\mathord{.}\Varid{mput}\;(\Varid{s}_{1},\Varid{s}_{2})\;\Varid{a}{}\<[27]%
\>[27]{}\mathrel{=}\mathbf{do}\;\{\mskip1.5mu \Varid{s}_{1}'\leftarrow \Varid{l}_{1}\mathord{.}\Varid{mput}\;\Varid{s}_{1}\;\Varid{a};\Varid{return}\;(\Varid{s}_{1}',\Varid{l}\mathord{.}\Varid{get}\;\Varid{s}_{1}')\mskip1.5mu\}{}\<[E]%
\\
\>[5]{}\Varid{l}_{0}\mathord{.}\Varid{mcreate}\;\Varid{a}{}\<[27]%
\>[27]{}\mathrel{=}\mathbf{do}\;\{\mskip1.5mu \Varid{s}_{1}\leftarrow \Varid{l}_{1}\mathord{.}\Varid{mcreate}\;\Varid{a};\Varid{return}\;(\Varid{s}_{1},\Varid{l}\mathord{.}\Varid{get}\;\Varid{s}_{1})\mskip1.5mu\}{}\<[E]%
\\
\>[5]{}\Varid{r}_{0}{}\<[27]%
\>[27]{}\mathbin{::}\monadic{\Conid{R}\rightsquigarrow\Conid{B}}{\Conid{M}}{}\<[E]%
\\
\>[5]{}\Varid{r}_{0}\mathord{.}\Varid{mget}\;(\Varid{s}_{1},\Varid{s}_{2}){}\<[27]%
\>[27]{}\mathrel{=}\Varid{r}_{1}\mathord{.}\Varid{mget}\;\Varid{s}_{1}{}\<[E]%
\\
\>[5]{}\Varid{r}_{0}\mathord{.}\Varid{mput}\;(\Varid{s}_{1},\Varid{s}_{2})\;\Varid{b}{}\<[27]%
\>[27]{}\mathrel{=}\mathbf{do}\;\{\mskip1.5mu \Varid{s}_{1}'\leftarrow \Varid{r}_{1}\mathord{.}\Varid{mput}\;\Varid{s}_{1}\;\Varid{b};\Varid{return}\;(\Varid{s}_{1}',\Varid{l}\mathord{.}\Varid{get}\;\Varid{s}_{1}')\mskip1.5mu\}{}\<[E]%
\\
\>[5]{}\Varid{r}_{0}\mathord{.}\Varid{mcreate}\;\Varid{b}{}\<[27]%
\>[27]{}\mathrel{=}\mathbf{do}\;\{\mskip1.5mu \Varid{s}_{1}\leftarrow \Varid{r}_{1}\mathord{.}\Varid{mcreate}\;\Varid{a};\Varid{return}\;(\Varid{s}_{1},\Varid{l}\mathord{.}\Varid{get}\;\Varid{s}_{1})\mskip1.5mu\}{}\<[E]%
\ColumnHook
\end{hscode}\resethooks
We must now show that \ensuremath{\Varid{l}_{0}} and \ensuremath{\Varid{r}_{0}} are well-behaved (full) lenses, and that the projections \ensuremath{\Varid{fst}} and \ensuremath{\Varid{snd}} map \ensuremath{\Varid{sp}\mathrel{=}(\Varid{l}_{0},\Varid{r}_{0})} to \ensuremath{\Varid{sp}_{\mathrm{1}}} and \ensuremath{\Varid{sp}_{\mathrm{2}}} respectively.

We first show that \ensuremath{\Varid{l}_{0}} is well-behaved; the reasoning for \ensuremath{\Varid{r}_{0}} is symmetric.
For \ensuremath{\mathsf{(MGetPut)}} we have:
\begin{hscode}\SaveRestoreHook
\column{B}{@{}>{\hspre}c<{\hspost}@{}}%
\column{BE}{@{}l@{}}%
\column{4}{@{}>{\hspre}l<{\hspost}@{}}%
\column{7}{@{}>{\hspre}l<{\hspost}@{}}%
\column{E}{@{}>{\hspre}l<{\hspost}@{}}%
\>[4]{}\Varid{l}_{0}\mathord{.}\Varid{mput}\;(\Varid{s}_{1},\Varid{s}_{2})\;(\Varid{l}_{0}\mathord{.}\Varid{mget}\;(\Varid{s}_{1},\Varid{s}_{2})){}\<[E]%
\\
\>[B]{}\mathrel{=}{}\<[BE]%
\>[7]{}\mbox{\commentbegin  Definition  \commentend}{}\<[E]%
\\
\>[B]{}\hsindent{4}{}\<[4]%
\>[4]{}\mathbf{do}\;\{\mskip1.5mu \Varid{s}_{1}'\leftarrow \Varid{l}_{1}\mathord{.}\Varid{mput}\;\Varid{s}_{1}\;(l_1\mathord{.}\Varid{mget}\;\Varid{s}_{1});\Varid{return}\;(\Varid{s}_{1}',\Varid{l}\mathord{.}\Varid{get}\;\Varid{s}_{1}')\mskip1.5mu\}{}\<[E]%
\\
\>[B]{}\mathrel{=}{}\<[BE]%
\>[7]{}\mbox{\commentbegin  \ensuremath{\mathsf{(MPutGet)}}  \commentend}{}\<[E]%
\\
\>[B]{}\hsindent{4}{}\<[4]%
\>[4]{}\mathbf{do}\;\{\mskip1.5mu \Varid{s}_{1}'\leftarrow \Varid{return}\;\Varid{s}_{1};\Varid{return}\;(\Varid{s}_{1}',\Varid{l}\mathord{.}\Varid{get}\;\Varid{s}_{1}')\mskip1.5mu\}{}\<[E]%
\\
\>[B]{}\mathrel{=}{}\<[BE]%
\>[7]{}\mbox{\commentbegin  Monad unit  \commentend}{}\<[E]%
\\
\>[B]{}\hsindent{4}{}\<[4]%
\>[4]{}\Varid{return}\;(\Varid{s}_{1},\Varid{l}\mathord{.}\Varid{get}\;\Varid{s}_{1}){}\<[E]%
\\
\>[B]{}\mathrel{=}{}\<[BE]%
\>[7]{}\mbox{\commentbegin  \ensuremath{\Varid{s}_{2}\mathrel{=}\Varid{l}\mathord{.}\Varid{get}\;\Varid{s}_{1}}  \commentend}{}\<[E]%
\\
\>[B]{}\hsindent{4}{}\<[4]%
\>[4]{}\Varid{return}\;(\Varid{s}_{1},\Varid{s}_{2}){}\<[E]%
\ColumnHook
\end{hscode}\resethooks

For \ensuremath{\mathsf{(MPutGet)}} we have:
\begin{hscode}\SaveRestoreHook
\column{B}{@{}>{\hspre}c<{\hspost}@{}}%
\column{BE}{@{}l@{}}%
\column{4}{@{}>{\hspre}l<{\hspost}@{}}%
\column{7}{@{}>{\hspre}l<{\hspost}@{}}%
\column{10}{@{}>{\hspre}l<{\hspost}@{}}%
\column{E}{@{}>{\hspre}l<{\hspost}@{}}%
\>[4]{}\mathbf{do}\;\{\mskip1.5mu {}\<[10]%
\>[10]{}(\Varid{s}_{1}'',\Varid{s}_{2}'')\leftarrow \Varid{l}_{0}\mathord{.}\Varid{mput}\;(\Varid{s}_{1},\Varid{s}_{2})\;\Varid{a};\Varid{return}\;((\Varid{s}_{1}'',\Varid{s}_{2}''),\Varid{l}_{0}\mathord{.}\Varid{mget}\;(\Varid{s}_{1}'',\Varid{s}_{2}''))\mskip1.5mu\}{}\<[E]%
\\
\>[B]{}\mathrel{=}{}\<[BE]%
\>[7]{}\mbox{\commentbegin  Definition  \commentend}{}\<[E]%
\\
\>[B]{}\hsindent{4}{}\<[4]%
\>[4]{}\mathbf{do}\;\{\mskip1.5mu {}\<[10]%
\>[10]{}\Varid{s}_{1}'\leftarrow \Varid{l}_{1}\mathord{.}\Varid{mput}\;\Varid{s}_{1}\;\Varid{a};(\Varid{s}_{1}'',\Varid{s}_{2}'')\leftarrow \Varid{return}\;(\Varid{s}_{1}',\Varid{l}\mathord{.}\Varid{get}\;\Varid{s}_{1}');{}\<[E]%
\\
\>[10]{}\Varid{return}\;((\Varid{s}_{1}'',\Varid{s}_{2}''),\Varid{l}_{1}\mathord{.}\Varid{mget}\;\Varid{s}_{1}')\mskip1.5mu\}{}\<[E]%
\\
\>[B]{}\mathrel{=}{}\<[BE]%
\>[7]{}\mbox{\commentbegin  Monad unit  \commentend}{}\<[E]%
\\
\>[B]{}\hsindent{4}{}\<[4]%
\>[4]{}\mathbf{do}\;\{\mskip1.5mu {}\<[10]%
\>[10]{}\Varid{s}_{1}'\leftarrow \Varid{l}_{1}\mathord{.}\Varid{mput}\;\Varid{s}_{1}\;\Varid{a};\Varid{return}\;((\Varid{s}_{1}',\Varid{l}\mathord{.}\Varid{get}\;\Varid{s}_{1}'),\Varid{l}_{1}\mathord{.}\Varid{mget}\;\Varid{s}_{1}')\mskip1.5mu\}{}\<[E]%
\\
\>[B]{}\mathrel{=}{}\<[BE]%
\>[7]{}\mbox{\commentbegin  \ensuremath{\mathsf{(MPutGet)}}  \commentend}{}\<[E]%
\\
\>[B]{}\hsindent{4}{}\<[4]%
\>[4]{}\mathbf{do}\;\{\mskip1.5mu {}\<[10]%
\>[10]{}\Varid{s}_{1}'\leftarrow \Varid{l}_{1}\mathord{.}\Varid{mput}\;\Varid{s}_{1}\;\Varid{a};\Varid{return}\;((\Varid{s}_{1}',\Varid{l}\mathord{.}\Varid{get}\;\Varid{s}_{1}'),\Varid{a})\mskip1.5mu\}{}\<[E]%
\\
\>[B]{}\mathrel{=}{}\<[BE]%
\>[7]{}\mbox{\commentbegin   Monad unit  \commentend}{}\<[E]%
\\
\>[B]{}\hsindent{4}{}\<[4]%
\>[4]{}\mathbf{do}\;\{\mskip1.5mu {}\<[10]%
\>[10]{}\Varid{s}_{1}'\leftarrow \Varid{l}_{1}\mathord{.}\Varid{mput}\;\Varid{s}_{1}\;\Varid{a};(\Varid{s}_{1}'',\Varid{s}_{2}'')\leftarrow \Varid{return}\;(\Varid{s}_{1}',\Varid{l}\mathord{.}\Varid{get}\;\Varid{s}_{1}');{}\<[E]%
\\
\>[10]{}\Varid{return}\;((\Varid{s}_{1}'',\Varid{s}_{2}''),\Varid{a})\mskip1.5mu\}{}\<[E]%
\\
\>[B]{}\mathrel{=}{}\<[BE]%
\>[7]{}\mbox{\commentbegin  Definition  \commentend}{}\<[E]%
\\
\>[B]{}\hsindent{4}{}\<[4]%
\>[4]{}\mathbf{do}\;\{\mskip1.5mu {}\<[10]%
\>[10]{}(\Varid{s}_{1}'',\Varid{s}_{2}'')\leftarrow \Varid{l}_{0}\mathord{.}\Varid{mput}\;(\Varid{s}_{1},\Varid{s}_{2})\;\Varid{a};\Varid{return}\;((\Varid{s}_{1}'',\Varid{s}_{2}''),\Varid{a})\mskip1.5mu\}{}\<[E]%
\ColumnHook
\end{hscode}\resethooks

Finally, for \ensuremath{\mathsf{(MCreateGet)}} we have:
\begin{hscode}\SaveRestoreHook
\column{B}{@{}>{\hspre}c<{\hspost}@{}}%
\column{BE}{@{}l@{}}%
\column{4}{@{}>{\hspre}l<{\hspost}@{}}%
\column{7}{@{}>{\hspre}l<{\hspost}@{}}%
\column{10}{@{}>{\hspre}l<{\hspost}@{}}%
\column{E}{@{}>{\hspre}l<{\hspost}@{}}%
\>[4]{}\mathbf{do}\;\{\mskip1.5mu {}\<[10]%
\>[10]{}(\Varid{s}_{1},\Varid{s}_{2})\leftarrow \Varid{l}_{0}\mathord{.}\Varid{mcreate}\;\Varid{a};\Varid{return}\;((\Varid{s}_{1},\Varid{s}_{2}),\Varid{l}_{0}\mathord{.}\Varid{mget}\;(\Varid{s}_{1},\Varid{s}_{2}))\mskip1.5mu\}{}\<[E]%
\\
\>[B]{}\mathrel{=}{}\<[BE]%
\>[7]{}\mbox{\commentbegin  Definition  \commentend}{}\<[E]%
\\
\>[B]{}\hsindent{4}{}\<[4]%
\>[4]{}\mathbf{do}\;\{\mskip1.5mu {}\<[10]%
\>[10]{}\Varid{s}_{1}'\leftarrow \Varid{l}_{1}\mathord{.}\Varid{mcreate}\;\Varid{a};(\Varid{s}_{1},\Varid{s}_{2})\leftarrow \Varid{return}\;(\Varid{s}_{1}',\Varid{l}\mathord{.}\Varid{get}\;\Varid{s}_{1}');{}\<[E]%
\\
\>[10]{}\Varid{return}\;((\Varid{s}_{1},\Varid{s}_{2}),\Varid{l}_{1}\mathord{.}\Varid{mget}\;\Varid{s}_{1})\mskip1.5mu\}{}\<[E]%
\\
\>[B]{}\mathrel{=}{}\<[BE]%
\>[7]{}\mbox{\commentbegin  Monad unit  \commentend}{}\<[E]%
\\
\>[B]{}\hsindent{4}{}\<[4]%
\>[4]{}\mathbf{do}\;\{\mskip1.5mu {}\<[10]%
\>[10]{}\Varid{s}_{1}'\leftarrow \Varid{l}_{1}\mathord{.}\Varid{mcreate}\;\Varid{a};\Varid{return}\;((\Varid{s}_{1}',\Varid{l}\mathord{.}\Varid{get}\;\Varid{s}_{1}'),\Varid{l}_{1}\mathord{.}\Varid{mget}\;\Varid{s}_{1}')\mskip1.5mu\}{}\<[E]%
\\
\>[B]{}\mathrel{=}{}\<[BE]%
\>[7]{}\mbox{\commentbegin  \ensuremath{\mathsf{(MCreateGet)}}  \commentend}{}\<[E]%
\\
\>[B]{}\hsindent{4}{}\<[4]%
\>[4]{}\mathbf{do}\;\{\mskip1.5mu {}\<[10]%
\>[10]{}\Varid{s}_{1}'\leftarrow \Varid{l}_{1}\mathord{.}\Varid{mcreate}\;\Varid{a};\Varid{return}\;((\Varid{s}_{1}',\Varid{l}\mathord{.}\Varid{get}\;\Varid{s}_{1}'),\Varid{a})\mskip1.5mu\}{}\<[E]%
\\
\>[B]{}\mathrel{=}{}\<[BE]%
\>[7]{}\mbox{\commentbegin   Monad unit  \commentend}{}\<[E]%
\\
\>[B]{}\hsindent{4}{}\<[4]%
\>[4]{}\mathbf{do}\;\{\mskip1.5mu {}\<[10]%
\>[10]{}\Varid{s}_{1}'\leftarrow \Varid{l}_{1}\mathord{.}\Varid{mcreate}\;\Varid{a};(\Varid{s}_{1},\Varid{s}_{2})\leftarrow \Varid{return}\;(\Varid{s}_{1}',\Varid{l}\mathord{.}\Varid{get}\;\Varid{s}_{1}');{}\<[E]%
\\
\>[10]{}\Varid{return}\;((\Varid{s}_{1},\Varid{s}_{2}),\Varid{a})\mskip1.5mu\}{}\<[E]%
\\
\>[B]{}\mathrel{=}{}\<[BE]%
\>[7]{}\mbox{\commentbegin  Definition  \commentend}{}\<[E]%
\\
\>[B]{}\hsindent{4}{}\<[4]%
\>[4]{}\mathbf{do}\;\{\mskip1.5mu {}\<[10]%
\>[10]{}(\Varid{s}_{1},\Varid{s}_{2})\leftarrow \Varid{l}_{0}\mathord{.}\Varid{mcreate}\;\Varid{a};\Varid{return}\;((\Varid{s}_{1},\Varid{s}_{2}),\Varid{a})\mskip1.5mu\}{}\<[E]%
\ColumnHook
\end{hscode}\resethooks

Next, we show that \ensuremath{\Varid{fst}} is a base map from \ensuremath{\Varid{l}_{0}} to \ensuremath{\Varid{l}_{1}} and \ensuremath{\Varid{snd}} is
a base map from \ensuremath{\Varid{l}_{0}} to \ensuremath{\Varid{l}_{2}}.  It is easy to show that  \ensuremath{\Varid{fst}} is a
base map from \ensuremath{\Varid{l}_{0}} to \ensuremath{\Varid{l}_{1}} by unfolding definitions and applying of monad laws. 
To show that \ensuremath{\Varid{snd}} is a base map  from \ensuremath{\Varid{l}_{0}} to \ensuremath{\Varid{l}_{2}}, we need to verify the following three equations that show that \ensuremath{\Varid{snd}} commutes with \ensuremath{\Varid{mget}}, \ensuremath{\Varid{mput}} and \ensuremath{\Varid{mcreate}}:
  \begin{hscode}\SaveRestoreHook
\column{B}{@{}>{\hspre}l<{\hspost}@{}}%
\column{5}{@{}>{\hspre}l<{\hspost}@{}}%
\column{55}{@{}>{\hspre}l<{\hspost}@{}}%
\column{E}{@{}>{\hspre}l<{\hspost}@{}}%
\>[5]{}\Varid{l}_{0}\mathord{.}\Varid{mget}\;(\Varid{s}_{1},\Varid{s}_{2}){}\<[55]%
\>[55]{}\mathrel{=}\Varid{l}_{2}\mathord{.}\Varid{mget}\;\Varid{s}_{2}{}\<[E]%
\\
\>[5]{}\mathbf{do}\;\{\mskip1.5mu (\Varid{s}_{1}',\Varid{s}_{2}')\leftarrow \Varid{l}_{0}\mathord{.}\Varid{mput}\;(\Varid{s}_{1},\Varid{s}_{2})\;\Varid{a};\Varid{return}\;\Varid{s}_{2}'\mskip1.5mu\}{}\<[55]%
\>[55]{}\mathrel{=}\Varid{l}_{2}\mathord{.}\Varid{mput}\;\Varid{s}_{2}\;\Varid{a}{}\<[E]%
\\
\>[5]{}\mathbf{do}\;\{\mskip1.5mu (\Varid{s}_{1},\Varid{s}_{2})\leftarrow \Varid{l}_{0}\mathord{.}\Varid{mcreate}\;\Varid{a};\Varid{return}\;\Varid{s}_{2}\mskip1.5mu\}{}\<[55]%
\>[55]{}\mathrel{=}\Varid{l}_{2}\mathord{.}\Varid{mcreate}\;\Varid{a}{}\<[E]%
\ColumnHook
\end{hscode}\resethooks
For the \ensuremath{\Varid{mget}} equation:
\begin{hscode}\SaveRestoreHook
\column{B}{@{}>{\hspre}c<{\hspost}@{}}%
\column{BE}{@{}l@{}}%
\column{4}{@{}>{\hspre}l<{\hspost}@{}}%
\column{7}{@{}>{\hspre}l<{\hspost}@{}}%
\column{E}{@{}>{\hspre}l<{\hspost}@{}}%
\>[4]{}\Varid{l}_{0}\mathord{.}\Varid{mget}\;(\Varid{s}_{1},\Varid{s}_{2}){}\<[E]%
\\
\>[B]{}\mathrel{=}{}\<[BE]%
\>[7]{}\mbox{\commentbegin  Definition  \commentend}{}\<[E]%
\\
\>[B]{}\hsindent{4}{}\<[4]%
\>[4]{}\Varid{l}_{1}\mathord{.}\Varid{mget}\;\Varid{s}_{1}{}\<[E]%
\\
\>[B]{}\mathrel{=}{}\<[BE]%
\>[7]{}\mbox{\commentbegin  Assumption \ensuremath{\Varid{l};\Varid{l}_{2}\mathrel{=}\Varid{l}_{1}}  \commentend}{}\<[E]%
\\
\>[B]{}\hsindent{4}{}\<[4]%
\>[4]{}(\Varid{l}\mathbin{;}\Varid{l}_{2})\mathord{.}\Varid{mget}\;\Varid{s}_{1}{}\<[E]%
\\
\>[B]{}\mathrel{=}{}\<[BE]%
\>[7]{}\mbox{\commentbegin  Definition  \commentend}{}\<[E]%
\\
\>[B]{}\hsindent{4}{}\<[4]%
\>[4]{}\Varid{l}_{2}\mathord{.}\Varid{mget}\;(\Varid{l}\mathord{.}\Varid{get}\;\Varid{s}_{1}){}\<[E]%
\\
\>[B]{}\mathrel{=}{}\<[BE]%
\>[7]{}\mbox{\commentbegin  \ensuremath{(\Varid{s}_{1},\Varid{s}_{2})\in\Conid{R}}  \commentend}{}\<[E]%
\\
\>[B]{}\hsindent{4}{}\<[4]%
\>[4]{}\Varid{l}_{2}\mathord{.}\Varid{mget}\;\Varid{s}_{2}{}\<[E]%
\ColumnHook
\end{hscode}\resethooks
For the \ensuremath{\Varid{mput}}  equation:
\begin{hscode}\SaveRestoreHook
\column{B}{@{}>{\hspre}c<{\hspost}@{}}%
\column{BE}{@{}l@{}}%
\column{4}{@{}>{\hspre}l<{\hspost}@{}}%
\column{7}{@{}>{\hspre}l<{\hspost}@{}}%
\column{E}{@{}>{\hspre}l<{\hspost}@{}}%
\>[4]{}\mathbf{do}\;\{\mskip1.5mu (\Varid{s}_{1}',\Varid{s}_{2}')\leftarrow \Varid{l}_{0}\mathord{.}\Varid{mput}\;(\Varid{s}_{1},\Varid{s}_{2})\;\Varid{a};\Varid{return}\;\Varid{s}_{2}'\mskip1.5mu\}{}\<[E]%
\\
\>[B]{}\mathrel{=}{}\<[BE]%
\>[7]{}\mbox{\commentbegin  Definition  \commentend}{}\<[E]%
\\
\>[B]{}\hsindent{4}{}\<[4]%
\>[4]{}\mathbf{do}\;\{\mskip1.5mu \Varid{s}_{1}''\leftarrow \Varid{l}_{1}\mathord{.}\Varid{mput}\;\Varid{s}_{1}\;\Varid{a};(\Varid{s}_{1}',\Varid{s}_{2}')\leftarrow \Varid{return}\;(\Varid{s}_{1}'',\Varid{l}\mathord{.}\Varid{get}\;\Varid{s}_{1}'');\Varid{return}\;\Varid{s}_{2}'\mskip1.5mu\}{}\<[E]%
\\
\>[B]{}\mathrel{=}{}\<[BE]%
\>[7]{}\mbox{\commentbegin  Monad laws  \commentend}{}\<[E]%
\\
\>[B]{}\hsindent{4}{}\<[4]%
\>[4]{}\mathbf{do}\;\{\mskip1.5mu \Varid{s}_{1}''\leftarrow \Varid{l}_{1}\mathord{.}\Varid{mput}\;\Varid{s}_{1}\;\Varid{a};\Varid{return}\;(\Varid{l}\mathord{.}\Varid{get}\;\Varid{s}_{1}'')\mskip1.5mu\}{}\<[E]%
\\
\>[B]{}\mathrel{=}{}\<[BE]%
\>[7]{}\mbox{\commentbegin  \ensuremath{\Varid{l}\mathbin{;}\Varid{l}_{2}\mathrel{=}\Varid{l}_{1}}  \commentend}{}\<[E]%
\\
\>[B]{}\hsindent{4}{}\<[4]%
\>[4]{}\mathbf{do}\;\{\mskip1.5mu \Varid{s}_{1}''\leftarrow (\Varid{l}\mathbin{;}\Varid{l}_{2})\mathord{.}\Varid{mput}\;\Varid{s}_{1}\;\Varid{a};\Varid{return}\;(\Varid{l}\mathord{.}\Varid{get}\;\Varid{s}_{1}'')\mskip1.5mu\}{}\<[E]%
\\
\>[B]{}\mathrel{=}{}\<[BE]%
\>[7]{}\mbox{\commentbegin  Definition  \commentend}{}\<[E]%
\\
\>[B]{}\hsindent{4}{}\<[4]%
\>[4]{}\mathbf{do}\;\{\mskip1.5mu \Varid{s}_{2}''\leftarrow \Varid{l}_{2}\mathord{.}\Varid{mput}\;(\Varid{l}\mathord{.}\Varid{get}\;\Varid{s}_{1})\;\Varid{a};\Varid{s}_{1}''\leftarrow \Varid{return}\;(\Varid{l}\mathord{.}\Varid{put}\;\Varid{s}_{1}\;\Varid{s}_{2}'');\Varid{return}\;(\Varid{l}\mathord{.}\Varid{get}\;\Varid{s}_{1}'')\mskip1.5mu\}{}\<[E]%
\\
\>[B]{}\mathrel{=}{}\<[BE]%
\>[7]{}\mbox{\commentbegin  Monad laws  \commentend}{}\<[E]%
\\
\>[B]{}\hsindent{4}{}\<[4]%
\>[4]{}\mathbf{do}\;\{\mskip1.5mu \Varid{s}_{2}''\leftarrow \Varid{l}_{2}\mathord{.}\Varid{mput}\;(\Varid{l}\mathord{.}\Varid{get}\;\Varid{s}_{1})\;\Varid{a};\Varid{return}\;(\Varid{l}\mathord{.}\Varid{get}\;(\Varid{l}\mathord{.}\Varid{put}\;\Varid{s}_{1}\;\Varid{s}_{2}''))\mskip1.5mu\}{}\<[E]%
\\
\>[B]{}\mathrel{=}{}\<[BE]%
\>[7]{}\mbox{\commentbegin  \ensuremath{\mathsf{(PutGet)}}  \commentend}{}\<[E]%
\\
\>[B]{}\hsindent{4}{}\<[4]%
\>[4]{}\mathbf{do}\;\{\mskip1.5mu \Varid{s}_{2}''\leftarrow \Varid{l}_{2}\mathord{.}\Varid{mput}\;(\Varid{l}\mathord{.}\Varid{get}\;\Varid{s}_{1})\;\Varid{a};\Varid{return}\;\Varid{s}_{2}''\mskip1.5mu\}{}\<[E]%
\\
\>[B]{}\mathrel{=}{}\<[BE]%
\>[7]{}\mbox{\commentbegin  \ensuremath{(\Varid{s}_{1},\Varid{s}_{2})\in\Conid{R}} so \ensuremath{\Varid{l}\mathord{.}\Varid{get}\;\Varid{s}_{1}\mathrel{=}\Varid{s}_{2}}  \commentend}{}\<[E]%
\\
\>[B]{}\hsindent{4}{}\<[4]%
\>[4]{}\mathbf{do}\;\{\mskip1.5mu \Varid{s}_{2}''\leftarrow \Varid{l}_{2}\mathord{.}\Varid{mput}\;\Varid{s}_{2}\;\Varid{a};\Varid{return}\;\Varid{s}_{2}''\mskip1.5mu\}{}\<[E]%
\\
\>[B]{}\mathrel{=}{}\<[BE]%
\>[7]{}\mbox{\commentbegin  Monad laws  \commentend}{}\<[E]%
\\
\>[B]{}\hsindent{4}{}\<[4]%
\>[4]{}\Varid{l}_{2}\mathord{.}\Varid{mput}\;\Varid{s}_{2}\;\Varid{a}{}\<[E]%
\ColumnHook
\end{hscode}\resethooks
For the \ensuremath{\Varid{mcreate}} equation:
\begin{hscode}\SaveRestoreHook
\column{B}{@{}>{\hspre}c<{\hspost}@{}}%
\column{BE}{@{}l@{}}%
\column{4}{@{}>{\hspre}l<{\hspost}@{}}%
\column{7}{@{}>{\hspre}l<{\hspost}@{}}%
\column{E}{@{}>{\hspre}l<{\hspost}@{}}%
\>[4]{}\mathbf{do}\;\{\mskip1.5mu (\Varid{s}_{1},\Varid{s}_{2})\leftarrow \Varid{l}_{0}\mathord{.}\Varid{mcreate}\;\Varid{a};\Varid{return}\;\Varid{s}_{2}\mskip1.5mu\}{}\<[E]%
\\
\>[B]{}\mathrel{=}{}\<[BE]%
\>[7]{}\mbox{\commentbegin  Definition  \commentend}{}\<[E]%
\\
\>[B]{}\hsindent{4}{}\<[4]%
\>[4]{}\mathbf{do}\;\{\mskip1.5mu \Varid{s}_{1}'\leftarrow \Varid{l}_{1}\mathord{.}\Varid{mcreate}\;\Varid{a};(\Varid{s}_{1},\Varid{s}_{2})\leftarrow \Varid{return}\;(\Varid{s}_{1}',\Varid{l}\mathord{.}\Varid{get}\;\Varid{s}_{1}');\Varid{return}\;\Varid{s}_{2}\mskip1.5mu\}{}\<[E]%
\\
\>[B]{}\mathrel{=}{}\<[BE]%
\>[7]{}\mbox{\commentbegin  Monad laws  \commentend}{}\<[E]%
\\
\>[B]{}\hsindent{4}{}\<[4]%
\>[4]{}\mathbf{do}\;\{\mskip1.5mu \Varid{s}_{1}'\leftarrow \Varid{l}_{1}\mathord{.}\Varid{mcreate}\;\Varid{a};\Varid{return}\;(\Varid{l}\mathord{.}\Varid{get}\;\Varid{s}_{1}')\mskip1.5mu\}{}\<[E]%
\\
\>[B]{}\mathrel{=}{}\<[BE]%
\>[7]{}\mbox{\commentbegin  \ensuremath{\Varid{l}\mathbin{;}\Varid{l}_{2}\mathrel{=}\Varid{l}_{1}}  \commentend}{}\<[E]%
\\
\>[B]{}\hsindent{4}{}\<[4]%
\>[4]{}\mathbf{do}\;\{\mskip1.5mu \Varid{s}_{1}'\leftarrow (\Varid{l}\mathbin{;}\Varid{l}_{2})\mathord{.}\Varid{mcreate}\;\Varid{a};\Varid{return}\;(\Varid{l}\mathord{.}\Varid{get}\;\Varid{s}_{1}')\mskip1.5mu\}{}\<[E]%
\\
\>[B]{}\mathrel{=}{}\<[BE]%
\>[7]{}\mbox{\commentbegin  Definition  \commentend}{}\<[E]%
\\
\>[B]{}\hsindent{4}{}\<[4]%
\>[4]{}\mathbf{do}\;\{\mskip1.5mu \Varid{s}_{2}'\leftarrow \Varid{l}_{2}\mathord{.}\Varid{mcreate}\;\Varid{a};\Varid{s}_{1}'\leftarrow \Varid{return}\;(\Varid{l}\mathord{.}\Varid{create}\;\Varid{s}_{2}');\Varid{return}\;(\Varid{l}\mathord{.}\Varid{get}\;\Varid{s}_{1}')\mskip1.5mu\}{}\<[E]%
\\
\>[B]{}\mathrel{=}{}\<[BE]%
\>[7]{}\mbox{\commentbegin   Monad laws  \commentend}{}\<[E]%
\\
\>[B]{}\hsindent{4}{}\<[4]%
\>[4]{}\mathbf{do}\;\{\mskip1.5mu \Varid{s}_{2}'\leftarrow \Varid{l}_{2}\mathord{.}\Varid{mcreate}\;\Varid{a};\Varid{return}\;(\Varid{l}\mathord{.}\Varid{get}\;(\Varid{l}\mathord{.}\Varid{create}\;\Varid{s}_{2}'))\mskip1.5mu\}{}\<[E]%
\\
\>[B]{}\mathrel{=}{}\<[BE]%
\>[7]{}\mbox{\commentbegin  \ensuremath{\mathsf{(CreateGet)}}  \commentend}{}\<[E]%
\\
\>[B]{}\hsindent{4}{}\<[4]%
\>[4]{}\mathbf{do}\;\{\mskip1.5mu \Varid{s}_{2}'\leftarrow \Varid{l}_{2}\mathord{.}\Varid{mcreate}\;\Varid{a};\Varid{return}\;\Varid{s}_{2}'\mskip1.5mu\}{}\<[E]%
\\
\>[B]{}\mathrel{=}{}\<[BE]%
\>[7]{}\mbox{\commentbegin  Monad laws  \commentend}{}\<[E]%
\\
\>[B]{}\hsindent{4}{}\<[4]%
\>[4]{}\Varid{l}_{2}\mathord{.}\Varid{mcreate}\;\Varid{a}{}\<[E]%
\ColumnHook
\end{hscode}\resethooks
Similar reasoning suffices to show that \ensuremath{\Varid{fst}} is a base map from \ensuremath{\Varid{r}_{0}}
to \ensuremath{\Varid{r}_{1}} and \ensuremath{\Varid{snd}} is a base map from \ensuremath{\Varid{r}_{0}} to \ensuremath{\Varid{r}_{2}}, so we can conclude
that \ensuremath{\Conid{R}} and \ensuremath{(\Varid{l},\Varid{r})} constitute a bisimulation between \ensuremath{\Varid{sp}_{\mathrm{1}}} and
\ensuremath{\Varid{sp}_{\mathrm{2}}}, that is, \ensuremath{\Varid{sp}_{\mathrm{1}}\equiv_{\mathrm{b}} \Varid{sp}_{\mathrm{2}}}.
\end{proof}

\restatableTheorem{thm:pure-bisim-implies-jr}
\begin{thm:pure-bisim-implies-jr}
Given \ensuremath{\Varid{sp}_{\mathrm{1}}\mathbin{::}\Conid{A}\mathbin{{\reflectbox{$\rightsquigarrow$}}}\Conid{S}_{1}\mathbin{{\rightsquigarrow}}\Conid{B},\Varid{sp}_{\mathrm{2}}\mathbin{::}\Conid{A}\mathbin{{\reflectbox{$\rightsquigarrow$}}}\Conid{S}_{2}\mathbin{{\rightsquigarrow}}\Conid{B}}, if  \ensuremath{\Varid{sp}_{\mathrm{1}}\equiv_{\mathrm{b}} \Varid{sp}_{\mathrm{2}}} then \ensuremath{\Varid{sp}_{\mathrm{1}}\equiv_{\mathrm{s}} \Varid{sp}_{\mathrm{2}}}.  
\end{thm:pure-bisim-implies-jr}
\begin{proof}
  For convenience, we again write \ensuremath{\Varid{sp}_{\mathrm{1}}\mathrel{=}(\Varid{l}_{1},\Varid{r}_{1})} and \ensuremath{\Varid{sp}_{\mathrm{2}}\mathrel{=}(\Varid{l}_{2},\Varid{r}_{2})}.  
We are given \ensuremath{\Conid{R}} and a span \ensuremath{\Varid{sp}_{\mathrm{0}}\mathbin{::}\Conid{A}\mathbin{{\reflectbox{$\rightsquigarrow$}}}\Conid{R}\mathbin{{\rightsquigarrow}}\Conid{B}} constituting a
  bisimulation \ensuremath{\Varid{sp}_{\mathrm{1}}\equiv_{\mathrm{b}} \Varid{sp}_{\mathrm{2}}}.  Let \ensuremath{\Varid{sp}_{\mathrm{0}}\mathrel{=}(\Varid{l}_{0},\Varid{r}_{0})}.  For later reference, we list
  the properties that must hold by virtue of this bisimulation for any \ensuremath{(\Varid{s}_{1},\Varid{s}_{2})\in\Conid{R}}:
\begin{hscode}\SaveRestoreHook
\column{B}{@{}>{\hspre}l<{\hspost}@{}}%
\column{4}{@{}>{\hspre}l<{\hspost}@{}}%
\column{27}{@{}>{\hspre}l<{\hspost}@{}}%
\column{46}{@{}>{\hspre}c<{\hspost}@{}}%
\column{46E}{@{}l@{}}%
\column{50}{@{}>{\hspre}c<{\hspost}@{}}%
\column{50E}{@{}l@{}}%
\column{56}{@{}>{\hspre}l<{\hspost}@{}}%
\column{59}{@{}>{\hspre}l<{\hspost}@{}}%
\column{62}{@{}>{\hspre}l<{\hspost}@{}}%
\column{80}{@{}>{\hspre}l<{\hspost}@{}}%
\column{E}{@{}>{\hspre}l<{\hspost}@{}}%
\>[4]{}\Varid{l}_{0}\mathord{.}\Varid{get}\;(\Varid{s}_{1},\Varid{s}_{2}){}\<[27]%
\>[27]{}\mathrel{=}\Varid{l}_{1}\mathord{.}\Varid{get}\;\Varid{s}_{1}{}\<[46]%
\>[46]{}\quad{}\<[46E]%
\>[50]{}\quad{}\<[50E]%
\>[62]{}\Varid{l}_{0}\mathord{.}\Varid{get}\;(\Varid{s}_{1},\Varid{s}_{2}){}\<[80]%
\>[80]{}\mathrel{=}\Varid{l}_{2}\mathord{.}\Varid{get}\;\Varid{s}_{2}{}\<[E]%
\\
\>[4]{}\Varid{fst}\;(\Varid{l}_{0}\mathord{.}\Varid{put}\;(\Varid{s}_{1},\Varid{s}_{2})\;\Varid{a}){}\<[27]%
\>[27]{}\mathrel{=}\Varid{l}_{1}\mathord{.}\Varid{put}\;\Varid{s}_{1}\;\Varid{a}{}\<[46]%
\>[46]{}\quad{}\<[46E]%
\>[50]{}\quad{}\<[50E]%
\>[56]{}\Varid{snd}\;(\Varid{l}_{0}\mathord{.}\Varid{put}\;(\Varid{s}_{1},\Varid{s}_{2})\;\Varid{a}){}\<[80]%
\>[80]{}\mathrel{=}\Varid{l}_{2}\mathord{.}\Varid{put}\;\Varid{s}_{2}\;\Varid{a}{}\<[E]%
\\
\>[4]{}\Varid{fst}\;(\Varid{l}_{0}\mathord{.}\Varid{create}\;\Varid{a}){}\<[27]%
\>[27]{}\mathrel{=}\Varid{l}_{1}\mathord{.}\Varid{create}\;\Varid{s}_{1}{}\<[46]%
\>[46]{}\quad{}\<[46E]%
\>[50]{}\quad{}\<[50E]%
\>[59]{}\Varid{snd}\;(\Varid{l}_{0}\mathord{.}\Varid{create}\;\Varid{a}){}\<[80]%
\>[80]{}\mathrel{=}\Varid{l}_{2}\mathord{.}\Varid{create}\;\Varid{a}{}\<[E]%
\\[\blanklineskip]%
\>[4]{}\Varid{r}_{0}\mathord{.}\Varid{get}\;(\Varid{s}_{1},\Varid{s}_{2}){}\<[27]%
\>[27]{}\mathrel{=}\Varid{r}_{1}\mathord{.}\Varid{get}\;\Varid{s}_{1}{}\<[46]%
\>[46]{}\quad{}\<[46E]%
\>[50]{}\quad{}\<[50E]%
\>[62]{}\Varid{r}_{0}\mathord{.}\Varid{get}\;(\Varid{s}_{1},\Varid{s}_{2}){}\<[80]%
\>[80]{}\mathrel{=}\Varid{r}_{2}\mathord{.}\Varid{get}\;\Varid{s}_{2}{}\<[E]%
\\
\>[4]{}\Varid{fst}\;(\Varid{r}_{0}\mathord{.}\Varid{put}\;(\Varid{s}_{1},\Varid{s}_{2})\;\Varid{b}){}\<[27]%
\>[27]{}\mathrel{=}\Varid{r}_{1}\mathord{.}\Varid{put}\;\Varid{s}_{1}\;\Varid{b}{}\<[46]%
\>[46]{}\quad{}\<[46E]%
\>[50]{}\quad{}\<[50E]%
\>[56]{}\Varid{snd}\;(\Varid{r}_{0}\mathord{.}\Varid{put}\;(\Varid{s}_{1},\Varid{s}_{2})\;\Varid{b}){}\<[80]%
\>[80]{}\mathrel{=}\Varid{r}_{2}\mathord{.}\Varid{put}\;\Varid{s}_{2}\;\Varid{b}{}\<[E]%
\\
\>[4]{}\Varid{fst}\;(\Varid{r}_{0}\mathord{.}\Varid{create}\;\Varid{b}){}\<[27]%
\>[27]{}\mathrel{=}\Varid{r}_{1}\mathord{.}\Varid{create}\;\Varid{s}_{1}{}\<[46]%
\>[46]{}\quad{}\<[46E]%
\>[50]{}\quad{}\<[50E]%
\>[59]{}\Varid{snd}\;(\Varid{r}_{0}\mathord{.}\Varid{create}\;\Varid{b}){}\<[80]%
\>[80]{}\mathrel{=}\Varid{r}_{2}\mathord{.}\Varid{create}\;\Varid{b}{}\<[E]%
\ColumnHook
\end{hscode}\resethooks
In addition, it follows that:
\begin{hscode}\SaveRestoreHook
\column{B}{@{}>{\hspre}l<{\hspost}@{}}%
\column{E}{@{}>{\hspre}l<{\hspost}@{}}%
\>[B]{}\Varid{l}_{0}\mathord{.}\Varid{put}\;(\Varid{s}_{1},\Varid{s}_{2})\;\Varid{a}\mathrel{=}(\Varid{l}_{1}\mathord{.}\Varid{put}\;\Varid{s}_{1}\;\Varid{a},\Varid{l}_{2}\mathord{.}\Varid{put}\;\Varid{s}_{2}\;\Varid{a})\in\Conid{R}{}\<[E]%
\\
\>[B]{}\Varid{r}_{0}\mathord{.}\Varid{put}\;(\Varid{s}_{1},\Varid{s}_{2})\;\Varid{b}\mathrel{=}(\Varid{r}_{1}\mathord{.}\Varid{put}\;\Varid{s}_{1}\;\Varid{b},\Varid{r}_{2}\mathord{.}\Varid{put}\;\Varid{s}_{2}\;\Varid{b})\in\Conid{R}{}\<[E]%
\\
\>[B]{}\Varid{l}_{0}\mathord{.}\Varid{create}\;\Varid{a}\mathrel{=}(\Varid{l}_{1}\mathord{.}\Varid{create}\;\Varid{a},\Varid{l}_{2}\mathord{.}\Varid{create}\;\Varid{a})\in\Conid{R}{}\<[E]%
\\
\>[B]{}\Varid{r}_{0}\mathord{.}\Varid{create}\;\Varid{b}\mathrel{=}(\Varid{r}_{1}\mathord{.}\Varid{create}\;\Varid{b},\Varid{r}_{2}\mathord{.}\Varid{create}\;\Varid{b})\in\Conid{R}{}\<[E]%
\ColumnHook
\end{hscode}\resethooks
which also implies the following identities, which we call \emph{twists}:
\begin{hscode}\SaveRestoreHook
\column{B}{@{}>{\hspre}l<{\hspost}@{}}%
\column{E}{@{}>{\hspre}l<{\hspost}@{}}%
\>[B]{}\Varid{r}_{1}\mathord{.}\Varid{get}\;(\Varid{l}_{1}\mathord{.}\Varid{put}\;\Varid{s}_{1}\;\Varid{a})\mathrel{=}\Varid{r}_{0}\mathord{.}\Varid{get}\;(\Varid{l}_{1}\mathord{.}\Varid{put}\;\Varid{s}_{1}\;\Varid{a},\Varid{l}_{2}\mathord{.}\Varid{put}\;\Varid{s}_{2}\;\Varid{a})\mathrel{=}\Varid{r}_{2}\mathord{.}\Varid{get}\;(\Varid{l}_{2}\mathord{.}\Varid{put}\;\Varid{s}_{2}\;\Varid{a}){}\<[E]%
\\
\>[B]{}\Varid{l}_{1}\mathord{.}\Varid{get}\;(\Varid{r}_{1}\mathord{.}\Varid{put}\;\Varid{s}_{1}\;\Varid{b})\mathrel{=}\Varid{l}_{0}\mathord{.}\Varid{get}\;(\Varid{r}_{1}\mathord{.}\Varid{put}\;\Varid{s}_{1}\;\Varid{b},\Varid{r}_{2}\mathord{.}\Varid{put}\;\Varid{s}_{2}\;\Varid{b})\mathrel{=}\Varid{l}_{2}\mathord{.}\Varid{get}\;(\Varid{r}_{2}\mathord{.}\Varid{put}\;\Varid{s}_{2}\;\Varid{b}){}\<[E]%
\\
\>[B]{}\Varid{r}_{1}\mathord{.}\Varid{get}\;(\Varid{l}_{1}\mathord{.}\Varid{create}\;\Varid{a})\mathrel{=}\Varid{r}_{0}\mathord{.}\Varid{get}\;(\Varid{l}_{1}\mathord{.}\Varid{create}\;\Varid{a},\Varid{l}_{2}\mathord{.}\Varid{create}\;\Varid{a})\mathrel{=}\Varid{r}_{2}\mathord{.}\Varid{get}\;(\Varid{l}_{2}\mathord{.}\Varid{create}\;\Varid{a}){}\<[E]%
\\
\>[B]{}\Varid{l}_{1}\mathord{.}\Varid{get}\;(\Varid{r}_{1}\mathord{.}\Varid{create}\;\Varid{b})\mathrel{=}\Varid{l}_{0}\mathord{.}\Varid{get}\;(\Varid{r}_{1}\mathord{.}\Varid{create}\;\Varid{b},\Varid{r}_{2}\mathord{.}\Varid{create}\;\Varid{b})\mathrel{=}\Varid{l}_{2}\mathord{.}\Varid{get}\;(\Varid{r}_{2}\mathord{.}\Varid{create}\;\Varid{b}){}\<[E]%
\ColumnHook
\end{hscode}\resethooks
 
It suffices to
  construct a span \ensuremath{\Varid{sp}\mathrel{=}(\Varid{l},\Varid{r})\mathbin{::}\Conid{S}_{1}\mathbin{{\reflectbox{$\rightsquigarrow$}}}\Conid{R}\mathbin{{\rightsquigarrow}}\Conid{S}_{2}} satisfying \ensuremath{\Varid{l}\mathbin{;}\Varid{l}_{1}\mathrel{=}\Varid{r}\mathbin{;}\Varid{l}_{2}} and \ensuremath{\Varid{l}\mathbin{;}\Varid{r}_{1}\mathrel{=}\Varid{r}\mathbin{;}\Varid{r}_{2}}.
Define \ensuremath{\Varid{l}} and \ensuremath{\Varid{r}} as follows:
\begin{hscode}\SaveRestoreHook
\column{B}{@{}>{\hspre}l<{\hspost}@{}}%
\column{3}{@{}>{\hspre}l<{\hspost}@{}}%
\column{23}{@{}>{\hspre}l<{\hspost}@{}}%
\column{E}{@{}>{\hspre}l<{\hspost}@{}}%
\>[3]{}\Varid{l}\mathord{.}\Varid{get}{}\<[23]%
\>[23]{}\mathrel{=}\Varid{fst}{}\<[E]%
\\
\>[3]{}\Varid{l}\mathord{.}\Varid{put}\;(\Varid{s}_{1},\Varid{s}_{2})\;\Varid{s}_{1}'{}\<[23]%
\>[23]{}\mathrel{=}\Varid{l}_{0}\mathord{.}\Varid{put}\;(\Varid{s}_{1},\Varid{s}_{2})\;(\Varid{l}_{1}\mathord{.}\Varid{get}\;\Varid{s}_{1}'){}\<[E]%
\\
\>[3]{}\Varid{l}\mathord{.}\Varid{create}\;\Varid{s}_{1}{}\<[23]%
\>[23]{}\mathrel{=}\Varid{l}_{0}\mathord{.}\Varid{create}\;(\Varid{l}_{1}\mathord{.}\Varid{get}\;\Varid{s}_{1}){}\<[E]%
\\[\blanklineskip]%
\>[3]{}\Varid{r}\mathord{.}\Varid{get}{}\<[23]%
\>[23]{}\mathrel{=}\Varid{snd}{}\<[E]%
\\
\>[3]{}\Varid{r}\mathord{.}\Varid{put}\;(\Varid{s}_{1},\Varid{s}_{2})\;\Varid{s}_{2}'{}\<[23]%
\>[23]{}\mathrel{=}\Varid{l}_{0}\mathord{.}\Varid{put}\;(\Varid{s}_{1},\Varid{s}_{2})\;(\Varid{l}_{2}\mathord{.}\Varid{get}\;\Varid{s}_{2}'){}\<[E]%
\\
\>[3]{}\Varid{r}\mathord{.}\Varid{create}\;\Varid{s}_{2}{}\<[23]%
\>[23]{}\mathrel{=}\Varid{l}_{0}\mathord{.}\Varid{create}\;(\Varid{l}_{2}\mathord{.}\Varid{get}\;\Varid{s}_{2}){}\<[E]%
\ColumnHook
\end{hscode}\resethooks

Notice that by construction \ensuremath{\Varid{l}\mathbin{::}\Conid{R}\mathbin{\leadsto}\Conid{S}_{1}} and \ensuremath{\Varid{r}\mathbin{::}\Conid{R}\mathbin{\leadsto}\Conid{S}_{2}}, that
is, since we have used \ensuremath{\Varid{l}_{0}} and \ensuremath{\Varid{r}_{0}} to define \ensuremath{\Varid{l}} and \ensuremath{\Varid{r}}, we do not
need to do any more work to check that the pairs produced by \ensuremath{\Varid{create}}
and \ensuremath{\Varid{put}} remain in \ensuremath{\Conid{R}}.  Notice also that \ensuremath{\Varid{l}} and \ensuremath{\Varid{r}} only use the lenses \ensuremath{\Varid{l}_{1}} and \ensuremath{\Varid{l}_{2}}, not \ensuremath{\Varid{r}_{1}} and \ensuremath{\Varid{r}_{2}}; we will show nevertheless that they satisfy the required properties.

First, to show that \ensuremath{\Varid{l}\mathbin{;}\Varid{l}_{1}\mathrel{=}\Varid{r}\mathbin{;}\Varid{l}_{2}}, we proceed as follows for each
  operation.
For \ensuremath{\Varid{get}}:
\begin{hscode}\SaveRestoreHook
\column{B}{@{}>{\hspre}c<{\hspost}@{}}%
\column{BE}{@{}l@{}}%
\column{4}{@{}>{\hspre}l<{\hspost}@{}}%
\column{6}{@{}>{\hspre}l<{\hspost}@{}}%
\column{E}{@{}>{\hspre}l<{\hspost}@{}}%
\>[4]{}(\Varid{l}\mathbin{;}\Varid{l}_{1})\mathord{.}\Varid{get}\;(\Varid{s}_{1},\Varid{s}_{2}){}\<[E]%
\\
\>[B]{}\mathrel{=}{}\<[BE]%
\>[6]{}\mbox{\commentbegin  definition  \commentend}{}\<[E]%
\\
\>[B]{}\hsindent{4}{}\<[4]%
\>[4]{}\Varid{l}_{1}\mathord{.}\Varid{get}\;(\Varid{l}\mathord{.}\Varid{get}\;(\Varid{s}_{1},\Varid{s}_{2})){}\<[E]%
\\
\>[B]{}\mathrel{=}{}\<[BE]%
\>[6]{}\mbox{\commentbegin  definition of \ensuremath{\Varid{l}\mathord{.}\Varid{get}\mathrel{=}\Varid{fst}}, \ensuremath{\Varid{fst}} commutes with \ensuremath{\Varid{get}}  \commentend}{}\<[E]%
\\
\>[B]{}\hsindent{4}{}\<[4]%
\>[4]{}\Varid{l}_{0}\mathord{.}\Varid{get}\;(\Varid{s}_{1},\Varid{s}_{2}){}\<[E]%
\\
\>[B]{}\mathrel{=}{}\<[BE]%
\>[6]{}\mbox{\commentbegin  reverse reasoning  \commentend}{}\<[E]%
\\
\>[B]{}\hsindent{4}{}\<[4]%
\>[4]{}(\Varid{r}\mathbin{;}\Varid{l}_{2})\mathord{.}\Varid{get}\;(\Varid{s}_{1},\Varid{s}_{2}){}\<[E]%
\ColumnHook
\end{hscode}\resethooks
For \ensuremath{\Varid{put}}, we have:
\begin{hscode}\SaveRestoreHook
\column{B}{@{}>{\hspre}c<{\hspost}@{}}%
\column{BE}{@{}l@{}}%
\column{4}{@{}>{\hspre}l<{\hspost}@{}}%
\column{6}{@{}>{\hspre}l<{\hspost}@{}}%
\column{E}{@{}>{\hspre}l<{\hspost}@{}}%
\>[4]{}(\Varid{l}\mathbin{;}\Varid{l}_{1})\mathord{.}\Varid{put}\;(\Varid{s}_{1},\Varid{s}_{2})\;\Varid{a}{}\<[E]%
\\
\>[B]{}\mathrel{=}{}\<[BE]%
\>[6]{}\mbox{\commentbegin  Definition  \commentend}{}\<[E]%
\\
\>[B]{}\hsindent{4}{}\<[4]%
\>[4]{}\Varid{l}\mathord{.}\Varid{put}\;(\Varid{s}_{1},\Varid{s}_{2})\;(\Varid{l}_{1}\mathord{.}\Varid{put}\;\Varid{s}_{1}\;\Varid{a}){}\<[E]%
\\
\>[B]{}\mathrel{=}{}\<[BE]%
\>[6]{}\mbox{\commentbegin  Definition  \commentend}{}\<[E]%
\\
\>[B]{}\hsindent{4}{}\<[4]%
\>[4]{}\Varid{l}_{0}\mathord{.}\Varid{put}\;(\Varid{s}_{1},\Varid{s}_{2})\;(\Varid{l}_{1}\mathord{.}\Varid{get}\;(\Varid{l}_{1}\mathord{.}\Varid{put}\;\Varid{s}_{1}\;\Varid{a})){}\<[E]%
\\
\>[B]{}\mathrel{=}{}\<[BE]%
\>[6]{}\mbox{\commentbegin  \ensuremath{\mathsf{(PutGet)}} for \ensuremath{\Varid{l}_{1}}  \commentend}{}\<[E]%
\\
\>[B]{}\hsindent{4}{}\<[4]%
\>[4]{}\Varid{l}_{0}\mathord{.}\Varid{put}\;(\Varid{s}_{1},\Varid{s}_{2})\;\Varid{a}{}\<[E]%
\\
\>[B]{}\mathrel{=}{}\<[BE]%
\>[6]{}\mbox{\commentbegin  \ensuremath{\mathsf{(PutGet)}} for \ensuremath{\Varid{l}_{2}}  \commentend}{}\<[E]%
\\
\>[B]{}\hsindent{4}{}\<[4]%
\>[4]{}\Varid{l}_{0}\mathord{.}\Varid{put}\;(\Varid{s}_{1},\Varid{s}_{2})\;(\Varid{l}_{2}\mathord{.}\Varid{get}\;(\Varid{l}_{2}\mathord{.}\Varid{put}\;\Varid{s}_{2}\;\Varid{a})){}\<[E]%
\\
\>[B]{}\mathrel{=}{}\<[BE]%
\>[6]{}\mbox{\commentbegin  Definition \commentend}{}\<[E]%
\\
\>[B]{}\hsindent{4}{}\<[4]%
\>[4]{}\Varid{r}\mathord{.}\Varid{put}\;(\Varid{s}_{1},\Varid{s}_{2})\;(\Varid{l}_{2}\mathord{.}\Varid{put}\;\Varid{s}_{2}\;\Varid{a}){}\<[E]%
\\
\>[B]{}\mathrel{=}{}\<[BE]%
\>[6]{}\mbox{\commentbegin  Definition \commentend}{}\<[E]%
\\
\>[B]{}\hsindent{4}{}\<[4]%
\>[4]{}(\Varid{r}\mathbin{;}\Varid{l}_{2})\mathord{.}\Varid{put}\;(\Varid{s}_{1},\Varid{s}_{2})\;\Varid{a}{}\<[E]%
\ColumnHook
\end{hscode}\resethooks
Finally, for \ensuremath{\Varid{create}} we have:
\begin{hscode}\SaveRestoreHook
\column{B}{@{}>{\hspre}c<{\hspost}@{}}%
\column{BE}{@{}l@{}}%
\column{4}{@{}>{\hspre}l<{\hspost}@{}}%
\column{6}{@{}>{\hspre}l<{\hspost}@{}}%
\column{E}{@{}>{\hspre}l<{\hspost}@{}}%
\>[4]{}(\Varid{l}\mathbin{;}\Varid{l}_{1})\mathord{.}\Varid{create}\;\Varid{a}{}\<[E]%
\\
\>[B]{}\mathrel{=}{}\<[BE]%
\>[6]{}\mbox{\commentbegin  Definition  \commentend}{}\<[E]%
\\
\>[B]{}\hsindent{4}{}\<[4]%
\>[4]{}\Varid{l}\mathord{.}\Varid{create}\;(\Varid{l}_{1}\mathord{.}\Varid{create}\;\Varid{a}){}\<[E]%
\\
\>[B]{}\mathrel{=}{}\<[BE]%
\>[6]{}\mbox{\commentbegin  Definition  \commentend}{}\<[E]%
\\
\>[B]{}\hsindent{4}{}\<[4]%
\>[4]{}\Varid{l}_{0}\mathord{.}\Varid{create}\;(\Varid{l}_{1}\mathord{.}\Varid{get}\;(\Varid{l}_{1}\mathord{.}\Varid{create}\;\Varid{a})){}\<[E]%
\\
\>[B]{}\mathrel{=}{}\<[BE]%
\>[6]{}\mbox{\commentbegin  \ensuremath{\mathsf{(CreateGet)}} for \ensuremath{\Varid{l}_{1}}  \commentend}{}\<[E]%
\\
\>[B]{}\hsindent{4}{}\<[4]%
\>[4]{}\Varid{l}_{0}\mathord{.}\Varid{create}\;\Varid{a}{}\<[E]%
\\
\>[B]{}\mathrel{=}{}\<[BE]%
\>[6]{}\mbox{\commentbegin  \ensuremath{\mathsf{(CreateGet)}} for \ensuremath{\Varid{l}_{2}}  \commentend}{}\<[E]%
\\
\>[B]{}\hsindent{4}{}\<[4]%
\>[4]{}\Varid{l}_{0}\mathord{.}\Varid{create}\;(\Varid{l}_{2}\mathord{.}\Varid{get}\;(\Varid{l}_{2}\mathord{.}\Varid{create}\;\Varid{a})){}\<[E]%
\\
\>[B]{}\mathrel{=}{}\<[BE]%
\>[6]{}\mbox{\commentbegin  Definition \commentend}{}\<[E]%
\\
\>[B]{}\hsindent{4}{}\<[4]%
\>[4]{}\Varid{r}\mathord{.}\Varid{create}\;(\Varid{l}_{2}\mathord{.}\Varid{create}\;\Varid{a}){}\<[E]%
\\
\>[B]{}\mathrel{=}{}\<[BE]%
\>[6]{}\mbox{\commentbegin  Definition \commentend}{}\<[E]%
\\
\>[B]{}\hsindent{4}{}\<[4]%
\>[4]{}(\Varid{r}\mathbin{;}\Varid{l}_{2})\mathord{.}\Varid{create}\;\Varid{a}{}\<[E]%
\ColumnHook
\end{hscode}\resethooks

Next, we show that \ensuremath{\Varid{l}\mathbin{;}\Varid{r}_{1}\mathrel{=}\Varid{r}\mathbin{;}\Varid{r}_{2}}.  
For \ensuremath{\Varid{get}}:
\begin{hscode}\SaveRestoreHook
\column{B}{@{}>{\hspre}c<{\hspost}@{}}%
\column{BE}{@{}l@{}}%
\column{4}{@{}>{\hspre}l<{\hspost}@{}}%
\column{6}{@{}>{\hspre}l<{\hspost}@{}}%
\column{E}{@{}>{\hspre}l<{\hspost}@{}}%
\>[4]{}(\Varid{l}\mathbin{;}\Varid{r}_{1})\mathord{.}\Varid{get}\;(\Varid{s}_{1},\Varid{s}_{2}){}\<[E]%
\\
\>[B]{}\mathrel{=}{}\<[BE]%
\>[6]{}\mbox{\commentbegin  Definition  \commentend}{}\<[E]%
\\
\>[B]{}\hsindent{4}{}\<[4]%
\>[4]{}\Varid{r}_{1}\mathord{.}\Varid{get}\;(\Varid{l}\mathord{.}\Varid{get}\;(\Varid{s}_{1},\Varid{s}_{2})){}\<[E]%
\\
\>[B]{}\mathrel{=}{}\<[BE]%
\>[6]{}\mbox{\commentbegin  definition of \ensuremath{\Varid{l}\mathord{.}\Varid{get}\mathrel{=}\Varid{fst}}, \ensuremath{\Varid{fst}} commutes with \ensuremath{\Varid{r}_{1}\mathord{.}\Varid{get}}  \commentend}{}\<[E]%
\\
\>[B]{}\hsindent{4}{}\<[4]%
\>[4]{}\Varid{r}_{0}\mathord{.}\Varid{get}\;(\Varid{s}_{1},\Varid{s}_{2}){}\<[E]%
\\
\>[B]{}\mathrel{=}{}\<[BE]%
\>[6]{}\mbox{\commentbegin  reverse above reasoning  \commentend}{}\<[E]%
\\
\>[B]{}\hsindent{4}{}\<[4]%
\>[4]{}(\Varid{r}\mathbin{;}\Varid{r}_{2})\mathord{.}\Varid{get}\;(\Varid{s}_{1},\Varid{s}_{2}){}\<[E]%
\ColumnHook
\end{hscode}\resethooks
For \ensuremath{\Varid{put}}, we have:
\begin{hscode}\SaveRestoreHook
\column{B}{@{}>{\hspre}c<{\hspost}@{}}%
\column{BE}{@{}l@{}}%
\column{4}{@{}>{\hspre}l<{\hspost}@{}}%
\column{6}{@{}>{\hspre}l<{\hspost}@{}}%
\column{E}{@{}>{\hspre}l<{\hspost}@{}}%
\>[4]{}(\Varid{l}\mathbin{;}\Varid{r}_{1})\mathord{.}\Varid{put}\;(\Varid{s}_{1},\Varid{s}_{2})\;\Varid{b}{}\<[E]%
\\
\>[B]{}\mathrel{=}{}\<[BE]%
\>[6]{}\mbox{\commentbegin  Definition  \commentend}{}\<[E]%
\\
\>[B]{}\hsindent{4}{}\<[4]%
\>[4]{}\Varid{l}\mathord{.}\Varid{put}\;(\Varid{s}_{1},\Varid{s}_{2})\;(\Varid{r}_{1}\mathord{.}\Varid{put}\;\Varid{s}_{1}\;\Varid{b}){}\<[E]%
\\
\>[B]{}\mathrel{=}{}\<[BE]%
\>[6]{}\mbox{\commentbegin  Definition  \commentend}{}\<[E]%
\\
\>[B]{}\hsindent{4}{}\<[4]%
\>[4]{}\Varid{l}_{0}\mathord{.}\Varid{put}\;(\Varid{s}_{1},\Varid{s}_{2})\;(\Varid{l}_{1}\mathord{.}\Varid{get}\;(\Varid{r}_{1}\mathord{.}\Varid{put}\;\Varid{s}_{1}\;\Varid{b})){}\<[E]%
\\
\>[B]{}\mathrel{=}{}\<[BE]%
\>[6]{}\mbox{\commentbegin  Twist equation  \commentend}{}\<[E]%
\\
\>[B]{}\hsindent{4}{}\<[4]%
\>[4]{}\Varid{l}_{0}\mathord{.}\Varid{put}\;(\Varid{s}_{1},\Varid{s}_{2})\;(\Varid{l}_{2}\mathord{.}\Varid{get}\;(\Varid{r}_{2}\mathord{.}\Varid{put}\;\Varid{s}_{2}\;\Varid{b})){}\<[E]%
\\
\>[B]{}\mathrel{=}{}\<[BE]%
\>[6]{}\mbox{\commentbegin  Definition  \commentend}{}\<[E]%
\\
\>[B]{}\hsindent{4}{}\<[4]%
\>[4]{}\Varid{r}\mathord{.}\Varid{put}\;(\Varid{s}_{1},\Varid{s}_{2})\;(\Varid{r}_{2}\mathord{.}\Varid{put}\;\Varid{s}_{2}\;\Varid{b}){}\<[E]%
\\
\>[B]{}\mathrel{=}{}\<[BE]%
\>[6]{}\mbox{\commentbegin  Definition  \commentend}{}\<[E]%
\\
\>[B]{}\hsindent{4}{}\<[4]%
\>[4]{}(\Varid{r}\mathbin{;}\Varid{r}_{2})\mathord{.}\Varid{put}\;(\Varid{s}_{1},\Varid{s}_{2})\;\Varid{b}{}\<[E]%
\ColumnHook
\end{hscode}\resethooks
Finally, for \ensuremath{\Varid{create}} we have:
\begin{hscode}\SaveRestoreHook
\column{B}{@{}>{\hspre}c<{\hspost}@{}}%
\column{BE}{@{}l@{}}%
\column{4}{@{}>{\hspre}l<{\hspost}@{}}%
\column{6}{@{}>{\hspre}l<{\hspost}@{}}%
\column{E}{@{}>{\hspre}l<{\hspost}@{}}%
\>[4]{}(\Varid{l}\mathbin{;}\Varid{r}_{1})\mathord{.}\Varid{create}\;\Varid{b}{}\<[E]%
\\
\>[B]{}\mathrel{=}{}\<[BE]%
\>[6]{}\mbox{\commentbegin  Definition  \commentend}{}\<[E]%
\\
\>[B]{}\hsindent{4}{}\<[4]%
\>[4]{}\Varid{l}\mathord{.}\Varid{create}\;(\Varid{r}_{1}\mathord{.}\Varid{create}\;\Varid{b}){}\<[E]%
\\
\>[B]{}\mathrel{=}{}\<[BE]%
\>[6]{}\mbox{\commentbegin  Definition  \commentend}{}\<[E]%
\\
\>[B]{}\hsindent{4}{}\<[4]%
\>[4]{}\Varid{l}_{0}\mathord{.}\Varid{create}\;(\Varid{l}_{1}\mathord{.}\Varid{get}\;(\Varid{r}_{1}\mathord{.}\Varid{create}\;\Varid{b})){}\<[E]%
\\
\>[B]{}\mathrel{=}{}\<[BE]%
\>[6]{}\mbox{\commentbegin  Twist equation  \commentend}{}\<[E]%
\\
\>[B]{}\hsindent{4}{}\<[4]%
\>[4]{}\Varid{l}_{0}\mathord{.}\Varid{create}\;(\Varid{l}_{2}\mathord{.}\Varid{get}\;(\Varid{r}_{2}\mathord{.}\Varid{create}\;\Varid{b})){}\<[E]%
\\
\>[B]{}\mathrel{=}{}\<[BE]%
\>[6]{}\mbox{\commentbegin  Definition  \commentend}{}\<[E]%
\\
\>[B]{}\hsindent{4}{}\<[4]%
\>[4]{}\Varid{r}\mathord{.}\Varid{create}\;(\Varid{r}_{2}\mathord{.}\Varid{create}\;\Varid{b}){}\<[E]%
\\
\>[B]{}\mathrel{=}{}\<[BE]%
\>[6]{}\mbox{\commentbegin  Definition  \commentend}{}\<[E]%
\\
\>[B]{}\hsindent{4}{}\<[4]%
\>[4]{}(\Varid{r}\mathbin{;}\Varid{r}_{2})\mathord{.}\Varid{create}\;\Varid{b}{}\<[E]%
\ColumnHook
\end{hscode}\resethooks

We must also show that \ensuremath{\Varid{l}} and \ensuremath{\Varid{r}} are well-behaved full lenses.  
To show that \ensuremath{\Varid{l}} is well-behaved, we proceed as follows.
For \ensuremath{\mathsf{(GetPut)}}:
\begin{hscode}\SaveRestoreHook
\column{B}{@{}>{\hspre}c<{\hspost}@{}}%
\column{BE}{@{}l@{}}%
\column{4}{@{}>{\hspre}l<{\hspost}@{}}%
\column{5}{@{}>{\hspre}l<{\hspost}@{}}%
\column{6}{@{}>{\hspre}l<{\hspost}@{}}%
\column{E}{@{}>{\hspre}l<{\hspost}@{}}%
\>[4]{}\Varid{l}\mathord{.}\Varid{get}\;(\Varid{l}\mathord{.}\Varid{put}\;(\Varid{s}_{1},\Varid{s}_{2})\;\Varid{s}_{1}'){}\<[E]%
\\
\>[B]{}\mathrel{=}{}\<[BE]%
\>[6]{}\mbox{\commentbegin  Definition  \commentend}{}\<[E]%
\\
\>[B]{}\hsindent{4}{}\<[4]%
\>[4]{}\Varid{fst}\;(\Varid{l}_{0}\mathord{.}\Varid{put}\;(\Varid{s}_{1},\Varid{s}_{2})\;(\Varid{l}_{1}\mathord{.}\Varid{get}\;\Varid{s}_{1}')){}\<[E]%
\\
\>[B]{}\mathrel{=}{}\<[BE]%
\>[6]{}\mbox{\commentbegin  \ensuremath{\Varid{fst}} commutes with \ensuremath{\Varid{put}}  \commentend}{}\<[E]%
\\
\>[B]{}\hsindent{4}{}\<[4]%
\>[4]{}\Varid{l}_{1}\mathord{.}\Varid{put}\;\Varid{s}_{1}\;(\Varid{l}_{1}\mathord{.}\Varid{get}\;\Varid{s}_{1}')){}\<[E]%
\\
\>[B]{}\mathrel{=}{}\<[BE]%
\>[6]{}\mbox{\commentbegin  \ensuremath{\mathsf{(GetPut)}} for \ensuremath{\Varid{l}_{1}}  \commentend}{}\<[E]%
\\
\>[B]{}\hsindent{5}{}\<[5]%
\>[5]{}\Varid{s}_{1}'{}\<[E]%
\ColumnHook
\end{hscode}\resethooks
For \ensuremath{\mathsf{(PutGet)}}:
\begin{hscode}\SaveRestoreHook
\column{B}{@{}>{\hspre}c<{\hspost}@{}}%
\column{BE}{@{}l@{}}%
\column{4}{@{}>{\hspre}l<{\hspost}@{}}%
\column{6}{@{}>{\hspre}l<{\hspost}@{}}%
\column{E}{@{}>{\hspre}l<{\hspost}@{}}%
\>[4]{}\Varid{l}\mathord{.}\Varid{put}\;(\Varid{s}_{1},\Varid{s}_{2})\;(\Varid{l}\mathord{.}\Varid{get}\;(\Varid{s}_{1},\Varid{s}_{2})){}\<[E]%
\\
\>[B]{}\mathrel{=}{}\<[BE]%
\>[6]{}\mbox{\commentbegin  Definition  \commentend}{}\<[E]%
\\
\>[B]{}\hsindent{4}{}\<[4]%
\>[4]{}\Varid{l}_{0}\mathord{.}\Varid{put}\;(\Varid{s}_{1},\Varid{s}_{2})\;(\Varid{l}_{1}\mathord{.}\Varid{get}\;\Varid{s}_{1}){}\<[E]%
\\
\>[B]{}\mathrel{=}{}\<[BE]%
\>[6]{}\mbox{\commentbegin  Eta-expansion for pairs  \commentend}{}\<[E]%
\\
\>[B]{}\hsindent{4}{}\<[4]%
\>[4]{}(\Varid{fst}\;(\Varid{l}_{0}\mathord{.}\Varid{put}\;(\Varid{s}_{1},\Varid{s}_{2})\;(\Varid{l}_{1}\mathord{.}\Varid{get}\;\Varid{s}_{1})),\Varid{snd}\;(\Varid{l}_{0}\mathord{.}\Varid{put}\;(\Varid{s}_{1},\Varid{s}_{2})\;(\Varid{l}_{1}\mathord{.}\Varid{get}\;\Varid{s}_{1}))){}\<[E]%
\\
\>[B]{}\mathrel{=}{}\<[BE]%
\>[6]{}\mbox{\commentbegin  \ensuremath{\Varid{fst}}, \ensuremath{\Varid{snd}} commutes with \ensuremath{\Varid{put}}  \commentend}{}\<[E]%
\\
\>[B]{}\hsindent{4}{}\<[4]%
\>[4]{}(\Varid{l}_{1}\mathord{.}\Varid{put}\;\Varid{s}_{1}\;(\Varid{l}_{1}\mathord{.}\Varid{get}\;\Varid{s}_{1}),\Varid{l}_{2}\mathord{.}\Varid{put}\;\Varid{s}_{2}\;(\Varid{l}_{1}\mathord{.}\Varid{get}\;\Varid{s}_{1})){}\<[E]%
\\
\>[B]{}\mathrel{=}{}\<[BE]%
\>[6]{}\mbox{\commentbegin  \ensuremath{\Varid{l}_{1}\mathord{.}\Varid{get}\;\Varid{s}_{1}\mathrel{=}\Varid{l}_{2}\mathord{.}\Varid{get}\;\Varid{s}_{2}}  \commentend}{}\<[E]%
\\
\>[B]{}\hsindent{4}{}\<[4]%
\>[4]{}(\Varid{l}_{1}\mathord{.}\Varid{put}\;\Varid{s}_{1}\;(\Varid{l}_{1}\mathord{.}\Varid{get}\;\Varid{s}_{1}),\Varid{l}_{2}\mathord{.}\Varid{put}\;\Varid{s}_{2}\;(\Varid{l}_{2}\mathord{.}\Varid{get}\;\Varid{s}_{2})){}\<[E]%
\\
\>[B]{}\mathrel{=}{}\<[BE]%
\>[6]{}\mbox{\commentbegin  \ensuremath{\mathsf{(PutGet)}} for \ensuremath{\Varid{l}_{1},\Varid{l}_{2}}  \commentend}{}\<[E]%
\\
\>[B]{}\hsindent{4}{}\<[4]%
\>[4]{}(\Varid{s}_{1},\Varid{s}_{2}){}\<[E]%
\ColumnHook
\end{hscode}\resethooks
For \ensuremath{\mathsf{(CreateGet)}}:
\begin{hscode}\SaveRestoreHook
\column{B}{@{}>{\hspre}c<{\hspost}@{}}%
\column{BE}{@{}l@{}}%
\column{4}{@{}>{\hspre}l<{\hspost}@{}}%
\column{6}{@{}>{\hspre}l<{\hspost}@{}}%
\column{7}{@{}>{\hspre}l<{\hspost}@{}}%
\column{E}{@{}>{\hspre}l<{\hspost}@{}}%
\>[4]{}\Varid{l}\mathord{.}\Varid{create}\;(\Varid{l}\mathord{.}\Varid{get}\;(\Varid{s}_{1},\Varid{s}_{2})){}\<[E]%
\\
\>[B]{}\mathrel{=}{}\<[BE]%
\>[7]{}\mbox{\commentbegin  Definition  \commentend}{}\<[E]%
\\
\>[B]{}\hsindent{4}{}\<[4]%
\>[4]{}\Varid{l}_{0}\mathord{.}\Varid{create}\;(\Varid{l}_{1}\mathord{.}\Varid{get}\;\Varid{s}_{1}){}\<[E]%
\\
\>[B]{}\mathrel{=}{}\<[BE]%
\>[7]{}\mbox{\commentbegin  Eta-expansion for pairs  \commentend}{}\<[E]%
\\
\>[B]{}\hsindent{4}{}\<[4]%
\>[4]{}(\Varid{fst}\;(\Varid{l}_{0}\mathord{.}\Varid{create}\;(\Varid{l}_{1}\mathord{.}\Varid{get}\;\Varid{s}_{1})),\Varid{snd}\;(\Varid{l}_{0}\mathord{.}\Varid{create}\;(\Varid{l}_{1}\mathord{.}\Varid{get}\;\Varid{s}_{1}))){}\<[E]%
\\
\>[B]{}\mathrel{=}{}\<[BE]%
\>[7]{}\mbox{\commentbegin  \ensuremath{\Varid{fst}}, \ensuremath{\Varid{snd}} commutes with \ensuremath{\Varid{put}}  \commentend}{}\<[E]%
\\
\>[B]{}\hsindent{4}{}\<[4]%
\>[4]{}(\Varid{l}_{1}\mathord{.}\Varid{create}\;(\Varid{l}_{1}\mathord{.}\Varid{get}\;\Varid{s}_{1}),\Varid{l}_{1}\mathord{.}\Varid{create}\;(\Varid{l}_{1}\mathord{.}\Varid{get}\;\Varid{s}_{1})){}\<[E]%
\\
\>[B]{}\mathrel{=}{}\<[BE]%
\>[6]{}\mbox{\commentbegin  \ensuremath{\Varid{l}_{1}\mathord{.}\Varid{get}\;\Varid{s}_{1}\mathrel{=}\Varid{l}_{2}\mathord{.}\Varid{get}\;\Varid{s}_{2}}  \commentend}{}\<[E]%
\\
\>[B]{}\hsindent{4}{}\<[4]%
\>[4]{}(\Varid{l}_{1}\mathord{.}\Varid{create}\;(\Varid{l}_{1}\mathord{.}\Varid{get}\;\Varid{s}_{1}),\Varid{l}_{1}\mathord{.}\Varid{create}\;(\Varid{l}_{2}\mathord{.}\Varid{get}\;\Varid{s}_{2})){}\<[E]%
\\
\>[B]{}\mathrel{=}{}\<[BE]%
\>[6]{}\mbox{\commentbegin  \ensuremath{\mathsf{(CreateGet)}}  \commentend}{}\<[E]%
\\
\>[B]{}\hsindent{4}{}\<[4]%
\>[4]{}(\Varid{s}_{1},\Varid{s}_{2}){}\<[E]%
\ColumnHook
\end{hscode}\resethooks
Finally, notice that \ensuremath{\Varid{l}} and \ensuremath{\Varid{r}} are defined symmetrically so
essentially the same reasoning shows \ensuremath{\Varid{r}} is well-behaved.

To conclude, \ensuremath{\Varid{sp}\mathrel{=}(\Varid{l},\Varid{r})} constitutes a span of lenses witnessing that
\ensuremath{\Varid{sp}_{\mathrm{1}}\equiv_{\mathrm{s}} \Varid{sp}_{\mathrm{2}}}.
 \end{proof}

\end{document}